\documentclass[acmsmall,screen]{acmart}

\newcommand{\ARXIVVERSION}{}

\usepackage{cleveref}
\usepackage{amsmath}
\usepackage{mdframed}
\usepackage{stmaryrd}
\def\fcmp{\mathbin{\raise 0.6ex\hbox{\oalign{\hfil$\scriptscriptstyle\mathrm{o}$\hfil\cr\hfil$\scriptscriptstyle\mathrm{9}$\hfil}}}}
\usepackage{subcaption} 
\usepackage{wrapfig}
\usepackage{svg}
\usepackage{tablefootnote}
\usepackage{enumitem}

\setlist[1]{leftmargin=2em}

\usepackage{longtable}

\providecommand{\tightlist}{%
  \setlength{\itemsep}{0pt}\setlength{\parskip}{0pt}}

\makeatletter
\@ifundefined{c@chapter}{\newfloat{codelisting}{h}{lop}}{\newfloat{codelisting}{h}{lop}[chapter]}
\floatname{codelisting}{Listing}

\usepackage{color}
\usepackage{fancyvrb}
\usepackage{fvextra} 

\newcommand{\VERB}{\Verb[commandchars=\\\{\},fontsize=\small]}
\DefineVerbatimEnvironment{Highlighting}{Verbatim}{commandchars=\\\{\},fontsize=\small,breaklines}
\usepackage{framed}
\definecolor{shadecolor}{RGB}{248,248,248}
\newenvironment{Shaded}{\begin{snugshade}}{\end{snugshade}}

\newcommand{\AttributeTok}[1]{\textcolor[rgb]{0.13,0.29,0.53}{#1}}

\newcommand{\BuiltInTok}[1]{#1}

\newcommand{\CommentTok}[1]{\textcolor[rgb]{0.56,0.35,0.01}{\textit{#1}}}

\newcommand{\ConstantTok}[1]{\textcolor[rgb]{0.56,0.35,0.01}{#1}}
\newcommand{\ControlFlowTok}[1]{\textcolor[rgb]{0.13,0.29,0.53}{\textbf{#1}}}
\newcommand{\DataTypeTok}[1]{\textcolor[rgb]{0.13,0.29,0.53}{#1}}
\newcommand{\DecValTok}[1]{\textcolor[rgb]{0.00,0.00,0.81}{#1}}

\newcommand{\FunctionTok}[1]{\textcolor[rgb]{0.13,0.29,0.53}{\textbf{#1}}}

\newcommand{\KeywordTok}[1]{\textcolor[rgb]{0.13,0.29,0.53}{\textbf{#1}}}
\newcommand{\NormalTok}[1]{#1}
\newcommand{\OperatorTok}[1]{\textcolor[rgb]{0.81,0.36,0.00}{\textbf{#1}}}
\newcommand{\OtherTok}[1]{\textcolor[rgb]{0.56,0.35,0.01}{#1}}
\newcommand{\PreprocessorTok}[1]{\textcolor[rgb]{0.56,0.35,0.01}{\textit{#1}}}

\newcommand{\SpecialCharTok}[1]{\textcolor[rgb]{0.81,0.36,0.00}{\textbf{#1}}}

\newcommand{\StringTok}[1]{\textcolor[rgb]{0.31,0.60,0.02}{#1}}

\AtBeginDocument{%
  }

\usepackage{newunicodechar}
\newunicodechar{↦}{\ensuremath{\mapsto}}
\newunicodechar{∧}{\ensuremath{\wedge}}
\newunicodechar{≫}{\ensuremath{\gg}}
\newunicodechar{⇓}{\ensuremath{\Downarrow}}
\newunicodechar{∥}{\ensuremath{\parallel}}

\newunicodechar{β}{\ensuremath{\beta}}
\newunicodechar{δ}{\ensuremath{\delta}}
\newunicodechar{⊥}{\ensuremath{\bot}}

\DeclareTextSymbol{\textrightarrow}{TS1}{25}

\setcitestyle{nosort}


\setcopyright{cc}
\setcctype{by}
\acmDOI{10.1145/3763050}
\acmYear{2025}
\acmJournal{PACMPL}
\acmVolume{9}
\acmNumber{OOPSLA2}
\acmArticle{272}
\acmMonth{10}
\received{2024-04-06}
\received[accepted]{2025-08-12}

\author{Patrick Redmond}
\orcid{0000-0001-5702-0860}
\affiliation{
  \institution{University of California, Santa Cruz}
  \country{USA}
}
\author{Jonathan Castello}
\orcid{0000-0002-8548-3683}
\affiliation{
  \institution{University of California, Santa Cruz}
  \country{USA}
}
\author{Jos\'{e} Manuel Calder\'{o}n Trilla}
\orcid{0009-0001-1145-2032}
\affiliation{
  \institution{Haskell Foundation}
  \country{Canada}
}
\author{Lindsey Kuper}
\orcid{0000-0002-1374-7715}
\affiliation{
  \institution{University of California, Santa Cruz}
  \country{USA}
}

\begin{document}

\title{Exploring the Theory and Practice of Concurrency in the Entity-Component-System Pattern}

\begin{abstract}
    The \emph{Entity-Component-System} (ECS) software design pattern, long used in game development, encourages a clean separation of identity (entities), data properties (components), and computational behaviors (systems).
    Programs written using the ECS pattern are naturally concurrent, and the pattern offers modularity, flexibility, and performance benefits that have led to a proliferation of ECS frameworks.
    Nevertheless, the ECS pattern is little-known and not well understood outside of a few domains.
    Existing explanations of the ECS pattern tend to be mired in the concrete details of particular ECS frameworks, or they explain the pattern in terms of imperfect metaphors or in terms of what it is \emph{not}.
    We seek a rigorous understanding of the ECS pattern via the design of a formal model, \emph{Core ECS}, that abstracts away the details of specific implementations to reveal the essence of software using the ECS pattern.
    We identify a class of Core ECS programs that behave deterministically regardless of scheduling, enabling use of the ECS pattern as a deterministic-by-construction concurrent programming model.
    With Core ECS as a point of comparison, we then survey several real-world ECS frameworks and find that they all leave opportunities for deterministic concurrency unexploited.
    Our findings point out a space for new ECS implementation techniques that better leverage such opportunities.
\end{abstract}

\maketitle

\section{Introduction}
\label{sec:intro}
The \emph{Entity-Component-System} software design pattern (or \emph{ECS pattern}) has been in use for game development since the early 2000s~\citep{bilas2002data,west2007evolve,martin2007entity,west2018using} as well as more recently in GUI programming~\citep{raffaillac-polyphony} and real-time interactive systems~\citep{hatledal-vico}.
Software written using the ECS pattern specifies computations (called \emph{systems}) acting over an association from identifiers (called \emph{entities}) to domain-specific data values (called \emph{components}). We refer to a program composed in this way as an \emph{ECS program}.
The essence of the ECS pattern is the expression of some domain in terms of loosely-coupled transformations over data in a shared non-hierarchical store. These transformations are paired with queries that describe how to fetch their input data from the store.

The ECS pattern encourages the programmer to break aggregated state into distinct components and break complex behaviors into distinct systems,
letting programmers focus on how objects in a domain relate to each other instead of on control flow~\citep{bilas2002data}.
This approach affords several benefits.
ECS programs are inherently concurrent, and for certain forms of concurrency the ECS pattern makes it easy to exploit parallelism and gain a significant and predictable speedup.
The pattern encourages the development of modular high-level domain-specific primitives, enabling members of a team who are working on different aspects of a project (for example, a game's art assets, networking code, and physics engine) to work independently.
Software written using the ECS pattern does not impose a hierarchy of object types, making it easier for developers to adapt to unplanned design changes.
The loosely-coupled architecture of an ECS program is meant to prioritize flexibility over the course of a long-lived development process during which requirements can shift drastically.

Thanks to these benefits, the ECS pattern has endured for more than 20 years\footnote{The earliest reference to the ECS pattern that we are aware of is \citet{bilas2002data}.} as a distinct approach to structuring programs,
and the last decade has seen a proliferation of ECS frameworks for a variety of programming languages~\citep{bevy2024ecs,unity2024ecs,mertens2024flecs,schaller2023specs,gillen2021legion,carpay2018apecs}.
Despite this history, the ECS pattern has rarely been an object of study for software researchers, and among software practitioners the ECS pattern remains well-known only in the domain of game development.
Our work seeks to rectify this lack of attention.

Since existing descriptions of the ECS pattern tend to be mired in low-level details such as memory layouts, or to conflate the pattern with the interface presented by a particular ECS framework, we take a different approach.  We seek to formally model the ECS pattern in an effort to describe it as faithfully to its nature as possible.
We develop a core calculus of the ECS pattern by investigating several prominent ECS frameworks and modeling what we assess to be their core metaphor.
Rather than explaining the ECS pattern in terms of existing theoretical models, we present our formal model as a \emph{synthetic theory} of an as-yet-untreated subject (as distinguished from an \emph{analytic theory}).
We hope it can provide a foundation against which a future analytic theory might be established.

We make the following contributions:
\begin{itemize}[leftmargin=1.5em]
    \item
        In \Cref{sec:core-ecs} we present a core calculus of the ECS pattern, \emph{Core ECS}, with a denotational semantics to give it meaning.
        Core ECS captures the essence of the ECS pattern and aids in reasoning about ECS programs.
        We believe this is the first formalization of the ECS pattern.
    \item
        In \Cref{sec:ecs-conc} we characterize the concurrency and determinism properties of the ECS pattern using Core ECS as a model.
        We define a class of concurrent Core ECS programs that are deterministic in the presence of scheduler non-determinism.
        We suggest that the ECS pattern can be used as a deterministic-by-construction concurrent programming model.
    \item
        In \Cref{sec:practical-ecs} we relate Core ECS to practical implementations of the ECS pattern by surveying real-world ECS frameworks.
        We use Core ECS to make sense of the myriad forms of superficially distinct concurrency that appear in those frameworks.
        We find that the ECS pattern has more concurrency available than existing frameworks take advantage of, suggesting that alternative implementation strategies for ECS frameworks are desirable.
\end{itemize}
We provide context for our contributions by introducing the ECS pattern in \Cref{sec:ecs-background}.
We briefly introduce an \emph{executable model} of Core ECS in \Cref{sec:executable-model-new}.
\Cref{sec:related} discusses related work.
We propose future work and conclude in \Cref{sec:conclusion}.

\section{ECS for the Unfamiliar}
\label{sec:ecs-background}
%

In this section we explain the ECS pattern and terminology for readers who are unfamiliar with it by way of a simple example: a one-dimensional toy physics simulation.
Variations of a toy physics simulation are common in articles explaining the ECS pattern and documentation for ECS frameworks, and we have chosen our example to continue that tradition.
We will use this example throughout the paper.

\subsection{Entities and Components and Systems! Oh, My!}
\label{sec:ecs-background-terms}

The ECS pattern is named for three concepts that are repeatedly instantiated in an ECS program.
\begin{enumerate}
    \item
        An \emph{entity} is an identifier value with which component values may be associated.
    \item
        A \emph{component} is a domain-specific datatype. A value of a component datatype is always associated with an entity (like an object's property).
        We will usually call such values ``components'' for brevity, though
        they are properly ``component values''.
    \item
        A \emph{system} is a domain-specific behavior, expressed as a computation that mutates the association of entities and component values (the program state).
\end{enumerate}

\begin{wrapfigure}{}{0.48\textwidth}
    \vspace{-1.5em}
    \centering
    { \small \includesvg[width=0.46\textwidth]{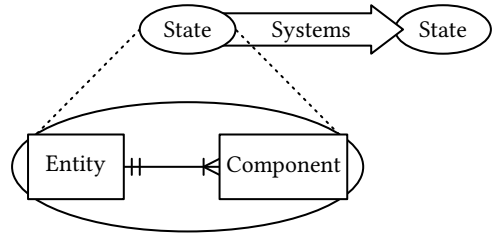} }
    \vspace{-0.25em}
    \caption{
        Entities and components comprise the state of a running ECS program.
        Systems, composed into a single arrow, describe transitions between states.
    }
    \label{fig:ecs-blocks}
    \vspace{-1.25em}
\end{wrapfigure}
\noindent
An ECS program is a \emph{fixed set} of components and systems, together with a \emph{varying one-to-many association} of entities to components.
From some initial state of the entity-component association, an ECS program proceeds iteratively:
The entity-component association determines which systems may be applied to which entities.
When applied, a system may mutate the entity-component association.
This process repeats by rounds to evolve the program state.

\Cref{fig:ecs-blocks} depicts the relationship between entities, components, and systems:
An entity-component association is the state of a running ECS program, while the systems, taken together, are its transitions between states.
In the rest of this section, we will make these concepts concrete with our physics simulation example.

\subsection{A Toy Physics Simulation}
\label{sec:ecs-background-phys}

\begin{figure}[b]
    \begin{subfigure}[t]{0.38\linewidth}
        \begin{center}
            Frame 1: $\texttt{\_→\_\_\_\_\_\_\_}$ \\
            Frame 2: $\texttt{\_\_\_\_\_→\_\_\_}$ \\
            \quad
        \end{center}
        \small
        \quad 
        \caption{
          An object with velocity 4 exhibits inertia by jumping rightward that many units between frame 1 and frame 2.
        }
        \label{fig:badphysics-visual-a}
    \end{subfigure}
    \hfill
    \begin{subfigure}[t]{0.58\linewidth}
        \begin{center}
            Frame 1: $\texttt{\_→\_\_\_s\_\_\_}$ \\
            Frame 2: $\texttt{\_\_\_\_\_2\_\_\_}$ \\
            Frame 3: $\texttt{\_\_\_←\_\_\_→\_}$
        \end{center}
        \small
        \quad 
        \caption{
            A moving object with velocity 4 collides with a stationary object in frame 2. In frame 3, two moving objects with halved velocity move away from the locus of collision.
        }
        \label{fig:badphysics-visual-b}
    \end{subfigure}
\caption{
    Simple text-based visualizations of our toy physics simulation.
    ``$\texttt{←}$'' or
    ``$\texttt{→}$'' indicate an object moving left or right,
    ``$\texttt{s}$'' indicates a stationary object, and
    ``$\texttt{\_}$'' indicates empty space.
    Where two objects occupy the same space, we write the number of objects.
}
\label{fig:badphysics-visual}
\end{figure}

We consider a toy physics simulation in which objects have a one-dimensional \emph{position} represented by an integer (as if scattered on a number line), and a \emph{velocity}, also represented by an integer (where sign represents the object's direction).
The simulation proceeds as a series of frames, one per line.
Between each frame, objects in the simulation may move to the left or right, exhibiting \emph{inertia}.
For example, in \Cref{fig:badphysics-visual-a}, an object with velocity 4 jumps rightward by 4 units along the line between frame 1 and frame 2 of the simulation.
A moving object may also \emph{collide} with a stationary object, annihilating one of them and causing the other to split into two objects traveling in opposite directions at half the original velocity, as in \Cref{fig:badphysics-visual-b}.\footnote{For this toy simulation, we do not consider collisions between a moving object and another moving object.}

To implement this simulation as an ECS program, we must decide which of its elements correspond to entities, components, and systems. We can make the following representation choices:
\begin{itemize}[leftmargin=1.5em]
\item We can represent the simulated objects themselves using entities.
Stationary and moving objects will be distinguished by which components they are associated with.
\item Since objects have a one-dimensional position and velocity, we can represent those quantities as component values associated with entities.
Stationary objects will be those entities having only a position component, and moving objects will be those with both position and velocity.
\item We can implement the inertia and collision resolution behaviors using systems.
\end{itemize}

\noindent
To start the simulation, we initialize it with some number of objects.
We create several entities and attach to each a one-dimensional position component.
To some subset of the entities, we also attach a nonzero velocity.
Frames of the simulation correspond to iterations of the state:
To advance to the next frame, determine which systems may be applied to which entities, and apply them to update the entity-component association.

\subsection{Queries and Schedules}
\label{sec:ecs-background-phys-details}
Our description of the toy physics simulation is necessarily very informal at this stage.
\Cref{sec:core-ecs} will give us the language to write down the behavior precisely, and \Cref{sec:ecs-conc} will make clear how that behavior is a model of a concurrent program.
Nonetheless we wish to prepare the reader by introducing two additional concepts informally: queries and schedules.

\paragraph{System queries}
Consider how the inertia and collision resolution behaviors (systems) will specify the entities to which they apply.
The inertia system moves objects according to their velocity, and the collision system handles the situation where a moving and a stationary object share the same position.
We can phrase these descriptions more precisely in terms of the inclusion or exclusion of components associated with entities (called a \emph{system query}):
\begin{itemize}[leftmargin=1.5em]
    \item The inertia system applies to any entity with both a position and a velocity component.
    \item The collision system applies to any pair of entities, where one has both a position and a velocity component, and the other has a position component and \emph{no} velocity component.
\end{itemize}

\noindent
When a system query is evaluated against the entity-component association, we call the result its \emph{entity matches}.
For the inertia system this could be a set of entities; for the collision system it would be a set of entity tuples, one moving and one stationary according to the query.
The entity matches could also be considered to contain the component values associated with each entity that made it a valid part of the result (and we do this in \Cref{sec:core-ecs}).
For now, we leave the order in which a system processes its entity matches unspecified; \Cref{sec:ecs-conc} will return to this question and the role it plays in our model of concurrency.

\paragraph{System schedules}
Our description of the physics simulation in \Cref{sec:ecs-background-phys} does not describe how the inertia and collision systems will interact with each other, nor does it indicate whether any behaviors are concurrent.
We can answer these questions by indicating how each individual system is treated and how distinct systems relate to each other (and this is called a \emph{system schedule}):
\begin{itemize}[leftmargin=1.5em]
    \item The inertia system runs concurrently, and afterward, the collision system runs sequentially.
\end{itemize}

\noindent
This is only one possible schedule for our systems.
Instead of running one after the other, a schedule could specify that two systems run concurrently.
However in this example we hope the reader will recognize that,
(1) changing the schedule to run the two systems concurrently, or
(2) changing only the collision system to run concurrently,
are both options which lead to undesirable non-determinism!
Consider the case where two moving objects share a position with a single stationary object, perhaps because they approached from either side, as in \Cref{fig:badphysics-conc}.
Our chosen schedule results in the scenario of \Cref{fig:badphysics-conc-a}, but taking the second option may result in the scenario of \Cref{fig:badphysics-conc-b}.\footnote{
    We assume that the collision system always destroys the moving object.
    If such a system were run concurrently with itself, it will allow both moving objects to collide with the stationary object.
    The choice of which object collides first is non-deterministic, and that collision's write to the velocity of the stationary object will be lost.
}
We discuss these scenarios while introducing our formalism in \Cref{sec:core-ecs}, and make clear how our formalization models concurrency trade-offs in \Cref{sec:ecs-conc}.

\begin{figure}[t]
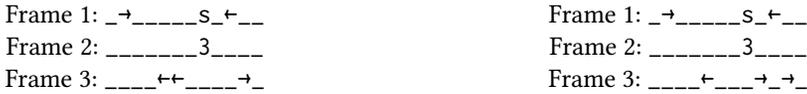

    \begin{subfigure}[t]{0.48\linewidth}
        \begin{center}
            Frame 1: $\texttt{\_→\_\_\_\_\_s\_←\_\_}$ \\
            Frame 2: $\texttt{\_\_\_\_\_\_\_3\_\_\_\_}$ \\
            Frame 3: $\texttt{\_\_\_\_←←\_\_\_\_→\_}$
        \end{center}
        \small
        \quad 
        \caption{Sequential application results in only one collision. The leftward object is unimpeded.}
        \label{fig:badphysics-conc-a}
    \end{subfigure}
    \hfill
    \begin{subfigure}[t]{0.48\linewidth}
        \begin{center}
            Frame 1: $\texttt{\_→\_\_\_\_\_s\_←\_\_}$ \\
            Frame 2: $\texttt{\_\_\_\_\_\_\_3\_\_\_\_}$ \\
            Frame 3: $\texttt{\_\_\_\_←\_\_\_→\_→\_}$
        \end{center}
        \quad 
        \small
        \caption{Concurrent application results in two collisions, a lost write, and non-determinism.}
        \label{fig:badphysics-conc-b}
    \end{subfigure}
\caption{
    A rightward-moving object travels at velocity 6 and a leftward-moving object travels at velocity -2.
    Both are set to collide with a stationary object, but there is more than one way to run the collision system, resulting in different states by frame 3.
}
\label{fig:badphysics-conc}
\end{figure}

\section{Core ECS}
\label{sec:core-ecs}
In this section we describe our core calculus for the ECS pattern, \emph{Core ECS}, and provide a formal denotational semantics of Core ECS to give meaning to the ECS pattern.
We will go on to discuss how the features of Core ECS relate to practical ECS frameworks in \Cref{sec:practical-ecs}.

In Core ECS all \emph{state} (\Cref{sec:core-ecs-state}) is represented by an association between entity identifiers and component values.
A Core ECS program consists of a \emph{schedule} (\Cref{sec:core-ecs-schedule}) of \emph{systems} (\Cref{sec:core-ecs-system}) describing how to decompose and update that state.
Each system is a pair of a \emph{query vector} (\Cref{sec:core-ecs-matches}) and a function that produces a \emph{mutation} (\Cref{sec:core-ecs-mutations}).
A single \emph{query} (\Cref{sec:core-ecs-queries}) describes the shape of a single entity.
A query vector specifies the input to a system, in terms of entities and their components, and a mutation describes how to add, update, or remove entity-component associations in the state.

\subsection{Parameters, Grammars, and State}
\label{sec:core-ecs-state}

\newcommand{\sublang}[1]{\colorbox{cyan!15}{\({#1}\)}}

\newcommand{\bnfeq}{ \mathbin{::=} }
\newcommand{\bnfalt}{ \mathbin{|} }

\newcommand{\qand}[2]{ {#1}\land{#2} }
\newcommand{\qincl}[1]{ \mathrm{incl}_{#1} }
\newcommand{\qexcl}[1]{ \mathrm{excl}_{#1} }
\newcommand{\qanyway}[1]{ \mathrm{anyway}_{#1} }

\newcommand{\query}[1]{ {\llbracket {#1} \rrbracket}_c }
\newcommand{\queryTwo}[2]{ {\llbracket {#1} \rrbracket}_{#2} }
\newcommand{\queryT}[1]{ \llparenthesis {#1} \rrparenthesis }

\newcommand{\mutseq}{ \bullet }
\newcommand{\mutattach}[3]{ \mathrm{attach}_{{#1}}({#2}, {#3}) }
\newcommand{\mutdetach}[2]{ \mathrm{detach}_{{#1}}({#2}) }
\newcommand{\mutnil}{ \mathrm{nil} }
\newcommand{\mutfresh}[1]{ \nu({#1}) }
\newcommand{\mutfreshbound}[2]{ \colorbox{cyan!15}{\nu\ {#1}.\ {#2}}  }

\newcommand{\spSeqComp}{ \fcmp{} }
\newcommand{\spSeqOne}[1]{ \mathrm{seq}({#1}) }
\newcommand{\spConcComp}{ \parallel }
\newcommand{\spConcOne}[1]{ \mathrm{conc}({#1}) }
\newcommand{\appStSched}[2]{ {#1}\mathbin{\Downarrow}{#2} }

\begin{wrapfigure}[]{R}{0.375\textwidth}
    \vspace{-1.5em}
    \centering
    \begin{tabular}{ c r l l }
        \(E\)
            && Entity identifier type \\
        \(K\)
            && Component schema \\
        \(I\)
            && Component label set \\
    \end{tabular}
    \caption{Symbolic parameters to Core ECS. $K$ is indexed by $i\in I$. Each $K_i$ is a type in the sub-language, and so is $E$.}
    \label{fig:core-ecs-params}
    \vspace{-0.5em}
\end{wrapfigure}

\paragraph{Parameters}
To describe an ECS program we will need a programming language with which to write down systems, but we wish to avoid fully integrating Core ECS and its semantics with a specific language.
Instead, Core ECS is parameterized by an underlying programming language (or sub-language).
We will use types from the underlying programming language as the entity identifiers and components that populate the ECS program state.
Therefore Core ECS is also parameterized by a type $E$ and a collection of component types $K$.
The component types come to us in a mapping from component label $i\in I$ to component type $K_i$, which we will call the program \emph{schema}.
We will use $e$ and $k_i$ to refer to terms of the types $E$ and $K_i$, respectively.
\Cref{fig:core-ecs-params} summarizes the parameters of Core ECS to which we assign a symbol (the sub-language has none).

\Cref{fig:example-state} shows examples of parameters $E$ and $K$ for the toy physics simulation of \Cref{sec:ecs-background-phys}, leaving parameter $I$ implied by the domain of the chosen $K$.
Throughout the paper, we use figures with a gray background to continue the running example of the toy physics simulation.

Since Core ECS characterizes a programming \emph{pattern}, most any underlying programming language is an appropriate choice as long as it supports a few key features (functions, sums, products, and a unit, whose types we will write with $\to$, $+$, $\times$, and $\top$, respectively).
The sub-language must also have a way to generate fresh (universally unique) entity identifiers and a way to communicate mutations ($M$ in \Cref{fig:core-ecs-terms}) to Core ECS.
Although we give examples of Core ECS programs that fix a sub-language, our specification of Core ECS itself cannot refer to terms of the underlying language directly because it is a parameter; instead we will write a type with a light blue background (e.g. \sublang{E\to M}) to stand in for a term in the underlying language having that type.

\newcommand{\queryGrammar}{
        \(q:Q\)
            & \(\bnfeq\)
            & \(            \qincl{i}
                \;\bnfalt\; \qexcl{i}
                \;\bnfalt\; \qanyway{i}
                \;\bnfalt\; (\qand{q}{q'})
                \)
}

\newcommand{\mutationGrammar}{
        \(m:M\)
            & \(\bnfeq\)
            & \(            \mutattach{i}{e}{k_i}
                \;\bnfalt\; \mutdetach{i}{e}
                \;\bnfalt\; (m \mutseq m')
                \;\bnfalt\; \mutfresh{ \sublang{E \to M} }
                \;\bnfalt\; \mutnil
                \)
}

\newcommand{\systemGrammar}{
        \(s:S\)
            & \(\bnfeq\)
            & \( (\,
                 \vec{q} 
                 ,\
                 \sublang{  E^{\mathrm{dim}(\vec{q})} \times \queryT{\vec{q}} \to M  }
                 \,)
              \)
}

\newcommand{\scheduleGrammar}{
        \(z:Z\)
            & \(\bnfeq\)
            & \(            \spConcOne{s}
                \;\bnfalt\; \spSeqOne{s}
                \;\bnfalt\; (z \spConcComp z')
                \;\bnfalt\; (z \spSeqComp  z')
                \)
}

\begin{figure}
    \begin{tabular}{ c r l l }
        \queryGrammar
            & \quad\quad Query \\
        \mutationGrammar
            & \quad\quad Mutation \\
        \systemGrammar
            & \quad\quad System \\
        \scheduleGrammar
            & \quad\quad Schedule \\
    \end{tabular}
    \vspace{6pt} 
    \caption{
        Grammars defining how to construct Core ECS terms.
        Types written with a light blue background indicate terms of that type in the underlying language.
    }
    \label{fig:core-ecs-terms}
\end{figure}
\begin{wrapfigure}[]{R}{0.40\textwidth}
    \vspace{-0.75em}
    \centering
    \(\begin{aligned}
            \queryT{ \qincl{i}         } &\triangleq K_i
        \\  \queryT{ \qexcl{i}         } &\triangleq \top
        \\  \queryT{ \qanyway{i}       } &\triangleq K_i + \top
        \\  \queryT{ \qand{q}{q'}      } &\triangleq \queryT{q} \times \queryT{q'}
        \\
        \\  \queryT{ \vec{q}           } &\triangleq
                \prod_{1 \le j \le \mathrm{dim}( \vec{q} )}{\queryT{ \vec{q}_j }}
    \end{aligned}\)
    \caption{Query and query-vector result type.}
    \label{fig:core-ecs-query-type}
    \vspace{-0.5em}
\end{wrapfigure}

\paragraph{Grammars}
\Cref{fig:core-ecs-terms} summarizes the grammars used to construct a Core ECS term; later subsections will go into more detail.
For now, know that:
A Core ECS program is a \emph{schedule} of \emph{systems}.
Each system is a pair of a \emph{query vector} and a function in the underlying programming language that produces a \emph{mutation}.
A system function takes as input a pair, for which type of the second component is dependent on the query vector according to the function $\queryT{-}$ in \Cref{fig:core-ecs-query-type}.

We write $\vec{\;}$ over terms indicate a vector and we use $\mathrm{dim}(-)$ for a vector's size.
We write $T^n$ for the type of a vector of size $n$ containing elements of type $T$.
\Cref{fig:core-ecs-terms} uses all three of these conventions; in particular, the type of the function in a system takes $E^{\mathrm{dim}(\vec{q})}$ as part of its input, which is a vector of size $\mathrm{dim}(\vec{q})$ that contains entities.

\paragraph{State}
All information about the progress of a running ECS program in Core ECS is represented by a single logical ``entity-component association''.
In this \emph{state} an entity may be associated with at most one value of each component type.
We structure Core ECS state as an indexed family of partial mappings from entities to component values.
Given parameters $E$, $K$, and $I$, state $c$ is indexed by $i\in I$, and each $c_i$ is a partial mapping with type $E\mathbin{\dot\rightarrow}K_i$.
We give the definition of state below at left, and \Cref{fig:example-state} depicts a state of our running example,
\begin{alignat*}{5}
        c
    &\triangleq
        \{c_i : E\mathbin{\dot\rightarrow}K_i\}_{i \in I}
    &&\hspace{9em}&
        \mathrm{live}(c)
    &\triangleq
        \bigcup_{i \in I}{ \mathrm{dom}(c_i) }
\end{alignat*}

At any point during the execution of an ECS program, it is helpful to know the set of entities that exist in at least one partial mapping of the state, because they have at least one component attached.
This set of live entities $\mathrm{live}(c)$ for state $c$, defined above at right, is the union of the domains of the partial mappings (writing $\mathrm{dom}(-)$ for the domain of a mapping).
Only live entities are considered by system queries for inclusion in entity matches.
Attaching a single component to an entity is sufficient for it to appear in the set of live entities, and conversely, detaching all components is necessary to remove the entity.

\begin{figure}[h]
    \begin{mdframed}[backgroundcolor=gray!15,linewidth=0]
    \[
    \begin{aligned}
        E &\triangleq \{e_1,e_2,e_3,\dots\}
        \\
        K &\triangleq
            \{
                \mathbf{Pos}\mapsto\mathbb{Z},
                \mathbf{Vel}\mapsto\mathbb{Z}
            \}
        \\
        c &\triangleq
            \left\{
                \mathbf{Pos}\mapsto
                    \{ e_1\mapsto 1
                     , e_2\mapsto 7
                     , e_3\mapsto 9
                    \},
                \mathbf{Vel}\mapsto
                    \{ e_1\mapsto 6
                     , e_3\mapsto -2
                    \}
            \right\}
    \end{aligned}
    \]
    \end{mdframed}
    \caption{
        We represent the state $c$ of the physics simulation from Frame 1 of \Cref{fig:badphysics-conc-a} by defining an entity type $E$ and schema $K$:
        We choose a set of opaque entity identifiers for $E$, and a $K$
        that maps the opaque component labels $\mathbf{Pos}$ and $\mathbf{Vel}$
        to the type of integers.
        In the example state $c$ we have two moving objects represented by
        entities $e_1$, $e_3$, and one stationary object $e_2$.
    }
    \label{fig:example-state}
\end{figure}

\subsection{Queries}
\label{sec:core-ecs-queries}

A system specifies its input using a simple query language to identify a subset of live entities.
A system may have more than one query, meaning that it takes more than one entity at a time as input (\Cref{sec:core-ecs-matches}), but here we focus on single queries.
A single query describes a single entity in terms of its components, but may have many matches among the live entities.
A query indicates not only the entity, but also which of its components, to provide to a system.
The grammar of queries given in \Cref{fig:core-ecs-terms} is repeated here.
\begin{center}
\begin{tabular}{ c r l }
    \queryGrammar
\end{tabular}
\end{center}
A query is a conjunction ($q \land q'$) of constraints ($\qincl{i}$ and $\qexcl{i}$) that filter entities based on whether each has an association with the component indicated by the label $i$ (except for $\qanyway{i}$, which performs no filtering).

The meaning of a query, with respect to a state $c$, is given by the function $\query{-}$ defined below at left.
This function retrieves a partial mapping of entities matching the query.
The value associated with each entity in the mapping, its \emph{component result}, is made up of those component values selected by the query.
The type of the component result is given for a query by the function $\queryT{-}$, shown in \Cref{fig:core-ecs-query-type} and duplicated below at right.
For a query $q$, the result $\query{q}$ will have type $E \mathbin{\dot\rightarrow} \queryT{q}$.
See \Cref{fig:example-queries} for an interpretation of the queries introduced in \Cref{sec:ecs-background-phys-details}.
\begin{alignat*}{5}
    \query{ \qincl{i} }
        &\triangleq c_i
&&\quad\quad\quad&
    \queryT{ \qincl{i} }
        &\triangleq K_i
\\
    \query{ \qexcl{i} }
        &\triangleq \{ \ast \}_{e \in \mathrm{live}(c) \setminus \mathrm{dom}(c_i)}
&&\quad\quad\quad&
    \queryT{ \qexcl{i} }
        &\triangleq \top
\\
    \query{ \qanyway{i} }
        &\triangleq c_i \mathbin{\dot\cup} \{ \ast \}_{e \in \mathrm{live}(c) \setminus \mathrm{dom}(c_i)}
&&\quad\quad\quad&
    \queryT{ \qanyway{i} }
        &\triangleq K_i + \top
\\
    \query{ \qand{q}{q'} }
        &\triangleq \{ (\query{q}(e), \query{q'}(e)) \}_{e \in \mathrm{live}(c)}
&&\quad\quad\quad&
    \queryT{ \qand{q}{q'} }
        &\triangleq \queryT{q} \times \queryT{q'}
\end{alignat*}
The $\qincl{i}$ constraint retrieves $c_i$, the partial mapping of entities that have the indicated component, mapped to the component values themselves.
The $\qexcl{i}$ constraint retrieves the complement of $\qincl{i}$, every live entity that does not have the indicated component.
Since there is no associated component value for this mapping, every entity is mapped to $\ast$ (unit) with type $\top$.
We write $a \setminus b$ for the members of $a$ not present in $b$.
The $\qanyway{i}$ constructor retrieves a mapping containing \emph{every} live entity, that yields the component if it is present and the unit value otherwise.\footnote{
    The $\qanyway{i}$ constructor does not play an important role in the rest of this paper. However, in practice it is common for an ECS program to deal with optional components, and so we include it for completeness.
}
The conjunction of constraints behaves like a relational join: results of the conjuncts are filtered to those entities matching both queries, and the component results retrieved for each query are merged into a tuple.

Our queries are related to a fragment of conjunctive query systems.
Conjunctive queries include selection, projection, and joins on N-ary relations~\citep{zhang2022relational}.
In Core ECS state, the partial mappings per component label ($E\mathbin{\dot\rightarrow}K_i$) are analogous to binary relations.
Core ECS queries have selection ($\qincl{i}$ and $\qexcl{i}$), projection ($\qincl{i}$ and $\qanyway{i}$), and joins ($\qand{q}{q'}$).

\begin{figure}
    \begin{subfigure}[t]{0.48\linewidth}
        \begin{mdframed}[backgroundcolor=gray!15,linewidth=0]
        \[
        \begin{aligned}
            x &\triangleq
                \qand{
                    \qincl{\mathbf{Pos}}
                }{
                    \qincl{\mathbf{Vel}}
                }
            \\
            \queryT{x} &=
                \queryT{ \qincl{\mathbf{Pos}} }
                \times
                \queryT{ \qincl{\mathbf{Vel}} }
                \\&=
                K_\mathbf{Pos}
                \times
                K_\mathbf{Vel}
                \\&=
                \mathbb{Z}
                \times
                \mathbb{Z}
            \\
            \query{x} &=
                \{
                    ( c_{\mathbf{Pos}}(e)
                    , c_{\mathbf{Vel}}(e)
                    )
                \}_{e \in \{e_1,e_2,e_3\}}
                \\&=
                \{
                    e_1\mapsto(1,6),
                    e_3\mapsto(9,-2)
                \}
        \end{aligned}
        \]
        \end{mdframed}
        \caption{Query $x$ retrieves entities with both position and velocity: the moving objects.}
        \label{fig:example-queries-a}
    \end{subfigure}
    \hfill
    \begin{subfigure}[t]{0.48\linewidth}
        \begin{mdframed}[backgroundcolor=gray!15,linewidth=0]
        \[
        \begin{aligned}
            y &\triangleq
                \qand{
                    \qincl{\mathbf{Pos}}
                }{
                    \qexcl{\mathbf{Vel}}
                }
            \\
            \queryT{y} &=
                \queryT{ \qincl{\mathbf{Pos}} }
                \times
                \queryT{ \qexcl{\mathbf{Vel}} }
                \\&=
                K_\mathbf{Pos}
                \times
                \top
                \\&=
                \mathbb{Z}
                \times
                \top
            \\
            \query{y} &=
                \{
                    ( c_{\mathbf{Pos}}(e)
                    , \{e_2\mapsto\ast\}(e)
                    )
                \}_{e \in \{e_1,e_2,e_3\}}
                \\&=
                \{
                    e_2\mapsto(7,\ast)
                \}
        \end{aligned}
        \]
        \end{mdframed}
        \caption{Query $y$ retrieves entities with position but without velocity: the stationary objects.}
        \label{fig:example-queries-b}
    \end{subfigure}
\caption{
    Two example queries with applications of $\query{-}$ and
    $\queryT{-}$ against the state given in \Cref{fig:example-state}.
}
\label{fig:example-queries}
\end{figure}

\begin{figure}
    \begin{mdframed}[backgroundcolor=gray!15,linewidth=0]
    \[
    \begin{aligned}
        \queryT{ \langle x,y \rangle }
            &=
                \queryT{ x } \times \queryT{ y }
            \\&=
                ( K_\mathbf{Pos} \times K_\mathbf{Vel} )
                \times
                ( K_\mathbf{Pos} \times \top )
            \\&=
                (\mathbb{Z} \times \mathbb{Z})
                \times
                (\mathbb{Z} \times \top)
        \\
        \query{ \langle x,y \rangle }
            &=
                \query{ x } \times \query{ y }
            \\&=
                \{ e_1\mapsto(1,6),
                   e_3\mapsto(9,-2)
                \}
                \times
                \{ e_2\mapsto(7,\ast)
                \}
            \\&=
            \{
                \langle e_1,e_2 \rangle \mapsto
                \left\langle
                    (1,6),
                    (7,\ast)
                \right\rangle,
            \\ &\hspace{1.65em}
                \langle e_3,e_2 \rangle \mapsto
                \left\langle
                    (9,-2),
                    (7,\ast)
                \right\rangle
            \}
    \end{aligned}
    \]
    \end{mdframed}
\caption{
    Entity matches for query vector $\langle x,y \rangle$ consisting of queries from \Cref{fig:example-queries}, interpreted over state from \Cref{fig:example-state}.
    Together, these queries are the query vector of the collision resolution system described in \Cref{sec:ecs-background}.
}
\label{fig:example-entity-matches}
\end{figure}

\subsection{Query Vectors}
\label{sec:core-ecs-matches}

Recall from \Cref{sec:ecs-background-phys-details} that a system specifies its input with a system query, and the results of evaluating a system query against the entity-component association are entity matches.
In Core ECS a system query is specified with a \emph{query vector}, and its entity matches are the cartesian product of its constituent queries' results.
Systems that implement interactions between entities (e.g. the collision system of \Cref{sec:ecs-background-phys}) use a query vector with multiple queries, and take as input an entity match containing that many entities and the same number of component results.
To compute the entity matches for a system query in Core ECS, we extend $\query{-}$ to take a query vector $\vec{q}$ and return a product of the partial mappings, below at left.
\[
\begin{aligned}
    \query{ \vec{q} }
    &\triangleq \prod_{1 \le j \le \mathrm{dim}( \vec{q} )}{\query{ q_j }}
  &&\hspace{7em}&
    \queryT{ \vec{q} }
    &\triangleq \prod_{1 \le j \le \mathrm{dim}( \vec{q} )}{\queryT{ q_j }}
\end{aligned}
\]
Recall from \Cref{sec:core-ecs-queries} that the type of a query result is a partial mapping.
This is still true for a vector of queries:
The cartesian product of \(n\) query results is a single partial mapping of type \(E^n \mathbin{\dot\rightarrow} \queryT{\vec{q}}\) that takes a vector of \(n\) entities $\vec{e}$ to a vector of \(n\) component results $\vec{w}$.
The $j$th entity and $j$th component result in each entity match $\vec{e}\mapsto\vec{w}$ are drawn from the result of the $j$th query in the query vector $\vec{q}$.
\Cref{fig:example-entity-matches} demonstrates the computation of entity matches for the example collision system.
We will also find it useful in \Cref{sec:core-ecs-system} to regard any partial mapping produced by \(\query{-}\) as having some consistent but unspecified total order as in \(\{\vec{e} \mapsto \vec{w},\; \vec{e'} \mapsto \vec{w'},\; ... \}\).

\subsection{Mutations}
\label{sec:core-ecs-mutations}

Given an entity match as input, a system specifies how Core ECS state should change by generating a \emph{mutation}.
Each system includes a function written in the programming language over which Core ECS is parameterized, and it is the task of this function to produce a mutation.
Intuitively, a mutation is a description of a state update; it replaces some of the mappings with updated values, others are made undefined, and some new mappings may be added.
The grammar for mutations first shown in \Cref{fig:core-ecs-terms} is as follows.
\begin{center}
\begin{tabular}{ c r l }
    \mutationGrammar
\end{tabular}
\end{center}
The second parameter of the $\mutattach{i}{e}{-}$ constructor is a value $k_i$ of the indicated component $i$.
The $\mutfresh{-}$ ``new'' form requires a function written in the underlying language (see \Cref{fig:example-mutation} for an example).
The composition operator $m \mutseq m'$ is right-biased,
and $\mutnil$ is an empty mutation.

\newcommand{\appStMut}[2]{ { #1 }\downarrow{ #2 } }

The meaning of a mutation $m$ interpreted against state $c$ is a new state given by the function $\appStMut{c}{m}$ defined below.
We use the notation $t\{k\mapsto v\}$ for the mapping $t$ in which $k$ has been updated to the value $v$, and $t\{k\mapsto\bot\}$ for the mapping $t$ in which $k$ has been made undefined.

\begin{wrapfigure}[]{l}{0.45\textwidth}
\begin{minipage}{1.0\linewidth}
\[
\begin{aligned}
    \appStMut{c}{ \mutattach{i}{e}{k_i} } &\triangleq
        c\{ i \mapsto
            c_i\{e\mapsto k_i\}
        \}
    \\
    \appStMut{c}{ \mutdetach{i}{e} } &\triangleq
        c\{ i \mapsto
            c_i\{e\mapsto \bot\}
        \}
    \\
    \appStMut{c}{ (m \mutseq m') } &\triangleq
        \appStMut{ (\appStMut{c}{m}) }{m'}
    \\
    \appStMut{c}{ \mutfresh{f} } &\triangleq
        \appStMut{c}{ f(\mathbf{fresh}) }
    \\
    \appStMut{c}{ \mutnil } &\triangleq
        c
\end{aligned}
\]
\end{minipage}
\end{wrapfigure}
The $\mutattach{i}{e}{k_i}$ and $\mutdetach{i}{e}$ mutations update state $c$ such that $i$ maps to $c_i$, itself updated with a new mapping for entity $e$.
Attach will add or overwrite the mapping for entity $e$ to the component value $k_i$.
Detach will remove the mapping for $e$ from $c_i$.
As discussed in \Cref{sec:core-ecs-state}, we assume a facility for generating never-before-seen entities, referenced here as $\mathbf{fresh}$, and used to handle the $\mutfresh{-}$ case.
For $\mutfresh{f}$ the environment provides a fresh entity, which is passed into the function $f$ to produce a mutation.
The composition $m \mutseq m'$ operator first applies $m$ and next $m'$, such that coincident mappings in $m$ will be replaced by those in $m'$.

We give an example interpretation of a mutation in \Cref{fig:example-mutation} using the vocabulary of our running example from \Cref{sec:ecs-background-phys}.
\Cref{fig:example-mutation} also introduces the underlying language that we will use throughout our subsequent examples.
Recall that the underlying language is a parameter to Core ECS; our particular choice of language for these examples is only for demonstration.
We use a lambda calculus extended with conveniences to make writing the examples concise (see
\ifdefined\ARXIVVERSION
\Cref{apx:example-grammar}
\fi
\ifdefined\PACMPLVERSION
\citet[Appendix A]{coreecs-extended}.
\fi
for its full grammar).
As it is not the focus of our paper, we do not provide a semantics for it --- its behavior is standard and contains no surprises.

Early iterations of Core ECS used freer monads~\citep{kiselyov2015freer} to represent the body of a system as a monadic computation with algebraic effects (attach, detach, and fresh).
However, it is simpler to use a monoid with generators attach \& detach, a binary composition operator, and an identity nil:
Instead of having a ``fresh'' effect chained via the monadic bind operation to its continuation, the $\mutfresh{-}$ mutation constructor takes that continuation directly and we push the chaining up into its abstract implementation.

\begin{figure}[b]
    \vspace{1em}
    \begin{mdframed}[backgroundcolor=gray!15,linewidth=0]
    \[
    \begin{aligned}
        & \hspace{0.42em}
          \appStMut{c}{
            ( \mutdetach{ \mathbf{Pos} }{e_1}
              \;\mutseq\;
              \mutfresh{ \lambda\,d.\ \mutattach{ \mathbf{Vel} }{d}{ 2 } }
            )
        }
        \\={}&
            \appStMut{
                ( \appStMut{c}{
                    \mutdetach{ \mathbf{Pos} }{e_1}
                } )\;
            }{     \;
                \mutfresh{ \lambda\,d.\ \mutattach{ \mathbf{Vel} }{d}{ 2 } }
            }
        \\={}&
            \appStMut{
                \left\{
                    \mathbf{Pos}\mapsto
                        \{ e_2\mapsto 7
                         , e_3\mapsto 9
                        \},
                    \mathbf{Vel}\mapsto
                        \{ e_1\mapsto 6
                         , e_3\mapsto -2
                        \}
                \right\}
            }{
                \mutfresh{ \lambda\,d.\ \mutattach{ \mathbf{Vel} }{d}{ 2 } }
            }
        \\={}&
            \appStMut{
                \left\{
                    \mathbf{Pos}\mapsto
                        \{ e_2\mapsto 7
                         , e_3\mapsto 9
                        \},
                    \mathbf{Vel}\mapsto
                        \{ e_1\mapsto 6
                         , e_3\mapsto -2
                        \}
                \right\}
            }{
                \mutattach{ \mathbf{Vel} }{e_4}{ 2 }
            }
        \\={}&
            {
                \left\{
                    \mathbf{Pos}\mapsto
                        \{ e_2\mapsto 7
                         , e_3\mapsto 9
                        \},
                    \mathbf{Vel}\mapsto
                        \{ e_1\mapsto 6
                         , e_3\mapsto -2
                         , e_4\mapsto 2
                        \}
                \right\}
            }
    \end{aligned}
    \]
    \end{mdframed}
\caption{
    An example of a mutation interpreted against the state in \Cref{fig:example-state}. It is only illustrative --- the resulting state is not used in our running example.
}
\label{fig:example-mutation}
\end{figure}

\subsection{Systems}
\label{sec:core-ecs-system}

With our discussion of system inputs in \Cref{sec:core-ecs-queries,sec:core-ecs-matches}, and system outputs in \Cref{sec:core-ecs-mutations}, we can now explain Core ECS's notion of a system and how it uses entity matches.
We reproduce the definition from \Cref{fig:core-ecs-terms} here.
\begin{center}
\begin{tabular}{ c r l }
    \systemGrammar
\end{tabular}
\end{center}
A system is a static pairing of a query vector and a function in the underlying programming language. 
For a system $s$, system query $q_s$ refers to the query vector and system function $f_s$ is the function written in the underlying language.
The function $f_s$ takes an entity match of the type produced by $q_s$, and returns a mutation to state.

\begin{figure}[b]
    \begin{subfigure}[t]{0.36\linewidth}
        \begin{mdframed}[backgroundcolor=gray!15,linewidth=0]
        \[
        \begin{aligned}
            \beta
                & \triangleq{} q_\beta,\ f_\beta \\
            q_{\beta}
                & \triangleq{}
                    \langle
                        \qand{
                            \qincl{\mathbf{Pos}}
                        }{
                            \qincl{\mathbf{Vel}}
                        }
                    \rangle
                    \\
            f_{\beta}
                & \triangleq{} \lambda \left( e_j, \; (p_j, v_j) \right). \\
                & \enspace \mutattach{\mathbf{Pos}}{e_j}{p_j + v_j}
        \end{aligned}
        \]
        \end{mdframed}
        \caption{
            The inertia system $\beta$ updates the position of moving entity $e_j$ by adding in one unit of its velocity $v_j$ (reflecting a fixed time gap).
        }
        \label{fig:example-system-functions-inertia}
    \end{subfigure}
    \hfill
    \begin{subfigure}[t]{0.60\linewidth}
        \begin{mdframed}[backgroundcolor=gray!15,linewidth=0]
        \[
        \begin{aligned}
            \delta
                & \triangleq{} q_\delta,\ f_\delta \\
            q_{\delta}
                & \triangleq{}
                    \langle
                        \qand{
                            \qincl{\mathbf{Pos}}
                        }{
                            \qincl{\mathbf{Vel}}
                        },\;
                        \qand{
                            \qincl{\mathbf{Pos}}
                        }{
                            \qexcl{\mathbf{Vel}}
                        }
                    \rangle
                    \\
            f_{\delta}
                & \triangleq{} \lambda
                    \left(
                        \langle e_j, e_h \rangle,\;
                        \langle (p_j, v_j), (p_h, -) \rangle
                    \right)
                    .
                \\&\enspace
                    \mathop{\textbf{if}}
                    p_j \overset{?}{=} p_h
                    \mathop{\textbf{then}}
                \\&\enspace\hspace{1em}
                    \mutdetach{\mathbf{Pos}}{e_j}
                    \mutseq
                    \mutdetach{\mathbf{Vel}}{e_j}
                    \;\mutseq
                \\&\enspace\hspace{1em}
                    \mutattach{\mathbf{Vel}}{e_h}{\lceil v_j/2 \rceil}
                    \;\mutseq
                \\&\enspace\hspace{1em}
                    \nu\left(
                        \lambda\ e_\ell.\
                        \mutattach{\mathbf{Pos}}{e_\ell}{p_j}
                        \mutseq
                        \mutattach{\mathbf{Vel}}{e_\ell}{\lceil -v_j/2 \rceil}
                    \right)
                \\&\enspace
                \mathop{\textbf{else}}
                    \mutnil
        \end{aligned}
        \]
        \end{mdframed}
        \caption{
            The collision resolution system $\delta$ splits the velocity of a moving entity $e_j$, when it collides with a stationary entity $e_h$, among the unmoving entity and a new entity, $e_\ell$.
        }
        \label{fig:example-system-functions-collision}
    \end{subfigure}
\caption{
    We formalize the systems from \Cref{sec:ecs-background-phys} with the queries from \Cref{fig:example-queries,fig:example-entity-matches} by writing down terms in the language defined for examples.
}
\label{fig:example-system-functions}
\end{figure}

\Cref{fig:example-system-functions} shows Core ECS definitions for the inertia and collision resolution systems first described in \Cref{sec:ecs-background-phys}.
Recall from \Cref{sec:ecs-background-phys-details}, that the collision system ($\delta$ in \Cref{fig:example-system-functions-collision}) applies to any pair a moving and a stationary object.
We made this query vector, and its result, concrete in \Cref{sec:core-ecs-matches} by describing that an entity match is an element of the cartesian product of query results from a query vector, and working out that result in \Cref{fig:example-entity-matches}.
Now we can see here in \Cref{fig:example-system-functions-collision} that, not only does system query $q_\delta$ have those two elements, the inputs to system function $f_\delta$ (the entity vector and component result vector) also have two elements, all corresponding to the moving and stationary object contained in each entity match.

We will now define precisely how those entity matches are divided and applied to a system function.
We will introduce two different ways to produce a mutation for a state using a system: \emph{concurrent production} and \emph{sequential production}.
Both use a common notion of applying a system to entity matches, but differ in which entity matches are applied.
We must first define two necessary utilities,
\emph{applying} a system to entity matches to produce a composite mutation, and \emph{rolling} a system over entity matches to produce a new entity matches.
With those definitions: concurrent production simply applies a system to entity matches to produce a mutation, whereas sequential production first rolls the system over the entity matches and then applies the system.

The definition of how to \emph{apply} a system $s$ to entity matches is given below.
Recall from \Cref{sec:core-ecs-matches} that entity matches have a consistent but unspecified total order as in \(\{\vec{e_1} \mapsto \vec{w_1},\; \vec{e_2} \mapsto \vec{w_2},\; ... \}\).
We use that ordering here to write a schematic definition of function application syntax against the name of the system, $s(-)$.
\begin{align*}
  s(\{\vec{e_1} \mapsto \vec{w_1},\; \vec{e_2} \mapsto \vec{w_2},\; ... \})
&\triangleq
  f_s(\vec{e_1}, \vec{w_1}) \mutseq f_s(\vec{e_2}, \vec{w_2}) \mutseq ...
\end{align*}
That is, applying a system to some entity matches \emph{en masse} is the same as composing the mutation generated by applying the system function to each match one at a time.
Since the order of the entity matches is unspecified, so too is the order of composition of the resulting mutation.
In \Cref{sec:ecs-conc} we will say more about the implications of this unspecified order.

We define how to \emph{roll} a system $s$ over entity matches $r$ from a starting state $c$, next.
The function \(\mathrm{roll}_s(c, r)\) replaces or drops elements of $r$ to account for the cumulative effect of mutations over $c$ generated by $s$ while passing through $r$.\footnote{
    Many readers will recognize this as a fold-left with an inlined function to filter and map the input structure.
}
It is called ``roll'' by analogy with the rolling shutter of a camera, in which each line of pixels is observed a slight amount of time after the line prior.
\begin{align*}
    \mathrm{roll}_s(c, \{\vec{e} \mapsto \vec{w},\; ... r\})
&\triangleq
    \textbf{let } \vec{w}' = \query{q_s}(\vec{e})
    \textbf{ in }
    \{
        \vec{e} \mapsto \vec{w}'
        ,\; ...
        \mathrm{roll}_s(\appStMut{c}{f_s(\vec{e},\vec{w}')}, r)
    \}
\\
    \mathrm{roll}_s(c, \emptyset)
&\triangleq
    \emptyset
\end{align*}
Roll inductively visits each of the given entity matches following their order \(\{\vec{e_1} \mapsto \vec{w_1}, \vec{e_2} \mapsto \vec{w_2},\; ... \}\).
Roll performs the system query $\query{q_s}$, at given state $c$, and looks up $\vec{e}$ in the resulting entity matches to obtain an updated component result vector $\vec{w}'$.
In the new entity matches that roll returns, $\vec{e}$ is mapped to the updated $\vec{w}'$, and the recursive step uses state updated with the mutation $f_s(\vec{e},\vec{w}')$.
As a notational convenience in this definition, when $\query{q_s}$ is undefined on $\vec{e}$ (meaning $\vec{w}'$ is $\bot$) the system function $f_s$ is considered to yield \(\mathrm{nil}\), and we regard $\vec{e}$ as dropped from the resulting entity matches (the rest of the entity matches $r$ are processed normally).
This elision is what underlies the behavior of \Cref{fig:badphysics-conc-a} in which only one collision occurs even though two appear to be eligible.

We can now define concurrent and sequential production of a mutation for a state using a system.
An illustration of sequential production appears in \Cref{fig:example-exec-collision-seq}.
\begin{center}
\begin{tabular}{ c c c }
    Concurrent production
    & \hspace{6em} &
    Sequential production
    \\
    \( s(\query{q_s}) \)
    & &
    \( s(\mathrm{roll}_s(c,\query{q_s})) \)
\end{tabular}
\end{center}
Concurrent production observes the state $c$ once to obtain and apply entity matches to the system, such that each mutation over $c$ produced by $f_s$ is not visible to the other uses of $f_s$.
In contrast, sequential production observes state once to obtain initial entity matches and again before processing each, such that each mutation produced by $f_s$ is in light of every prior mutation produced by $f_s$.

\begin{figure}[h]
    \begin{mdframed}[backgroundcolor=gray!15,linewidth=0]
    \[
    \begin{aligned}
        c' \triangleq{}&
            \left\{
                \mathbf{Pos}\mapsto
                    \{ e_1\mapsto 7
                     , e_2\mapsto 7
                     , e_3\mapsto 7
                    \},
                \mathbf{Vel}\mapsto
                    \{ e_1\mapsto 6
                     , e_3\mapsto -2
                    \}
            \right\}
    \\
        \delta(\mathrm{roll}_\delta(c', \queryTwo{q_\delta}{c'}))
        ={}&
            \delta(\mathrm{roll}_\delta(c', \{
                                    \langle e_1, e_2 \rangle \mapsto \langle (7, 6), (7, \ast)\rangle,
                                    \langle e_3, e_2 \rangle \mapsto \langle (7, -2), (7, \ast) \rangle
                                  \} ))
        \\={}&
            \delta(\{
                                    \langle e_1, e_2 \rangle \mapsto \langle (7, 6), (7, \ast)\rangle
                                  \} )
        \\={}&
            f_\delta(
                                    \langle e_1, e_2 \rangle , \langle (7, 6), (7, \ast)\rangle
                                   )
        \\={}&
            \mutdetach{\mathbf{Pos}}{e_1} \mutseq
            \mutdetach{\mathbf{Vel}}{e_1} \mutseq
            \mutattach{\mathbf{Vel}}{e_2}{\lceil 6/2 \rceil}\ \mutseq
        \\& \mutattach{\mathbf{Pos}}{e_4}{7} \mutseq
            \mutattach{\mathbf{Vel}}{e_4}{\lceil -6/2 \rceil}
    \end{aligned}
    \]
    \end{mdframed}
\caption{
    Sequential production of a mutation using collision system $\delta$ (\Cref{fig:example-system-functions-collision}) for state $c'$.
    In state $c'$, two moving objects share position with one stationary object.
    State $c'$ was obtained (not shown) by applying to state $c$ (\Cref{fig:example-state}) the concurrent production mutation for state $c$ using inertia system $\beta$ (\Cref{fig:example-system-functions-inertia}).
}
\label{fig:example-exec-collision-seq}
\end{figure}

\subsection{Schedules}
\label{sec:core-ecs-schedule}

We complete our description of Core ECS by introducing how one or more systems combine in a \emph{schedule} to make a complete Core ECS program.
A schedule precisely describes when each system query vector observes state and when each mutation produced by a system is applied, all relative to each other.
The grammar of schedules is as follows.
\begin{center}
\begin{tabular}{ c r l }
    \scheduleGrammar
\end{tabular}
\end{center}
The $\spConcOne{s}$ and $\spSeqOne{s}$ constructors lift a single system to a schedule.
The $(z_1 \spConcComp z_2)$ and $(z_1 \spSeqComp z_2)$ constructors compose sub-schedules.
Both pairs include a concurrent and a sequential variant, and we will explore the implications of the concurrent variants for the determinism of a Core ECS program in \Cref{sec:ecs-conc}.

We define the function $\appStSched{c}{z}$, below, to interpret a schedule $z$ over a state $c$ and produce a composite mutation which may be used to advance the state of an ECS program.

\begin{wrapfigure}[]{l}{0.45\textwidth}
\begin{minipage}{1.0\linewidth}
\[
\begin{aligned}
    \appStSched{c}{ \spConcOne{s} } &\triangleq
        s(\query{q_s})
\\
    \appStSched{c}{ \spSeqOne{s} } &\triangleq
        s(\mathrm{roll}_s(c,\query{q_s}))
\\
    \appStSched{c}{ (z \spConcComp z') } &\triangleq
        (\appStSched{c}{z}) \mutseq (\appStSched{c}{z'})
\\
    \appStSched{c}{ (z \spSeqComp z') } &\triangleq
        (\appStSched{c}{z})
        \mutseq
        (\appStSched{ (\appStMut{c}{(\appStSched{c}{z})}) }{z'})
\end{aligned}
\]
\end{minipage}
\end{wrapfigure}
The cases for $\appStSched{c}{\spConcOne{s}}$ and $\appStSched{c}{\spSeqOne{s}}$ correspond to the concurrent production and sequential production of a mutation using system $s$, as discussed in \Cref{sec:core-ecs-system}.
Recall that concurrent production observes state only once, whereas sequential production observes it before every call to $f_s$.
The cases for $\appStSched{c}{(z \spConcComp z')}$ and $\appStSched{c}{(z \spSeqComp z')}$ have a similar relationship to each other.
In the $\appStSched{c}{(z \spConcComp z')}$ case the sub-schedules are treated as concurrent; the mutation generated by $z$ is not visible in $z'$ and vice versa.
Whereas the $\appStSched{c}{(z\spSeqComp z')}$ case executes sub-schedules in sequence, such that the mutation generated by $z$ is applied to $c$ and visible to $z'$ (but not vice versa).

To advance the state of a Core ECS program:
Interpret the schedule $z$ over the current state $c$, as in $\appStSched{c}{z}$, to obtain a mutation representing the current state transition.
Next interpret that mutation $m$ over state $c$, as in $\appStMut{c}{m}$, to obtain updated state $c'$.
We define schedule application $z(c)$ below, as a shorthand for this process, by borrowing function application syntax.
$$
z(c) \triangleq \appStMut{c}{(\appStSched{c}{z})}
$$
We illustrate the interpretation of an example schedule in \Cref{fig:example-exec-schedule}.
In \Cref{sec:ecs-conc} we describe how to assess whether a schedule, such as the one in \Cref{fig:example-exec-schedule}, is deterministic.

\begin{figure}[t]
    \begin{mdframed}[backgroundcolor=gray!15,linewidth=0]
    \[
    \begin{aligned}
        z &\triangleq
            \spConcOne{\beta} \spSeqComp \spSeqOne{\delta}
    \\
        z(c) &= \appStMut{c}{(\appStSched{c}{z})}
             &&\text{Substitute definition of $z$.}
    \\
            &= \appStMut{c}{(
                \appStSched{c}{ ( \spConcOne{\beta} \spSeqComp \spSeqOne{\delta} ) }
            )}
            &&\text{Expand $\appStSched{c}{ ( - \spSeqComp - ) }$.}
    \\
            &= \appStMut{c}{(
                (\appStSched{c}{\spConcOne{\beta}})
                \mutseq
                (\appStSched{
                    (\appStMut{c}{(\appStSched{c}{\spConcOne{\beta}})})
                }{\spSeqOne{\delta}})
            )}
            &&\text{Expand two $\appStSched{c}{\spConcOne{\beta}}$.}
    \\
            &= \appStMut{c}{(
                \beta(\query{q_\beta})
                \mutseq
                (\appStSched{
                    (\appStMut{c}{\beta(\query{q_\beta})})
                }{\spSeqOne{\delta}})
            )}
            &&\text{Expand $\appStSched{-}{\spSeqOne{\delta}}$.}
    \\
            &= \appStMut{c}{(
                \beta(\query{q_\beta})
                \mutseq
                \delta(\mathrm{roll}_\delta(
                    \appStMut{c}{\beta(\query{q_\beta})},
                    \queryTwo{q_\delta}{
                        \appStMut{c}{\beta(\query{q_\beta})}
                    }
                ))
            )}
            &&\text{Expand $\appStMut{c}{(-\mutseq -)}$.}
    \\
            &=  \appStMut{
                (\appStMut{c}{\beta(\query{q_\beta})})
                }{\delta(\mathrm{roll}_\delta(
                    \appStMut{c}{\beta(\query{q_\beta})}, 
                    \queryTwo{q_\delta}{
                        \appStMut{c}{\beta(\query{q_\beta})}
                    }
                ))}
            &&\text{Substitute $c'$ for $\appStMut{c}{\beta(\query{q_\beta})}$.}
    \\
            &=  \appStMut{c'}{\delta(\mathrm{roll}_\delta(
                    c',
                    \queryTwo{q_\delta}{c'}
                ))}
            &&\text{This $c'$ is that of \Cref{fig:example-exec-collision-seq}.}
    \end{aligned}
    \]
    \end{mdframed}
\caption{
    We formalize the schedule discussed in \Cref{sec:ecs-background-phys-details} as $z$ for the systems $\beta$ and $\delta$ defined in \Cref{fig:example-system-functions}.
    We interpret $z$ by performing $z(c)$, showing how execution proceeds from \Cref{fig:example-state} to \Cref{fig:example-exec-collision-seq}.
}
\label{fig:example-exec-schedule}
\end{figure}

\section{Concurrency in ECS}
\label{sec:ecs-conc}

In this section we characterize the concurrency and determinism properties of the ECS pattern using Core ECS as a model.
We identify the elements of Core ECS that are notionally concurrent, and that under scheduler non-determinism could lead to non-deterministic results (\Cref{sec:core-ecs-conc}).
We identify a new ECS-based model for deterministic concurrency by using Core ECS to define a class of ECS programs that are deterministic despite scheduler non-determinism (\Cref{sec:det-conc}).

\begin{wrapfigure}[]{r}{0.555\textwidth}
    \vspace{-2.5em}
    \begin{subfigure}[t]{\linewidth}
        \centering
        {\small \includesvg[width=\linewidth]{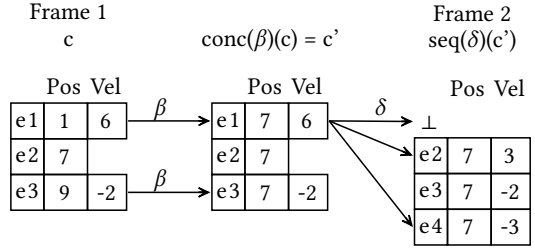} }
        \vspace{-2.5em}
        \caption{
            Schedule $\spConcOne{\beta} \spSeqComp \spSeqOne{\delta}$ calls the collision system $\delta$ only once: After $e_2$ gets a velocity it is no longer a result of $q_\delta$.
        }
        \label{fig:badphysics-sched-a}
    \end{subfigure}
    \\
    \vspace{-.25em}
    \begin{subfigure}[t]{\linewidth}
        \centering
        {\small \includesvg[width=\linewidth]{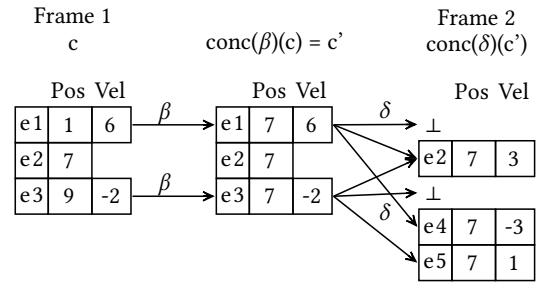} }
        \vspace{-2.5em}
        \caption{
            Schedule $\spConcOne{\beta} \spSeqComp \spConcOne{\delta}$ calls the collision system $\delta$ twice on $e_2$, causing a lost write (of $-1$) to its velocity.
        }
        \label{fig:badphysics-sched-b}
    \end{subfigure}
    \vspace{-0.5em}
    \caption{
        Comparison of schedules corresponding to the scenarios of \Cref{fig:badphysics-conc-a,fig:badphysics-conc-b}.
    }
    \label{fig:badphysics-sched}
    \vspace{-1.25em}
\end{wrapfigure}

\subsection{Concurrency Model}
\label{sec:core-ecs-conc}
Our characterization of concurrency in Core ECS rests on the question of whether one mutation is visible during the production of another.
Since mutations represent state updates, this notion of visibility establishes an ordering relationship between updates.
The design of schedules reflects this idea.
Recall from \Cref{sec:core-ecs-system,sec:core-ecs-schedule} that in the interpretation of a schedule to produce a composite mutation,
\emph{sequential} means that the effects of some constituent mutation are visible in subsequently produced mutations (whether by the same system or others).
Conversely, \emph{concurrent} means that the effects of some constituent mutation are \textbf{not} visible in subsequently produced mutations.

\Cref{fig:badphysics-sched} illustrates the distinction between sequential and concurrent schedules with the application of two schedules to a state where two moving objects are about to converge on the same stationary object (Frame 1 in \Cref{fig:badphysics-conc}).
After $\spConcOne{\beta}$ is applied to $c$, should both $e_1$ and $e_3$ collide with $e_2$?
Even though both moving objects, paired with the stationary object, are in the entity matches of collision system $\delta$ for state $c'$, under schedule $\spSeqOne{\delta}$ only one invocation of $f_\delta$ occurs because its effect is visible when the second is considered (\Cref{fig:badphysics-sched-a}).
On the other hand, under schedule $\spConcOne{\delta}$ both collisions occur, resulting in a lost write to the stationary object's velocity (\Cref{fig:badphysics-sched-b}).
Were this a real parallel implementation, the question of which of the concurrent writes would complete last would be settled by scheduler non-determinism.

We therefore analyze the concurrency of ECS programs by mimicking scheduler non-determinism, and adopt a relaxed perspective of the order in which mutations are produced under concurrent schedules $\spConcOne{-}$ and $({-}\spConcComp{-})$.
For a given system $s$, while $\spSeqOne{s}$ remains beholden to the ``consistent but unspecified total order'' of entity matches (\Cref{sec:core-ecs-matches}), the interpretation of $\spConcOne{s}$ is regarded as taking place in an arbitrary order over entity matches (intra-system concurrency).
Similarly, given sub-schedules $z$ and $z'$, under schedule ${z}\spSeqComp{z'}$ the mutations produced in $z$ are produced prior to the mutations produced by $z'$, but the interpretation of ${z}\spConcComp{z'}$ is regarded to be the same as ${z'}\spConcComp{z}$ (inter-system concurrency).

To formalize this relaxed view, we characterize the admissible orderings of system function invocations under a schedule as the linearizations of a partial order.
For schedule $z$ at state $c$ we define a partial order \(\langle \mathbf{po}_c(z),\ \preceq_z \rangle\).
Its elements are the system function invocations performed by \(z(c)\), defined inductively.
\begin{align*}
    \mathbf{po}_c(\spConcOne{s})
        &\triangleq \{(f_s, r) \mid r \in \query{q_s}\} \\
    \mathbf{po}_c(\spSeqOne{s})
        &\triangleq \{(f_s, r) \mid r \in \mathrm{roll}_s(c, \query{q_s})\} \\
    \mathbf{po}_c(z \spConcComp z')
        &\triangleq \mathbf{po}_c(z) \cup \mathbf{po}_c(z') \\
    \mathbf{po}_c(z \spSeqComp z')
        &\triangleq \mathbf{po}_c(z) \cup \mathbf{po}_{z(c)}(z')
\end{align*}
For brevity, we have treated the collection of entity matches \(\{\vec{e_i} \mapsto \vec{w_i}\}_i\) as its graph \(\{(\vec{e_i}, \vec{w_i})\}_i\).

Next we give the ordering relation among system function invocations, inductively.
\begin{align*}
    u \preceq_{\spConcOne{s}} u'
        &\triangleq u = u' \\
    (f_s, r) \preceq_{\spSeqOne{s}} (f_s, r')
        &\triangleq \text{\(r\) and \(r'\) are like-ordered in \(\mathrm{roll}_s(c, \query{q_s})\)} \\
    u \preceq_{z \spConcComp z'} u'
        &\triangleq (u \preceq_{z} u') \lor (u \preceq_{z'} u') \\
    u \preceq_{z \spSeqComp z'} u'
        &\triangleq (u \preceq_{z} u') \lor (u \preceq_{z'} u') \lor (u \in \mathbf{po}_c(z) \land u' \in \mathbf{po}_{z(c)}(z'))
\end{align*}
For system $s$, schedule $\spConcOne{s}$ imposes no dependencies among $f_s$ invocations, whereas $\spSeqOne{s}$ orders $f_s$ invocations as rolling the entity matches of $s$ would (``like-ordered''), and vacuously does not order other invocations.
Schedules formed by $({-}\spConcComp{-})$ require only that the ordering dependencies of the sub-schedules be respected, while schedules formed by $({-}\spSeqComp{-})$ additionally require that one side applies wholly after the other (by recursing on an updated state).

Every linearization of \(\mathbf{po}_c(z)\) thus respects $(-\preceq_z -)$ and fits our relaxed perspective of the order in which mutations are produced.
Furthermore, every such linearization corresponds to a composite mutation, \(f_1(r_1) \mutseq \dots \mutseq f_n(r_n)\), and we will refer to the linearization and its composite mutation interchangeably.

\subsection{Deterministic Concurrency}
\label{sec:det-conc}

We can now identify a class of concurrent ECS programs which are nonetheless deterministic under scheduler non-determinism.
Our result can be summarized imprecisely:
A system that does not obtain entities from component values and has only a singleton query vector is deterministic under schedule $\spConcOne{-}$, and two sub-schedules that write to disjoint sets of component labels are deterministic under schedule $({-}\spConcComp{-})$.
Schedules $\spSeqOne{-}$ and $({-}\spSeqComp{-})$ are always deterministic because they are not concurrent (they allow only one order of system function executions).

Our analysis is based on the observation that mutations affecting disjoint locations in ECS state \emph{commute}.
Extending this observation, a schedule for which every concurrent pair of mutations commutes is safe to execute concurrently.
To formalize this argument, we will first identify the conditions under which a schedule gives rise to determinism, and next describe how to construct such schedules.
To get started, we give a few definitions.

\begin{definition}[Mutation equivalence]
    Mutations \(m\) and \(m'\) are \emph{equivalent} (written \(m \simeq m'\)) if for all states \(c\) we can arrange that \(\appStMut{c}{m} = \appStMut{c}{m'}\) (i.e., by a correlated choice of fresh entities).\footnote{
        Since the choice of a fresh entity to satisfy $\mutfresh{-}$ is arbitrary, interpreting $\appStMut{c}{\mutfresh{\lambda\,e.\ \mutattach{i}{e}{k_i}}}$ may produce distinct result states. However, since $E$ is an opaque parameter to Core ECS, those states are indistinguishable.}
    Similarly, two linearizations of a schedule are equivalent if their corresponding mutations are.
\end{definition}

\begin{definition}[Mutation commutativity]
    Mutations \(m\) and \(m'\) are said to \emph{commute} if \(m \mutseq m' \simeq m' \mutseq m\).
\end{definition}

\begin{definition}[Schedule safety]
    A schedule \(z\) is said to be \emph{safe at state \(c\)} if for every concurrent pair $(f_s, r)$ and $(f_{s'}, r')$ in $\mathbf{po}_c(z)$ the mutations $f_s(r)$ and $f_{s'}(r')$ commute.
\label{def:safe}
\end{definition}

\begin{definition}[Schedule determinism]
    A schedule $z$ is said to be \emph{deterministic at state $c$} if all linearizations of $\mathbf{po}_c(z)$ are equivalent.
\label{def:deterministic}
\end{definition}

\begin{theorem}[Schedule safety implies schedule determinism]
    Any schedule \(z\) safe at state \(c\) is deterministic at state $c$.
\label{thm:safe-is-deterministic}
\end{theorem}
\begin{proof}
    Since we may transform one linearization into any other by a chain of order-respecting adjacent transpositions~\citep{etienne-linearizations}, it suffices to show that any two linearizations related by a single order-respecting adjacent transposition are equivalent.
    Given \(m\) and \(m'\) differing only in the order of two adjacent mutations:
    \begin{align*}
         m = f_1(r_1) \mutseq \dots \mutseq f_a(r_a) &\mutseq f_b(r_b) \mutseq \dots \mutseq f_n(r_n) \\
        m' = f_1(r_1) \mutseq \dots \mutseq f_b(r_b) &\mutseq f_a(r_a) \mutseq \dots \mutseq f_n(r_n)
    \end{align*}
    Since both \(m\) and \(m'\) respect \(\mathbf{po}_c(z)\) but are different, it must be that \((f_a, r_a)\) and \((f_b, r_b)\) are concurrent.
    Given that and the assumption that \(z\) is safe at state \(c\), we know that \(f_a(r_a)\) and \(f_b(r_b)\) commute.
    Therefore, \(m \simeq m'\).
\end{proof}

Since any safe schedule is deterministic by \Cref{thm:safe-is-deterministic}, we now turn to a discussion of how to construct safe schedules.
We will require a mechanism to examine which entities and components are affected by a schedule at some state in order to show that it is safe.
For this we define the \emph{influence} of a mutation $m$, $\mathrm{infl}(m)$, or of a schedule $z$ at some state $c$, $\mathrm{infl}_c(z)$, to be the subset of \(E \times I\) pairs necessarily affected by it.
In the case of a mutation, the complement of \(\mathrm{infl(m)}\) is the set of entities and components (as indexes into state) for which the state is unchanged, i.e. $(\appStMut{c}{m})_i(e) = c_i(e)$ holds for all $c$.
Influence is defined inductively for both mutations (left) and schedules (right):
\\ 
\begin{minipage}{0.5\linewidth}
    \begin{alignat*}{2}
        \mathrm{infl}(\mutattach{i}{e}{k_i})
        &\triangleq \{(e, i)\}
    \\
        \mathrm{infl}(\mutdetach{i}{e})
        &\triangleq \{(e, i)\}
    \\
        \mathrm{infl}(\mathrm{m \mutseq m'})
        &\triangleq \mathrm{infl}(m) \cup \mathrm{infl}(m')
    \\
        \mathrm{infl}(\mutfresh{f})
        &\triangleq \bigcap_{e \in E} \mathrm{infl}(f(e))
    \\
        \mathrm{infl}(\mutnil)
        &\triangleq \emptyset
    \end{alignat*}
\end{minipage}
\begin{minipage}{0.5\linewidth}
    \begin{alignat*}{2}
        \mathrm{infl}_c(\spConcOne{s})
        &\triangleq \bigcup_{r \in \query{q_s}} \mathrm{infl}(f_s(r))
    \\
        \mathrm{infl}_c(\spSeqOne{s})
        &\triangleq \bigcup_{r \in \mathrm{roll}_s(c, \query{q_s})} \mathrm{infl}(f_s(r))
    \\
        \mathrm{infl}_c(z \spConcComp z')
        &\triangleq \mathrm{infl}_c(z) \cup \mathrm{infl}_c(z')
    \\
        \mathrm{infl}_c(z \spSeqComp z')
        &\triangleq \mathrm{infl}_c(z) \cup \mathrm{infl}_{z(c)}(z')
    \end{alignat*}
\end{minipage}
\\ 
\\ 
The influence of a mutation (or a schedule) is the union of those entity and component label pairs that it affects --- a ``memory footprint'' --- except in the case of $\mutfresh{-}$.
Since the interpretation $\appStMut{c}{\mutfresh{f}}$ calls $f$ on a fresh entity not present in $c$, the influence on that entity and its components is not visible to any concurrent mutation, so we discard that influence using intersection.
Of the influence generated by $f$, only that invariant to the choice of fresh entity is preserved.

With influence we can now make statements about when mutations or schedules commute, which will enable us to state rules of safe construction.
We do so in the following lemma and theorem:

\begin{lemma}
    \label{thm:infl-comm}

    For mutations $m$ and $m'$,
    if \(\mathrm{infl}(m) \cap \mathrm{infl}(m') = \emptyset\)
    then $m$ and $m'$ commute.
\end{lemma}
\begin{proof}
    For mutations not involving \(\nu(-)\), this is immediate: such mutations are compositions of attach, detach, and nil, which will commute because they affect distinct entity-component indices.

    For mutations involving \(\nu(-)\) we arrange for disjoint fresh entities to be provided.
    Then the mutations will behave as though those provided entities were inlined within them. They reduce to mutations not involving \(\nu(-)\) whose influences remain disjoint, and commute as described above.
\end{proof}

\begin{theorem}
    \label{thm:infl-safe-schedule}
    Construction of safe schedules by cases is as follows:
    \begin{enumerate}
        \item The schedule \(\spSeqOne{s}\) is safe at all states for any system \(s\).
        \item The schedule \(z \spSeqComp z'\) is safe at state \(c\) if \(z\) is safe at \(c\) and \(z'\) is safe at \(z(c)\).
        \item The schedule \(\spConcOne{s}\) is safe at state \(c\) if \(\mathrm{infl}(f_s(r)) \cap \mathrm{infl}(f_s(r')) = \emptyset\) for every distinct pair of entity matches \(r\) and \(r'\) in \(\query{q_s}\).
        \item The schedule \(z \spConcComp z'\) is safe at state \(c\) if \(z\) and \(z'\) are safe at \(c\) and \(\mathrm{infl}_c(z) \cap \mathrm{infl}_c(z') = \emptyset\).
    \end{enumerate}
\end{theorem}
\begin{proof}
    By cases:
    \begin{enumerate}
        \item Since \(\spSeqOne{s}\) has no concurrent pairs of system function invocations, it is vacuously safe.
        \item Any two concurrent invocations in \(z \spSeqComp z'\) are either both in \(z\) or both in \(z'\), which are safe at their respective states, and since there are no other pairs to consider, \(z \spSeqComp z'\) is safe at $c$.
        \item Since the pairs of concurrent invocations in $\spConcOne{s}$ yield pairs of mutations that satisfy \Cref{thm:infl-comm}, \(\spConcOne{s}\) is safe at $c$.
        \item
            Any two concurrent invocations in ${z}\spConcComp{z'}$ that are both in $z$ or both in $z'$ are safe because $z$ and $z'$ are safe at $c$.
            For pairs of concurrent invocations that straddle $z$ and $z'$, since the influences of $z$ and $z'$ are disjoint, every pair  of invocations yields a pair of mutations that satisfy \Cref{thm:infl-comm}.
            Therefore \(z \spConcComp z'\) is safe at $c$.
        \qedhere
    \end{enumerate}
\end{proof}

Any schedule constructed by the rules in \Cref{thm:infl-safe-schedule} is safe, and therefore by \Cref{thm:safe-is-deterministic} is deterministic.
However, in practice, the full force of \Cref{thm:infl-safe-schedule} requires understanding the influence of each system under all possible runtime inputs, which can be a nontrivial analysis.
We close this section with two corollaries of \Cref{thm:infl-safe-schedule} that can be judged statically.

\begin{corollary}
    \label{coro:det-conc-sched}
    If schedules \(z\) and \(z'\) are safe and influence disjoint sets of component labels, then \(z \spConcComp z'\) is safe (and thus deterministic) on any state $c$.
\end{corollary}

\begin{corollary}
    \label{coro:det-conc-one}
    For a system \(s\), if \(q_s = \langle q_1 \rangle\) is a singleton-vector and the component types of \(\queryT{q_1}\) do not reference the entity type \(E\), then \(\spConcOne{s}\) is safe (and thus deterministic) on any state $c$.
\end{corollary}

\Cref{coro:det-conc-one} says a system that only ever influences one entity per invocation, run concurrently with itself, won't influence the same entity in two of the invocations.
By restricting the components, no entity enters a system invocation via component values.
By restricting the query vector to one query, a distinct entity enters each system invocation.
This is possible to judge statically by examining the the fixed query vector for a system and the component types in the schema $K$.
Use of a system in this way is a parallel map!

Meanwhile \Cref{coro:det-conc-sched} says that if two sub-schedules work over entirely disjoint sets of component labels, then it doesn't matter whether they influence the same entities --- their influences must necessarily be disjoint.
A benefit of the ECS pattern is that an existing ECS program can be easily extended with new systems and components in response to design changes, and so this scenario arises frequently in practice.
Indeed, it is recognizably the act of running distinct tasks on distinct regions of memory.

It may seem obvious that if concurrent tasks are restricted to only modify disjoint regions of memory, we will have deterministic concurrency; however, as we find in \Cref{sec:practical-ecs}, none of the ECS implementations we studied fully implement even this degree of concurrency.
Our results establish a kind of ``green light'', go-ahead signal for concurrency in ECS: there is no theoretical obstacle to achieving this degree of concurrency, only practical obstacles dependent on implementation decisions.

\section{ECS in Practice}
\label{sec:practical-ecs}
In this section we survey five prominent open-source ECS frameworks listed in \Cref{fig:ecs-frameworks}, reflecting on Core ECS for comparison.
We focus on ECS frameworks with open-source implementations that prioritize efficient execution, and biased our selection toward widely-used frameworks.
In \Cref{sec:implementations}, we discuss the five frameworks' approaches to two key implementation decisions: their approach to fresh entity generation, and their choice of component storage strategy.
With these implementation decisions in view, \Cref{sec:impl-allowed-conc} turns to a comparison of the frameworks' support for concurrency.
Our purpose in discussing the implementation strategies at length in \Cref{sec:implementations} is to inform the reader about the prioritizations and biases, common in ECS frameworks, that directly lead to the constraints on concurrency support discussed in \Cref{sec:impl-allowed-conc}.

Our overall finding is that, compared to Core ECS, all five of the frameworks we surveyed \emph{disallow} a subset of the safe schedules that \Cref{thm:infl-safe-schedule} identifies.
That is, in all five frameworks there is a gap between the degree of deterministic concurrency that the ECS pattern (as modeled by Core ECS) \emph{can} support, and the degree of deterministic concurrency that ECS frameworks \emph{do} support.
In particular, none of the five frameworks support the ability to attach a new component to an entity concurrently with another mutation.

The five frameworks we survey (in order of popularity\footnote{As indicated by GitHub stars, measured 2024 October 9.}) are:
\emph{Bevy ECS}~\citep{bevy2024ecs}, a Rust ECS framework developed for the Bevy game engine;
\emph{EnTT}~\citep{entt2024ecs}, an ECS framework for C++, known for its use in Minecraft~\citep{minecraft2024attrib};
\emph{Flecs}~\citep{mertens2024flecs}, an ECS framework for C and C++;
\emph{Specs}~\citep{schaller2023specs}, a Rust ECS framework developed for the Amethyst game engine~\citep{amethyst2021kalderon};
and \emph{apecs}~\citep{carpay2018apecs}, an ECS framework for Haskell.
These frameworks variously bill themselves as ``fast''~\citep{carpay2018apecs} (or ``[i]ncredibly fast''~\citep{entt2024ecs}),  ``massively parallel''~\citep{bevy2024ecs}, supporting multithreading with a ``fast lockless scheduler''~\citep{mertens2024flecs}, and offering ``easy parallelism'' and ``high performance''~\citep{schaller2023specs} -- claims that existing ECS benchmarking efforts have sought to validate~\citep{lschmierer2017benchmark,abeimler2024benchmark}.
Our purpose in this section, however, is not a quantitative performance assessment, but rather a qualitative comparison of design and implementation decisions made in the interest of performance.

\subsection{Implementation Techniques}
\label{sec:implementations}

We guide our exploration of implementation techniques with two questions that often drive the design of an ECS framework.
How are fresh entity identifiers generated?
What strategy is used to store components?
(One further question, how a programmer specifies the system schedule, is explored by \Cref{sec:impl-allowed-conc}.)
The way an implementation answers these questions directly influences the degree to which it can realize concurrency with parallelism (often due to the introduction mutual exclusion).
In practice the answers to these questions tend to be correlated; certain strategies for storage implicate certain strategies for indexing that storage with entity identifiers.
Moreover, the desire to minimize indirection, speeding up iteration over (and updates to) the entity-component association, is in tension with the desire to maximize efficient querying and structure changes, aspects which may benefit from indirection.

\begin{table}
    \caption{
        Summary of the ECS frameworks we surveyed.
        The ``Stars'' column is a count of GitHub stars.
        ``Entity generation'' refers to frameworks' fresh entity generation approach, as discussed in \Cref{sec:implementations}.
        ``Strategy'' refers to frameworks' component storage strategy, as discussed in \Cref{sec:implementations}.
        ``Scheduling'' refers to frameworks' system scheduling interface, as discussed in \Cref{sec:impl-allowed-conc}.
}
    \begin{tabular}{ l r l l l l }
          {\bf Framework}
        & {\bf Stars}
        & {\bf Language}
        & {\bf Entity generation}
        & {\bf Strategy}
        & {\bf Scheduling}
        \\ \hline\hline

        Bevy ECS
        & 35,6xx
        & Rust
        & Generational indexing
        & Archetype
        & Semi-automated
        \\

        EnTT
        & 10,1xx
        & C++
        & Generational indexing
        & Columnar
        & Manual\tablefootnote{
            ``In general, the entire registry isn't thread safe as it is. Thread safety isn't something that users should want out of the box for several reasons. Just to mention one of them: performance.''~\citep{caini2024multithreading}
        }
        \\

        Flecs
        & 6,4xx
        & C
        & Generational indexing
        & Archetype
        & Either\tablefootnote{
            Flecs works with or without its ``pipeline'' scheduler~\citep{mertens2025scheduler}, supporting both Semi-automated or Manual uses.
        }
        \\

        Specs
        & 2,5xx
        & Rust
        & Generational indexing
        & Columnar
        & Semi-automated
        \\

        apecs
        & 392
        & Haskell
        & Sequential numeric
        & Columnar
        & Manual\tablefootnote{
            Apecs removed concurrency support with the update from v0.4.1.1 to v0.5.0.0~\citep{carpay0.4.1.1apecsconc,carpay0.5.0.0apecs}.
            The removed functions \texttt{pmap} and \texttt{concurrency} provided intra- and inter- system concurrency, respectively. ``Provides zero protection against race conditions and other hazards, so use with caution.''~\citep{carpay0.4.1.1apecsconc}
        }
        \\

        \hline
    \end{tabular}
    \label{fig:ecs-frameworks}
\end{table}

\paragraph{Fresh entity generation}
As discussed in \Cref{sec:core-ecs}, to use the ECS pattern in some language, a programmer designates a type to represent entity identifiers, and implements some means to generate fresh elements of that type.
There are a few obvious implementation-level approaches to this, such as the use of \emph{sequential numeric identifiers}, or generating \emph{random numbers} with many bits.
While these approaches are not widely used in practice, it is helpful to think through their trade-offs before we consider the more commonly used technique, called \emph{generational indexing}.
What is the state required to generate entities?
What coordination is necessary to ensure that concurrently generated entities are unique?

Using sequential numerical identifiers to generate entities, the approach employed by apecs~\citep{carpay0.9.6nextEntity}, is reminiscent of an auto-incrementing primary key in a database table.
A fresh entity identifier is obtained by incrementing a shared global variable.
This strategy requires a minimum of state and only modest coordination when generating new entities:
    A multithreaded framework may use an atomic integer to generate distinct entities concurrently.

Random identifiers require no state nor coordination to ensure that concurrently generated entities are unique.
A fresh entity identifier is obtained by generating a random number with enough bits of entropy to minimize the possibility of a collision.
This approach eliminates state and coordination at the expense of necessitating a larger identifier size.

Generational indexing is a more sophisticated strategy in which an entity identifier consists of an index part and a generation part~\citep{west2018using,slucas2021generational}.
Flecs, Specs, EnTT, and Bevy ECS all generate fresh entities using generational indexing~\citep{mertens2020entity,schaller2023specs,caini2024entity,bevy2024entity}.
In this approach, a fresh entity identifier is either produced by recycling a previously used index or minting a new one.
This strategy requires more state and coordination than the other strategies we mention, but enables higher occupancy in a smaller allocation:
Array indexes generated this way help to reuse the allocated space in arrays (as compared with sequential numerical identifiers) by reducing the frequency of array re-allocations that require global coordination.
As such, generational indexing presents another trade off, moving coordination from those arrays (component storage) to fresh entity generation.

\paragraph{Component storage strategies}

\begin{wrapfigure}[]{R}{0.5\textwidth}
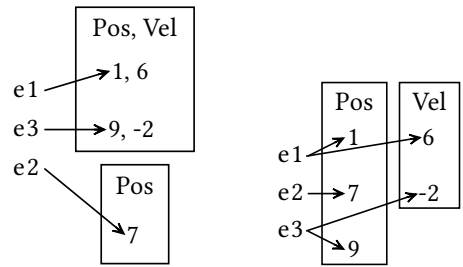

    \vspace{-1.5em}
    \begin{subfigure}[t]{0.48\linewidth}
        \centering
        \includesvg[width=0.8\linewidth]{storage-archetype.gv.svg}
        \caption{
            Archetypes are ad hoc regions filled with entities having like-components.
        }
        \label{fig:storage-archetype}
    \end{subfigure}
    \hfill
    \begin{subfigure}[t]{0.48\linewidth}
        \centering
        \includesvg[width=0.8\linewidth]{storage-columnar.gv.svg}
        \caption{
            Columnar frameworks maintain components separately by label.
        }
        \label{fig:storage-columnar}
    \end{subfigure}
    \vspace{-0.25em}
    \caption{
        The two main component storage strategies.
    }
    \label{fig:storage}
    \vspace{-1.0em}
\end{wrapfigure}

We observe two main strategies for the storage of components, calling one \emph{columnar} and the other \emph{archetype} (as seen in \citet{gillen2021legion}).

Archetype-style frameworks  logically group entities by the set of component labels that they have, and place the component values for each logical group (each \emph{archetype}) in the same region of memory.
For example, in \Cref{fig:storage-archetype}, $(\mathbf{Pos}, \mathbf{Vel})$ is an archetype for entities with both a $\mathbf{Pos}$ and a $\mathbf{Vel}$ component.
Flecs and Bevy ECS are archetype-style frameworks~\citep{mertens2020archetype,bevy2024archetype}.
This strategy facilitates the execution of queries, because examination of the component labels in each archetype immediately establishes the groups of entities that match the query.
This strategy may also facilitate cache locality in a tight loop over those entities, if the component values for each entity are stored together.
However, the cost of moving component data between regions when an entity's set of components changes may be a disadvantage.

Columnar-style frameworks group component values with the same component label together, enabling distinct storage styles for distinct component labels (\Cref{fig:storage-columnar}).
EnTT, Specs, and apecs are columnar-style frameworks~\citep{caini2024column,schaller2023specs,carpay0.9.6makeWorld}.
This strategy minimizes the cost of updating the entity-component association by requiring changes in only one logical column; however, it may suffer from poor locality when accessing many component values of a single entity in a tight loop over entities.
There are many concrete column storage styles and we list only a few here for flavor.
\begin{itemize}[leftmargin=1.5em]
    \item
        \textbf{Array} ---
        By obtaining an index from an entity identifier, component data can be stored in an array.
        A sentinel value or bit-mask may indicate which indexes are in use (i.e. which entities have the given component).
        Unused indexes contribute to fragmentary memory use.
        Paired with generational indexing, array reallocation and fragmentation can be minimized.

    \item
        \textbf{Dense array} ---
        Adding indirection between entity identifiers and component array indexes (e.g. using a sparse-set structure) may alleviate the problem of fragmentary memory use and reduce the preference for generational indexing.

    \item
        \textbf{Map} ---
        Choosing an off-the-shelf hash-map or tree-map abstracts the problem of storage completely from the strategy for indexing, at potential cost of iteration speed.

\end{itemize}
The many trade-offs between component storage are also influenced by the use cases of different applications, and thus extend beyond the scope of our purpose here. 

\subsection{Practically Available Concurrency}
\label{sec:impl-allowed-conc}

To explore the concurrency available in the ECS frameworks of \Cref{fig:ecs-frameworks}, we compare the expressiveness of system scheduling in those frameworks with that of Core ECS schedules (\Cref{sec:core-ecs-schedule}).
We find that there are broadly two approaches to system scheduling interfaces in the ECS frameworks we examine: \emph{manual} and \emph{semi-automated}.
Our examples in this section will use the inertia ($\beta$) and collision ($\delta$) systems from our toy physics simulation example in \Cref{sec:ecs-background,sec:core-ecs}, and for demonstration purposes we throw in two more systems: a rotation system (which we call $\gamma$) and a render system (which we call $\eta$).

\paragraph{Manual Scheduling}
Frameworks we labeled ``manual'' lack a first-class concept of scheduling systems.
Instead the programmer manually schedules systems by directly invoking them on ECS state in the desired order, typically in the context of an outer ``main loop''.
Each system function expresses a query over the ECS state and iterates sequentially over entity matches, using facilities provided either by the ECS framework or by the host language.
Among the frameworks we surveyed, this is the approach taken by
apecs~\citep{carpay2019schedule} (which removed its support for concurrency), EnTT~\citep{kernick2020schedule} (for which documentation implies that concurrency is not an intended use case), and Flecs~\citep{mertens2025scheduler} (when used without its ``pipeline'' scheduler or a parallel job system).
Such lightweight scheduling may be straightforward and incur little overhead, but it places the burden of scheduling, and thus of correctly managing concurrency, if any, on the programmer.
In practice, this amounts to sequential schedules composed of $\spSeqOne{-}$ and $({-}\spSeqComp{-})$, as in the following Core ECS schedule at left, and corresponding hypothetical EnTT schedule at right.
\vspace{0.3em}
\\ 
\begin{minipage}{0.5\linewidth}
    $$\begin{aligned}
         \spSeqOne{\beta}  \spSeqComp
         \spSeqOne{\gamma} \spSeqComp
         \spSeqOne{\delta} \spSeqComp
         \spSeqOne{\eta}
    \end{aligned}$$
\end{minipage}
\begin{minipage}{0.5\linewidth}
    \begin{mdframed}[backgroundcolor=gray!15,linewidth=0]
        \begin{small}
        \begin{verbatim}
inertiaSystem(state);
rotationSystem(state);
collisionSystem(state);
renderSystem(state);
\end{verbatim}
        \end{small}
    \end{mdframed}
\end{minipage}

\paragraph{Semi-Automated Scheduling}
Frameworks labeled ``semi-automated'' are varied, but one commonality is that system functions are invoked by a framework main loop that the programmer configures.
Inter-system concurrency may be realized with parallelism, or by running those systems in an order determined by the scheduler.
Specs, Bevy ECS, and Flecs (with its ``pipeline'' scheduler)
offer this kind of semi-automated scheduling.
Such schedules may consist of a sequential spine of $({-}\spSeqComp{-})$ that orders concurrent groups of systems composed with $({-}\spConcComp{-})$.
We evoke this pattern with the following Core ECS schedule at left, and corresponding hypothetical Bevy ECS schedule at right.
\vspace{0.3em}
\\ 
\begin{minipage}{0.5\linewidth}
    $$\begin{aligned}
         &\left(\spConcOne{\beta} \spConcComp \spSeqOne{\gamma}\right) \spSeqComp
       \\&\spSeqOne{\delta} \spSeqComp
       \\&\spSeqOne{\eta}
    \end{aligned}$$
\end{minipage}
\begin{minipage}{0.5\linewidth}
    \begin{mdframed}[backgroundcolor=gray!15,linewidth=0]
        \begin{small}
        \begin{verbatim}
let mut z = Schedule::default();
z.add_systems((
    (inertia, rotation),
    collision,
    render
  ).chain());
\end{verbatim}
    \end{small}
    \end{mdframed}
\end{minipage}

In both scheduling interfaces, the option to use intra-system concurrency is usually internalized to the implementation of those systems rather than being made explicit when they are scheduled.
For example, it is possible to achieve intra-system concurrency, as in $\spConcOne{\beta}$, via parallel iteration.
One common idiom is a framework-provided parallel-map function that takes a function to call once per entity match.
We show an example of this idiom as it appears in Bevy ECS in \Cref{fig:bevy-inertia}.
This idiom also highlights a difference between Core ECS and the frameworks we investigated: Core ECS relocates intra-system concurrency from the system to the schedule.

\begin{figure}
    \begin{mdframed}[backgroundcolor=gray!15,linewidth=0]
        \begin{small}
            \begin{verbatim}
fn inertia(mut q: Query<(&mut Pos, &Vel)>) {
  q.par_iter_mut().for_each(|(mut p, Vel(v))| { p.0 += v; });
}
\end{verbatim}
        \end{small}
    \end{mdframed}
    \caption{
        A Bevy ECS inertia system, expressed as a Rust function, that demonstrates intra-system concurrency, as in $\spConcOne{\beta}$, and exclusive access to (logical) stores.
    }
    \label{fig:bevy-inertia}
\end{figure}

Continuing with the example in \Cref{fig:bevy-inertia}, we observe that support for concurrency in the scheduling interfaces of Flecs (with ``pipeline''), Specs, and Bevy ECS are all designed to avoid \emph{write conflicts}.
For the two Rust frameworks, Bevy ECS and Specs, this avoidance is achieved by shallowly embedding the problem into the borrow checker:
In \Cref{fig:bevy-inertia} the position store is mutably borrowed by the inertia system in \texttt{\&mut Pos}, and the velocity store is immutably borrowed by \texttt{\&Vel}.
No concurrent system may borrow the position store, and no concurrent system may mutably borrow the velocity store.
The documentation for Bevy ECS explicitly addresses this restriction:
``Not all systems can run together: if a system mutably accesses data, no other system that reads or writes that data can be run at the same time. These systems are said to be \textbf{incompatible}.''~\citep{bevy2024systemcompat}

Furthermore, in all three of the multithreaded frameworks we investigated (Flecs, Specs, and Bevy ECS), it is not possible to attach a new component to an entity in any concurrent setting.
Flecs documentation provides us with a name for this problem: a \emph{structure change} is altering the set of components attached to an entity.
In the Flecs authors' words, ``By default systems are ran while the world is in `readonly' mode, where all ECS operations are enqueued as commands. Readonly here means that structural changes, such as changing the components of an entity are deferred.''~\citep{mertens2024structurechange}

All three of the multithreaded ECS frameworks in \Cref{fig:ecs-frameworks} provide the same two workarounds for the restrictions imposed to avoid write conflicts and concurrent structure changes.
\begin{enumerate}
    \item Deferring modifications to entities, by storing them in a buffer to be applied sequentially by the ECS framework later (called \emph{lazy updates} by Specs~\citep{schaller2022lazyupdate} and \emph{parallel commands} by Bevy ECS~\citep{bevy2024parallelcommands}) allows any modification to be expressed in a concurrent context.
    \item Running a system in a single thread with exclusive access to ECS state (called an \emph{immediate system} in Flecs and an \emph{exclusive system} in Bevy ECS) allows it to make any modifications in any component store.
\end{enumerate}
Of these, only deferred modifications (1) occur in concurrent contexts, and so we ignore exclusive systems (2) in further discussion.

The well-intentioned attempts by the authors of these frameworks to provide concurrency free of data races results in a zoo of mutation categories, which we catalog in \Cref{fig:concurrency-slice-and-dice}.
With each mutation category we provide an example Core ECS system which, if translated to Flecs, Specs, or Bevy ECS, would demonstrate a distinct variety of mutation in that framework (though some categories do not have a translation to some frameworks).
As in \Cref{fig:bevy-inertia}, a component that appears in the query vector of these example systems can be regarded as owned (or borrowed mutably) in a translation.
We ask the reader to consider whether and how each example may be translated.

\begin{table}
    \caption{
        Categories of mutations having implications for concurrent execution in the multithreaded ECS frameworks we studied.
        Those that would influence the same component store that they query are regarded as owning that store.
        When a store is not owned, influence against it must be deferred.
        Insertions and deletions are both structure changes which, respectively, potentially allocate or are assumed to not allocate.
    }
    \begin{tabular}{ l l l }
          {\bf Mutation Category}
        & {\bf Example System}
        & {\bf Side condition}
        \\ \hline\hline

        Owned update
        & $\langle\qincl{\mathbf{Pos}}\rangle,\ \lambda (e,\ p).\;\mutattach{\mathbf{Pos}}{e}{0}$
        \\

        Owned insert
        & $\langle\qexcl{\mathbf{Pos}}\rangle,\ \lambda (e,\ p).\;\mutattach{\mathbf{Pos}}{e}{0}$
        & 
        \\

        Owned initialize
        & $\langle\qanyway{\mathbf{Pos}}\rangle,\ \lambda (e,\ p).\;\mutfresh{\lambda e'.\;\mutattach{\mathbf{Pos}}{e'}{0}}$
        & 
        \\

        Owned delete
        & $\langle\qincl{\mathbf{Pos}}\rangle,\ \lambda (e,\ p).\;\mutdetach{\mathbf{Pos}}{e}$
        \\

        Deferred update
        & $\langle\qanyway{\mathbf{Pos}}\rangle,\ \lambda (e,\ p).\;\mutattach{\mathbf{Vel}}{e}{1}$
        & $e$ has $\mathbf{Vel}$
        \\

        Deferred insert
        & $\langle\qanyway{\mathbf{Pos}}\rangle,\ \lambda (e,\ p).\;\mutattach{\mathbf{Vel}}{e}{1}$
        & $e$ does not have $\mathbf{Vel}$
        \\

        Deferred initialize
        & $\langle\qanyway{\mathbf{Pos}}\rangle,\ \lambda (e,\ p).\;\mutfresh{\lambda e'.\;\mutattach{\mathbf{Vel}}{e'}{1}}$
        & 
        \\

        Deferred delete
        & $\langle\qanyway{\mathbf{Pos}}\rangle,\ \lambda (e,\ p).\;\mutdetach{\mathbf{Vel}}{e}$
        & $e$ has $\mathbf{Vel}$
        \\

        \hline
    \end{tabular}
    \label{fig:concurrency-slice-and-dice}
\end{table}

In the multithreaded ECS frameworks that we studied, only the first mutation category in \Cref{fig:concurrency-slice-and-dice} (Owned update) is possible to execute concurrently.
The four ``deferred'' mutation categories are nominally parallel according to the frameworks' documentation, but factually they are accumulated in a buffer and applied to state serially after system execution completes.
The remaining three ``owned'' mutation categories (Owned insert, Owned initialize, and Owned delete) are, in concurrent contexts, only possible to express as the corresponding deferred structure changes.
These owned structure changes may cause allocation in a component store (insert, initialize) or in fresh entity structures (initialize).

If the premises of safe schedule construction in \Cref{thm:infl-safe-schedule} are observed, any of the mutation categories in \Cref{fig:concurrency-slice-and-dice} are safe to execute concurrently, and will produce deterministic behavior.
In principle, then, it would seem that frameworks would let an ECS program express any safe schedule, and in \Cref{coro:det-conc-sched,coro:det-conc-one} we identify two subsets of safe schedules that can be checked statically.
It is clear that a structure change involving one component, executed concurrently with mutations at different components, is deterministic because --- as in the frame rule in a concurrent separation logic --- these actions deal with disjoint regions of memory.
Yet none of the multithreaded ECS frameworks we investigate fully support \Cref{coro:det-conc-sched} by allowing concurrent structure changes.
These practical frameworks disallow obviously correct forms of concurrency, reflecting a bias toward domains where the structure of entities does not change frequently.

\section{An Executable Model}
\label{sec:executable-model-new}
We have implemented an \emph{executable model} of Core ECS as described in \Cref{sec:core-ecs}, and this section describes it briefly.
The executable model is intended to facilitate experimentation with Core ECS and exploration of the behavior of Core ECS programs; it should not be considered a substitute for a user-ready, production-quality ECS implementation, such as those discussed in \Cref{sec:practical-ecs}.
We have included a discussion of several example ECS programs written against our executable model in
\ifdefined\ARXIVVERSION
\Cref{apx:executable-examples}.
\fi
\ifdefined\PACMPLVERSION
\citet[Appendix B]{coreecs-extended}.
\fi
The full implementation of the executable model and these example programs are in our accompanying artifact~\citep{coreecsartifact}.

The executable model faithfully represents the expressiveness of Core ECS with negligible differences.
The executable model is written in Haskell, and the underlying programming language used to write system functions for it is also Haskell.
The type of entities, $E$, is fixed to Haskell's \texttt{Int}.
The component schema, $K$, is given as a type-level list of component types; each type is interpreted as a component label referring to the component values that inhabit that type.
Accordingly, queries are entirely represented at the type level, and it is recommended to use newtypes to give meaningful handles to components.

Recall that in \Cref{sec:implementations} we discuss how conventionally-chosen data structures for entity allocation and component storage both drive a need for mutual exclusion, and in \Cref{sec:impl-allowed-conc} we discuss approaches used by practical ECS frameworks for scheduling systems.
Now let us briefly consider the executable model in those terms, against those frameworks in \Cref{fig:ecs-frameworks}:
\begin{itemize}[leftmargin=1.5em]
    \item
        \textbf{Entity generation} ---
        The executable model uses \emph{sequential numeric} identifiers and does not reuse identifiers from deleted entities.
    \item
        \textbf{Strategy} ---
        The executable model component storage is columnar.
        In particular it uses a two-level mapping:
        An association-list relates component types to stores (the ``columns'').
        A store is a Haskell \texttt{IntMap} optionally relating each entity (identified by an \texttt{Int}) to a component value of the store type.
    \item
        \textbf{Scheduling} ---
        The interface of the executable model follows Core ECS, which is neither ``manual'' nor ``semi-automated''.
        A programmer writes a schedule that specifies which parts of the ECS program should be concurrent, and that schedule may be interpreted in different ways. The executable model provides two schedule interpreter functions, described below.
\end{itemize}

Our executable model includes two reference implementations of schedule application.
The first implementation interprets a schedule exactly as the function $\appStSched{c}{z}$ defined in \Cref{sec:core-ecs-schedule} does, for some state $c$ and schedule $z$.
This reference implementation fixes the ``consistent but unspecified total order'' (\Cref{sec:core-ecs-matches}) of entity matches to the lexicographic order of entity vectors.
That is, in this reference implementation, concurrency is \emph{not} implemented with parallelism, but with an arbitrary ordering of entity matches.

By contrast, the second implementation of schedule application uses Haskell threads and a shared mutable reference to the state.
Every mutation produced by an invocation of a system function is immediately applied to the state under mutual exclusion.
By implementing concurrency with parallelism, we achieve a relaxed ordering of mutations described in \Cref{sec:core-ecs-conc}, dependent on scheduler non-determinism.

With either implementation of schedule application, the executable model will be deterministic exactly when the schedule is safe, according to \Cref{thm:infl-safe-schedule}.

\section{Related Work}
\label{sec:related}
\paragraph{Relational database management systems}
It is commonly observed that the ECS pattern is in effect a very narrow and regimented use of a relational database management system (RDBMS)~\citep{bilas2002data,martin2007entitypart3,mertens2023databases,gutekanst2022databases,borisova2024bevydata}.
The entity-component association can be encoded directly with a single RDBMS table using the ``entity attribute'' schema.
A columnar layout can be encoded by storing each component type in a separate single-column table with a foreign key that references an entity table.
Through use of dynamically created multi-column tables, an archetype layout can also be encoded.
Given any such layout, ECS systems are possible to encode with a select query followed by an update query, possibly with an embedded third-party programming language in between, all persisted via stored procedures.
Finally, triggers could be set up to activate the stored procedures based on their select queries.

Despite these subsuming similarities, the historical development of the ECS pattern, on underpowered hardware in which all memory is allocated to the ECS program, emphasizes that it is distinct from RDBMS~\citep{bilas2002data,west2018using}.
Historically and presently, ECS frameworks are designed for efficiency first, and this focus leads to a few differences from RDBMSs:
The ECS pattern is not concerned with persistent storage.
It is common for the entity-component association to be represented plainly as a struct of arrays~\citep{west2018using,mertens2024flecs,schaller2023specs}.
Furthermore, the ECS pattern focuses on a small schedule of queries that are fixed during runtime, and so sidesteps the whole question of query optimization characteristic of RDBMSes.

\paragraph{Deterministic concurrent programming models}
\label{sec:related-detconc}

There is a very long tradition of work on abstractions for deterministic concurrent programming~\cite{Tesler-1968,Kahn-1974,IStructures,dph,dpj-oopsla,flowpools,lvars}.
In general, deterministic concurrent programming models must somehow restrict access to mutable shared state.
Our determinism result in \Cref{sec:det-conc} depends on the fact that concurrent tasks only access \emph{disjoint} state.
This approach to deterministic concurrency is similar to that taken by Deterministic Parallel Java~\citep{dpj-oopsla, dpj-popl}.
We are not aware of any other work on determinism of ECS programs specifically, but given the disjointness condition, our determinism result is conceptually straightforward.
Abstractions such as LVars~\citep{lvars, lvish}, on the other hand, do allow concurrent tasks to access overlapping state in arbitrary order, but retain determinism by carefully restricting \emph{how} the state can be updated and queried.
In future work, it would be interesting to investigate how the ECS programming model could be combined with the approach taken by deterministic concurrency abstractions that allow some degree of (well-behaved) overlap in the memory footprint of tasks, rather than the total disjointness that we currently consider.

\paragraph{Multiset Rewriting}
Rewriting logics are an established family of computational logics characterized by a set of rules $R$ which apply to elements of an equational theory $(\Sigma, E)$~\citep{meseguer2012twenty}.
With appropriate constructors $\Sigma$ and algebraic identities $E$, one obtains a multiset rewriting logic~\citep{meseguer1990logical}.
Rewriting logics may be seen to represent either computation or deduction.
Similarly, multiset rewriting may be seen to represent concurrent or distributed computation.
With the very general framework of multiset rewriting, it is possible to encode the ECS pattern.
ECS state can be encoded as a set (further constraining the multiset) of 3-tuples, each containing an entity identifier, a component label, and a component value.
ECS systems can be encoded as rules, which selectively apply themselves only to those tuples matching an encoding of their query semantics.
To specify a schedule by which these rules are applied, it is necessary to instead package them into a single rule encoding the desired sequencing and concurrent application of the system sub-rules.
There is quite a lot of machinery necessary to add atop a multiset rewriting logic to achieve the semantics of the ECS pattern as described by Core ECS, however, the overall shape of multiset rewriting is a generalization of what Core ECS expresses.

\paragraph{Join calculus}
The join calculus of \citet{fournet1996reflexive,fournet2000join} is a combination of a minimal ML-inspired language with ``processes'' that may be run concurrently with other processes.
Running a process $P$ concurrently with another $Q$ is written with a parallel composition operator as in $P \mid Q$, similar to our notion of concurrent composition of Core ECS schedules.
Processes are not necessarily required to return a value and hence are not ML expressions, much like Core ECS schedules do not return a value and are not expressions in the underlying programming language.
Via a syntax sugar over $\mathbf{let}$ expressions, a side-effecting ML expression $E$ may be sequenced with either an expression or a process $P$, as in $E\mathbin{;}P$, which is similar to our notion of sequential composition of schedules, except that it interlaces ML expressions with processes.

Indeed, the join calculus allows process abstractions to be defined and run within expressions, and functions containing ML expressions to be defined and run within processes.
This sort of interlacing of the two domains is unlike Core ECS, which maintains a schedule at the top level, with expressions of the underlying language appearing only at the leaves within systems.

Furthermore, processes may be called from ML expressions and return results to ML expressions, communicating across these boundaries via ``channels''.
The join calculus additionally defines a variety of ``pattern-matching'' ``inter-process synchronization'' primitives.
These are unlike Core ECS in that schedules neither receive input from expressions, nor return results, nor synchronize with each other --- all communication between systems in Core ECS is in the form of updates to the entity-component association.

\section{Conclusion and Future Work}
\label{sec:conclusion}
\label{sec:future-work}

We have presented Core ECS, a formal model for the ECS software design pattern that abstracts away from the implementation details of specific ECS frameworks to reveal what we believe are the essential characteristics of the ECS pattern.
We precisely characterized concurrency in the ECS pattern using Core ECS as its model, we provided rules of construction for well-behaved concurrent ECS programs, and proved that those programs are deterministic --- invocations of systems that mutate disjoint state are safe to execute concurrently.
Our determinism result suggests that the ECS pattern can be viewed as a general deterministic concurrent programming model.
While it is unsurprising that concurrent mutations of disjoint state give rise to determinism, through our result we identified that only a small part of the deterministic concurrency available in the ECS pattern is also available in commonly used ECS frameworks.
By identifying where ECS frameworks restrict concurrency unnecessarily, our result can guide the design of new ECS frameworks that do not require such restrictions, opening up opportunities for efficient and correct concurrent programming.
There are several possible directions for future work; here, we highlight two of them:

\paragraph{Quantitative evaluation of alternative parallel ECS implementation techniques}
Through our comparison of deterministic concurrency in Core ECS with that available in multithreaded ECS frameworks (\Cref{sec:impl-allowed-conc}), we identified forms of concurrency not served by the common existing implementation techniques for the ECS pattern.
While it is notable that frameworks' avoidance of write conflicts prevents two concurrent systems from writing to the same component store at disjoint sets of entities, 
we find it more surprising that the avoidance of concurrent structure changes prevents two concurrent systems from adding distinct components to entities, or from creating new entities with any components.
We are therefore very interested to investigate alternative implementation techniques that rely less on global structures (such as those required for generational indexing), rely less on correlated data (such as bit-masks that are often paired with arrays in columnar component storage), and in general avoid mutual exclusion.

We believe a significant runtime advantage may be available via a technique that eschews fanatical prioritization of throughput, and instead amortizes costs, reduces contention, and consequently maximizes available concurrency.
The ECS pattern has the potential to be a new concurrent programming model or even a compiler target.
We propose to implement several existing and new techniques and measure their wall-clock time and relative speedup with a suite of benchmarks more broad than traditional simulations lacking in structure changes.

\paragraph{Query expressivity.}
The querying capabilities we describe in \Cref{sec:core-ecs-queries} are intentionally minimal, to better focus on the essence of the ECS pattern.
However, this minimalism is quite limiting along a few distinct dimensions.
Most obviously, taking the cartesian product of \(n\) queries will result in \(O(x^n)\) system invocations, which scales poorly with an increasing number of queries.
In many cases a system may only do useful work on a very small number of entity matches, such as in a kinematics simulation~\citep{barnes1986hierarchical} that only concerns points that are sufficiently close.
Extending Core ECS with a more expressive query language, such as non-cartesian joins or filters over individual queries, may alleviate this issue.
Alternatively, allowing a system to perform a series of queries, each informed by the previous, may provide a completely different solution.

\newpage

\section*{Data-Availability Statement}

We have made an artifact available at Zenodo that includes an executable model of Core ECS as well as examples of its use and comparable examples written against Bevy ECS \cite{coreecsartifact}.
Our executable model is described in \Cref{sec:executable-model-new} and explored in some detail in
\ifdefined\ARXIVVERSION
\Cref{apx:executable-examples}.
\fi
\ifdefined\PACMPLVERSION
\citet[Appendix B]{coreecs-extended}.
\fi

\begin{acks}
This material is based upon work supported by the National Science Foundation under Grant No. 2145367. 
Any opinions, findings, and conclusions or recommendations expressed in this material 
are those of the author(s) and do not necessarily reflect the views of the National Science Foundation.
\end{acks}

\bibliographystyle{ACM-Reference-Format}
\bibliography{references}


\begin{thebibliography}{58}


\ifx \showCODEN    \undefined \def \showCODEN     #1{\unskip}     \fi
\ifx \showISBNx    \undefined \def \showISBNx     #1{\unskip}     \fi
\ifx \showISBNxiii \undefined \def \showISBNxiii  #1{\unskip}     \fi
\ifx \showISSN     \undefined \def \showISSN      #1{\unskip}     \fi
\ifx \showLCCN     \undefined \def \showLCCN      #1{\unskip}     \fi
\ifx \shownote     \undefined \def \shownote      #1{#1}          \fi
\ifx \showarticletitle \undefined \def \showarticletitle #1{#1}   \fi
\ifx \showURL      \undefined \def \showURL       {\relax}        \fi
\providecommand\bibfield[2]{#2}
\providecommand\bibinfo[2]{#2}
\providecommand\natexlab[1]{#1}
\providecommand\showeprint[2][]{arXiv:#2}

\bibitem[Anderson and {Bevy Contributors}(2024)]%
        {bevy2024ecs}
\bibfield{author}{\bibinfo{person}{Carter Anderson} {and}
  \bibinfo{person}{{Bevy Contributors}}.} \bibinfo{year}{2024}\natexlab{}.
\newblock \bibinfo{title}{Bevy Engine}.
\newblock
\newblock
\shownote{\url{https://bevyengine.org/}. Releases
  \url{https://docs.rs/bevy_ecs/latest/bevy_ecs/}. Accessed 2024-10-09}.


\bibitem[Arvind et~al\mbox{.}(1989)]%
        {IStructures}
\bibfield{author}{\bibinfo{person}{Arvind}, \bibinfo{person}{Rishiyur~S.
  Nikhil}, {and} \bibinfo{person}{Keshav~K. Pingali}.}
  \bibinfo{year}{1989}\natexlab{}.
\newblock \showarticletitle{I-structures: data structures for parallel
  computing}.
\newblock \bibinfo{journal}{\emph{{ACM Trans. Program. Lang. Syst.}}}
  \bibinfo{volume}{11}, \bibinfo{number}{4} (\bibinfo{date}{Oct.}
  \bibinfo{year}{1989}).
\newblock


\bibitem[Barnes and Hut(1986)]%
        {barnes1986hierarchical}
\bibfield{author}{\bibinfo{person}{Josh Barnes} {and} \bibinfo{person}{Piet
  Hut}.} \bibinfo{year}{1986}\natexlab{}.
\newblock \showarticletitle{A hierarchical O(N log N) force-calculation
  algorithm}.
\newblock \bibinfo{journal}{\emph{Nature}} \bibinfo{volume}{324},
  \bibinfo{number}{6096} (\bibinfo{year}{1986}), \bibinfo{pages}{446--449}.
\newblock


\bibitem[Beimler(2024)]%
        {abeimler2024benchmark}
\bibfield{author}{\bibinfo{person}{Alex Beimler}.}
  \bibinfo{year}{2024}\natexlab{}.
\newblock \bibinfo{title}{Entity-Component-System Benchmarks}.
\newblock
\newblock
\shownote{\url{https://github.com/abeimler/ecs_benchmark}. Accessed 2024-10-9}.


\bibitem[{Bevy Contributors}(2024a)]%
        {bevy2024archetype}
\bibfield{author}{\bibinfo{person}{{Bevy Contributors}}.}
  \bibinfo{year}{2024}\natexlab{a}.
\newblock \bibinfo{title}{bevy\_ecs::archetype - Rust}.
\newblock
\newblock
\shownote{Documentation
  \url{https://docs.rs/bevy_ecs/latest/bevy_ecs/archetype/index.html}. Accessed
  2024-10-15}.


\bibitem[{Bevy Contributors}(2024b)]%
        {bevy2024entity}
\bibfield{author}{\bibinfo{person}{{Bevy Contributors}}.}
  \bibinfo{year}{2024}\natexlab{b}.
\newblock \bibinfo{title}{Entity in bevy\_ecs::entity - Rust}.
\newblock
\newblock
\shownote{Documentation
  \url{https://docs.rs/bevy_ecs/latest/bevy_ecs/entity/struct.Entity.html}.
  Accessed 2024-10-15}.


\bibitem[{Bevy Contributors}(2024c)]%
        {bevy2024systemcompat}
\bibfield{author}{\bibinfo{person}{{Bevy Contributors}}.}
  \bibinfo{year}{2024}\natexlab{c}.
\newblock \bibinfo{title}{{Module bevy\_ecs::system}}.
\newblock
\newblock
\shownote{Documentation
  \url{https://docs.rs/bevy_ecs/0.14.2/bevy_ecs/system/index.html\#system-ordering}.
  Accessed 2025-03-16}.


\bibitem[{Bevy Contributors}(2024d)]%
        {bevy2024parallelcommands}
\bibfield{author}{\bibinfo{person}{{Bevy Contributors}}.}
  \bibinfo{year}{2024}\natexlab{d}.
\newblock \bibinfo{title}{{ParallelCommands in bevy\_ecs::system}}.
\newblock
\newblock
\shownote{Documentation
  \url{https://docs.rs/bevy_ecs/0.14.2/bevy_ecs/system/struct.ParallelCommands.html}.
  Accessed 2024-10-16}.


\bibitem[Bilas(2002)]%
        {bilas2002data}
\bibfield{author}{\bibinfo{person}{Scott Bilas}.}
  \bibinfo{year}{2002}\natexlab{}.
\newblock \bibinfo{title}{A Data-Driven Game Object System}.
\newblock
\newblock
\shownote{Presentation
  \url{https://www.gamedevs.org/uploads/data-driven-game-object-system.pdf}.
  Video \url{https://www.youtube.com/watch?v=Eb4-0M2a9xE}. Audio
  \url{https://www.gdcvault.com/play/1022543/A-Data-Driven-Object}. Accessed
  2024-04-04}.


\bibitem[Bocchino et~al\mbox{.}(2009)]%
        {dpj-oopsla}
\bibfield{author}{\bibinfo{person}{Robert~L. Bocchino},
  \bibinfo{person}{Vikram~S. Adve}, \bibinfo{person}{Danny Dig},
  \bibinfo{person}{Sarita~V. Adve}, \bibinfo{person}{Stephen Heumann},
  \bibinfo{person}{Rakesh Komuravelli}, \bibinfo{person}{Jeffrey Overbey},
  \bibinfo{person}{Patrick Simmons}, \bibinfo{person}{Hyojin Sung}, {and}
  \bibinfo{person}{Mohsen Vakilian}.} \bibinfo{year}{2009}\natexlab{}.
\newblock \showarticletitle{A type and effect system for deterministic parallel
  Java}. In \bibinfo{booktitle}{\emph{Proceedings of the 24th ACM SIGPLAN
  Conference on Object Oriented Programming Systems Languages and
  Applications}} (Orlando, Florida, USA) \emph{(\bibinfo{series}{OOPSLA '09})}.
  \bibinfo{publisher}{Association for Computing Machinery},
  \bibinfo{address}{New York, NY, USA}, \bibinfo{pages}{97–116}.
\newblock
\showISBNx{9781605587660}
\href{https://doi.org/10.1145/1640089.1640097}{doi:\nolinkurl{10.1145/1640089.1640097}}


\bibitem[Bocchino et~al\mbox{.}(2011)]%
        {dpj-popl}
\bibfield{author}{\bibinfo{person}{Robert~L. Bocchino},
  \bibinfo{person}{Stephen Heumann}, \bibinfo{person}{Nima Honarmand},
  \bibinfo{person}{Sarita~V. Adve}, \bibinfo{person}{Vikram~S. Adve},
  \bibinfo{person}{Adam Welc}, {and} \bibinfo{person}{Tatiana Shpeisman}.}
  \bibinfo{year}{2011}\natexlab{}.
\newblock \showarticletitle{Safe nondeterminism in a deterministic-by-default
  parallel language}. In \bibinfo{booktitle}{\emph{Proceedings of the 38th
  Annual ACM SIGPLAN-SIGACT Symposium on Principles of Programming Languages}}
  (Austin, Texas, USA) \emph{(\bibinfo{series}{POPL '11})}.
  \bibinfo{publisher}{Association for Computing Machinery},
  \bibinfo{address}{New York, NY, USA}, \bibinfo{pages}{535–548}.
\newblock
\showISBNx{9781450304900}
\href{https://doi.org/10.1145/1926385.1926447}{doi:\nolinkurl{10.1145/1926385.1926447}}


\bibitem[Borisova(2024)]%
        {borisova2024bevydata}
\bibfield{author}{\bibinfo{person}{Ida Borisova}.}
  \bibinfo{year}{2024}\natexlab{}.
\newblock \bibinfo{title}{Unofficial Bevy Cheat Book -- Intro: Your Data}.
\newblock
\newblock
\shownote{\url{https://bevy-cheatbook.github.io/programming/intro-data.html}.
  Living
  \url{https://github.com/bevy-cheatbook/bevy-cheatbook/blob/main/src/programming/intro-data.md}.
  Accessed 2024-10-15}.


\bibitem[Caini(2024a)]%
        {caini2024multithreading}
\bibfield{author}{\bibinfo{person}{Michele Caini}.}
  \bibinfo{year}{2024}\natexlab{a}.
\newblock \bibinfo{title}{Crash Course: entity component system}.
\newblock
\newblock
\shownote{Heading ``Multithreading''.
  \url{https://github.com/skypjack/entt/wiki/Crash-Course:-entity-component-system/aa053854f18d4589fe7ce6284752dfdbe1d5d3b0\#multithreading}.
  Accessed 2024-10-09}.


\bibitem[Caini(2024b)]%
        {caini2024entity}
\bibfield{author}{\bibinfo{person}{Michele Caini}.}
  \bibinfo{year}{2024}\natexlab{b}.
\newblock \bibinfo{title}{Crash Course: entity component system}.
\newblock
\newblock
\shownote{Heading ``The Registry, the Entity and the Component''.
  \url{https://github.com/skypjack/entt/wiki/Crash-Course:-entity-component-system/aa053854f18d4589fe7ce6284752dfdbe1d5d3b0\#the-registry-the-entity-and-the-component}.
  Accessed 2024-10-09}.


\bibitem[Caini(2024c)]%
        {caini2024column}
\bibfield{author}{\bibinfo{person}{Michele Caini}.}
  \bibinfo{year}{2024}\natexlab{c}.
\newblock \bibinfo{title}{Crash Course: entity component system}.
\newblock
\newblock
\shownote{Heading ``All or nothing''.
  \url{https://github.com/skypjack/entt/wiki/Crash-Course:-entity-component-system/465d90e0f5961adc460cd9d1e9358370987fbcd3\#all-or-nothing}.
  Accessed 2024-10-15}.


\bibitem[Caini(2024d)]%
        {entt2024ecs}
\bibfield{author}{\bibinfo{person}{Michele Caini}.}
  \bibinfo{year}{2024}\natexlab{d}.
\newblock \bibinfo{title}{EnTT: Gaming meets Modern C++}.
\newblock
\newblock
\shownote{\url{https://skypjack.github.io/entt/}. Living
  \url{https://github.com/skypjack/entt}. Accessed 2024-10-15}.


\bibitem[Carpay(2018a)]%
        {carpay2018apecs}
\bibfield{author}{\bibinfo{person}{Jonas Carpay}.}
  \bibinfo{year}{2018}\natexlab{a}.
\newblock \bibinfo{title}{Apecs: A Type-Driven Entity-Component-System
  Framework}.
\newblock
\newblock
\shownote{Preprint
  \url{https://github.com/jonascarpay/apecs/blob/master/apecs/prepub.pdf}.
  Accessed 2024-10-09}.


\bibitem[Carpay(2018b)]%
        {carpay0.5.0.0apecs}
\bibfield{author}{\bibinfo{person}{Jonas Carpay}.}
  \bibinfo{year}{2018}\natexlab{b}.
\newblock \bibinfo{title}{apecs: Fast ECS framework for game programming
  (v0.5.0.0)}.
\newblock
\newblock
\shownote{Release \url{https://hackage.haskell.org/package/apecs-0.5.0.0}.
  Accessed 2024-10-09}.


\bibitem[Carpay(2018c)]%
        {carpay0.4.1.1apecsconc}
\bibfield{author}{\bibinfo{person}{Jonas Carpay}.}
  \bibinfo{year}{2018}\natexlab{c}.
\newblock \bibinfo{title}{Apecs.Concurrent (v.0.4.1.1)}.
\newblock
\newblock
\shownote{Module
  \url{https://hackage.haskell.org/package/apecs-0.4.1.1/docs/Apecs-Concurrent.html}.
  Accessed 2024-10-09}.


\bibitem[Carpay(2019)]%
        {carpay2019schedule}
\bibfield{author}{\bibinfo{person}{Jonas Carpay}.}
  \bibinfo{year}{2019}\natexlab{}.
\newblock \bibinfo{title}{Apecs Shmup example}.
\newblock
\newblock
\shownote{Source
  \url{https://github.com/jonascarpay/apecs/blob/be72edab17eaaf9f771e45392b3837d31d8ff663/examples/Shmup.lhs\#L338-L350}.
  Accessed 2024-10-16}.


\bibitem[Carpay(2024a)]%
        {carpay0.9.6makeWorld}
\bibfield{author}{\bibinfo{person}{Jonas Carpay}.}
  \bibinfo{year}{2024}\natexlab{a}.
\newblock \bibinfo{title}{Apecs.TH, makeWorld (v0.9.6)}.
\newblock
\newblock
\shownote{Source
  \url{https://hackage.haskell.org/package/apecs-0.9.6/docs/Apecs-TH.html\#v:makeWorld}.
  Accessed 2024-10-15}.


\bibitem[Carpay(2024b)]%
        {carpay0.9.6nextEntity}
\bibfield{author}{\bibinfo{person}{Jonas Carpay}.}
  \bibinfo{year}{2024}\natexlab{b}.
\newblock \bibinfo{title}{Apecs.Util, nextEntity (v0.9.6)}.
\newblock
\newblock
\shownote{Source
  \url{https://hackage.haskell.org/package/apecs-0.9.6/docs/src/Apecs.Util.html\#nextEntity}.
  Accessed 2024-10-09}.


\bibitem[Etienne(1984)]%
        {etienne-linearizations}
\bibfield{author}{\bibinfo{person}{Gwihen Etienne}.}
  \bibinfo{year}{1984}\natexlab{}.
\newblock \showarticletitle{Linear extensions of finite posets and a conjecture
  of G. Kreweras on permutations}.
\newblock \bibinfo{journal}{\emph{Discrete Math.}} \bibinfo{volume}{52},
  \bibinfo{number}{1} (\bibinfo{date}{March} \bibinfo{year}{1984}),
  \bibinfo{pages}{107–111}.
\newblock
\showISSN{0012-365X}
\href{https://doi.org/10.1016/0012-365X(84)90108-0}{doi:\nolinkurl{10.1016/0012-365X(84)90108-0}}


\bibitem[Fournet and Gonthier(1996)]%
        {fournet1996reflexive}
\bibfield{author}{\bibinfo{person}{C\'{e}dric Fournet} {and}
  \bibinfo{person}{Georges Gonthier}.} \bibinfo{year}{1996}\natexlab{}.
\newblock \showarticletitle{The reflexive {CHAM} and the join-calculus}. In
  \bibinfo{booktitle}{\emph{Proceedings of the 23rd ACM SIGPLAN-SIGACT
  Symposium on Principles of Programming Languages}} (St. Petersburg Beach,
  Florida, USA) \emph{(\bibinfo{series}{POPL '96})}.
  \bibinfo{publisher}{Association for Computing Machinery},
  \bibinfo{address}{New York, NY, USA}, \bibinfo{pages}{372–385}.
\newblock
\showISBNx{0897917693}
\href{https://doi.org/10.1145/237721.237805}{doi:\nolinkurl{10.1145/237721.237805}}


\bibitem[Fournet and Gonthier(2000)]%
        {fournet2000join}
\bibfield{author}{\bibinfo{person}{C{\'e}dric Fournet} {and}
  \bibinfo{person}{Georges Gonthier}.} \bibinfo{year}{2000}\natexlab{}.
\newblock \showarticletitle{The join calculus: A language for distributed
  mobile programming}.
\newblock In \bibinfo{booktitle}{\emph{International Summer School on Applied
  Semantics}}. \bibinfo{publisher}{Springer}, \bibinfo{pages}{268--332}.
\newblock


\bibitem[Gillen(2021)]%
        {gillen2021legion}
\bibfield{author}{\bibinfo{person}{Thomas Gillen}.}
  \bibinfo{year}{2021}\natexlab{}.
\newblock \bibinfo{title}{Legion: High performance entity component system
  (ECS) library}.
\newblock
\newblock
\shownote{Releases \url{https://crates.io/crates/legion}. Living
  \url{https://github.com/amethyst/legion}. Accessed 2024-04-04}.


\bibitem[Gutekanst(2022)]%
        {gutekanst2022databases}
\bibfield{author}{\bibinfo{person}{Stephen Gutekanst}.}
  \bibinfo{year}{2022}\natexlab{}.
\newblock \bibinfo{title}{Let's build an Entity Component System (part 2):
  databases}.
\newblock
\newblock
\shownote{Blog
  \url{https://devlog.hexops.com/2022/lets-build-ecs-part-2-databases/}.
  Accessed 2024-10-15}.


\bibitem[Hatledal et~al\mbox{.}(2021)]%
        {hatledal-vico}
\bibfield{author}{\bibinfo{person}{Lars~I. Hatledal},
  \bibinfo{person}{Yingguang Chu}, \bibinfo{person}{Arne Styve}, {and}
  \bibinfo{person}{Houxiang Zhang}.} \bibinfo{year}{2021}\natexlab{}.
\newblock \showarticletitle{Vico: An entity-component-system based
  co-simulation framework}.
\newblock \bibinfo{journal}{\emph{Simulation Modelling Practice and Theory}}
  \bibinfo{volume}{108} (\bibinfo{year}{2021}), \bibinfo{pages}{102243}.
\newblock
\showISSN{1569-190X}
\href{https://doi.org/10.1016/j.simpat.2020.102243}{doi:\nolinkurl{10.1016/j.simpat.2020.102243}}


\bibitem[Kahn(1974)]%
        {Kahn-1974}
\bibfield{author}{\bibinfo{person}{G. Kahn}.} \bibinfo{year}{1974}\natexlab{}.
\newblock \showarticletitle{The Semantics of a Simple Language for Parallel
  Programming}.
\newblock In \bibinfo{booktitle}{\emph{Information Processing '74: Proceedings
  of the {IFIP} Congress}}, \bibfield{editor}{\bibinfo{person}{J.~L.
  Rosenfeld}} (Ed.). \bibinfo{publisher}{North-Holland}.
\newblock


\bibitem[Kalderon(2021)]%
        {amethyst2021kalderon}
\bibfield{author}{\bibinfo{person}{Eyal Kalderon}.}
  \bibinfo{year}{2021}\natexlab{}.
\newblock \bibinfo{title}{Amethyst Game Engine}.
\newblock
\newblock
\shownote{Releases \url{https://crates.io/crates/amethyst}. Living
  \url{https://github.com/amethyst/amethyst}. Accessed 2024-10-09}.


\bibitem[Kernick(2020)]%
        {kernick2020schedule}
\bibfield{author}{\bibinfo{person}{Indiana Kernick}.}
  \bibinfo{year}{2020}\natexlab{}.
\newblock \bibinfo{title}{EnTT Pacman}.
\newblock
\newblock
\shownote{Source
  \url{https://github.com/indianakernick/EnTT-Pacman/blob/8d0ad586decc80aebe1431456b092507be26b372/src/core/game.cpp/\#L79-L89}.
  Accessed 2024-10-16}.


\bibitem[Kiselyov and Ishii(2015)]%
        {kiselyov2015freer}
\bibfield{author}{\bibinfo{person}{Oleg Kiselyov} {and} \bibinfo{person}{Hiromi
  Ishii}.} \bibinfo{year}{2015}\natexlab{}.
\newblock \showarticletitle{Freer monads, more extensible effects}. In
  \bibinfo{booktitle}{\emph{Proceedings of the 2015 ACM SIGPLAN Symposium on
  Haskell}} (Vancouver, BC, Canada) \emph{(\bibinfo{series}{Haskell '15})}.
  \bibinfo{publisher}{Association for Computing Machinery},
  \bibinfo{address}{New York, NY, USA}, \bibinfo{pages}{94–105}.
\newblock
\showISBNx{9781450338080}
\href{https://doi.org/10.1145/2804302.2804319}{doi:\nolinkurl{10.1145/2804302.2804319}}


\bibitem[Kuper and Newton(2013)]%
        {lvars}
\bibfield{author}{\bibinfo{person}{Lindsey Kuper} {and}
  \bibinfo{person}{Ryan~R. Newton}.} \bibinfo{year}{2013}\natexlab{}.
\newblock \showarticletitle{LVars: Lattice-based Data Structures for
  Deterministic Parallelism}. In \bibinfo{booktitle}{\emph{Proceedings of the
  2Nd ACM SIGPLAN Workshop on Functional High-performance Computing}} (Boston,
  Massachusetts, USA) \emph{(\bibinfo{series}{FHPC '13})}.
  \bibinfo{publisher}{ACM}, \bibinfo{address}{New York, NY, USA},
  \bibinfo{pages}{71--84}.
\newblock
\showISBNx{978-1-4503-2381-9}
\href{https://doi.org/10.1145/2502323.2502326}{doi:\nolinkurl{10.1145/2502323.2502326}}


\bibitem[Kuper et~al\mbox{.}(2014)]%
        {lvish}
\bibfield{author}{\bibinfo{person}{Lindsey Kuper}, \bibinfo{person}{Aaron
  Turon}, \bibinfo{person}{Neelakantan~R. Krishnaswami}, {and}
  \bibinfo{person}{Ryan~R. Newton}.} \bibinfo{year}{2014}\natexlab{}.
\newblock \showarticletitle{Freeze After Writing: Quasi-deterministic Parallel
  Programming with LVars}. In \bibinfo{booktitle}{\emph{Proceedings of the 41st
  ACM SIGPLAN-SIGACT Symposium on Principles of Programming Languages}} (San
  Diego, California, USA) \emph{(\bibinfo{series}{POPL '14})}.
  \bibinfo{publisher}{ACM}, \bibinfo{address}{New York, NY, USA},
  \bibinfo{pages}{257--270}.
\newblock
\showISBNx{978-1-4503-2544-8}
\href{https://doi.org/10.1145/2535838.2535842}{doi:\nolinkurl{10.1145/2535838.2535842}}


\bibitem[Martin(2007a)]%
        {martin2007entity}
\bibfield{author}{\bibinfo{person}{Adam Martin}.}
  \bibinfo{year}{2007}\natexlab{a}.
\newblock \bibinfo{title}{Entity Systems are the future of MMOG development -
  Part 1}.
\newblock
\newblock
\shownote{Blog
  \url{https://t-machine.org/index.php/2007/09/03/entity-systems-are-the-future-of-mmog-development-part-1/}.
  Accessed 2024-04-04}.


\bibitem[Martin(2007b)]%
        {martin2007entitypart3}
\bibfield{author}{\bibinfo{person}{Adam Martin}.}
  \bibinfo{year}{2007}\natexlab{b}.
\newblock \bibinfo{title}{Entity Systems are the future of MMOG development -
  Part 3}.
\newblock
\newblock
\shownote{Blog
  \url{https://t-machine.org/index.php/2007/12/22/entity-systems-are-the-future-of-mmog-development-part-3/}.
  Accessed 2024-10-15}.


\bibitem[Mertens(2020a)]%
        {mertens2020archetype}
\bibfield{author}{\bibinfo{person}{Sander Mertens}.}
  \bibinfo{year}{2020}\natexlab{a}.
\newblock \bibinfo{title}{Building an ECS \#2: Archetypes and Vectorization}.
\newblock
\newblock
\shownote{Blog
  \url{https://ajmmertens.medium.com/building-an-ecs-2-archetypes-and-vectorization-fe21690805f9}.
  Accessed 2024-10-15}.


\bibitem[Mertens(2020b)]%
        {mertens2020entity}
\bibfield{author}{\bibinfo{person}{Sander Mertens}.}
  \bibinfo{year}{2020}\natexlab{b}.
\newblock \bibinfo{title}{Making the most of ECS identifiers}.
\newblock
\newblock
\shownote{Blog
  \url{https://ajmmertens.medium.com/doing-a-lot-with-a-little-ecs-identifiers-25a72bd2647}.
  Accessed 2024-10-09}.


\bibitem[Mertens(2023)]%
        {mertens2023databases}
\bibfield{author}{\bibinfo{person}{Sander Mertens}.}
  \bibinfo{year}{2023}\natexlab{}.
\newblock \bibinfo{title}{Why it is time to start thinking of games as
  databases}.
\newblock
\newblock
\shownote{Blog
  \url{https://ajmmertens.medium.com/why-it-is-time-to-start-thinking-of-games-as-databases-e7971da33ac3}.
  Accessed 2024-10-15}.


\bibitem[Mertens(2024a)]%
        {mertens2024flecs}
\bibfield{author}{\bibinfo{person}{Sander Mertens}.}
  \bibinfo{year}{2024}\natexlab{a}.
\newblock \bibinfo{title}{Flecs: A fast entity component system (ECS) for C \&
  C++}.
\newblock
\newblock
\shownote{\url{https://www.flecs.dev/flecs/}. Living
  \url{https://github.com/SanderMertens/flecs/}. Accessed 2024-04-04}.


\bibitem[Mertens(2024b)]%
        {mertens2024structurechange}
\bibfield{author}{\bibinfo{person}{Sander Mertens}.}
  \bibinfo{year}{2024}\natexlab{b}.
\newblock \bibinfo{title}{Flecs: Systems}.
\newblock
\newblock
\shownote{\url{https://www.flecs.dev/flecs/md_docs_2Systems.html\#immediate-systems}.
  Accessed 2024-10-16}.


\bibitem[Mertens(2025)]%
        {mertens2025scheduler}
\bibfield{author}{\bibinfo{person}{Sander Mertens}.}
  \bibinfo{year}{2025}\natexlab{}.
\newblock \bibinfo{title}{Flecs: FAQ "Can I use my own scheduler
  implementation?"}.
\newblock
\newblock
\shownote{\url{https://www.flecs.dev/flecs/md_docs_2FAQ.html\#can-i-use-my-own-scheduler-implementation}.
  Accessed 2025-08-03}.


\bibitem[Meseguer(1990)]%
        {meseguer1990logical}
\bibfield{author}{\bibinfo{person}{Jos\'{e} Meseguer}.}
  \bibinfo{year}{1990}\natexlab{}.
\newblock \showarticletitle{A logical theory of concurrent objects}. In
  \bibinfo{booktitle}{\emph{Proceedings of the European Conference on
  Object-Oriented Programming on Object-Oriented Programming Systems,
  Languages, and Applications}} (Ottawa, Canada)
  \emph{(\bibinfo{series}{OOPSLA/ECOOP '90})}. \bibinfo{publisher}{Association
  for Computing Machinery}, \bibinfo{address}{New York, NY, USA},
  \bibinfo{pages}{101–115}.
\newblock
\showISBNx{0897914112}
\href{https://doi.org/10.1145/97945.97958}{doi:\nolinkurl{10.1145/97945.97958}}


\bibitem[Meseguer(2012)]%
        {meseguer2012twenty}
\bibfield{author}{\bibinfo{person}{Jos{\'e} Meseguer}.}
  \bibinfo{year}{2012}\natexlab{}.
\newblock \showarticletitle{Twenty years of rewriting logic}.
\newblock \bibinfo{journal}{\emph{The Journal of Logic and Algebraic
  Programming}} \bibinfo{volume}{81}, \bibinfo{number}{7-8}
  (\bibinfo{year}{2012}), \bibinfo{pages}{721--781}.
\newblock


\bibitem[{Mojang AB} and {Microsoft Corporation}(2024)]%
        {minecraft2024attrib}
\bibfield{author}{\bibinfo{person}{{Mojang AB}} {and}
  \bibinfo{person}{{Microsoft Corporation}}.} \bibinfo{year}{2024}\natexlab{}.
\newblock \bibinfo{title}{Minecraft Attributions}.
\newblock
\newblock
\shownote{\url{https://www.minecraft.net/en-us/attribution}. Accessed
  2024-10-9}.


\bibitem[Peyton~Jones et~al\mbox{.}(2008)]%
        {dph}
\bibfield{author}{\bibinfo{person}{Simon Peyton~Jones}, \bibinfo{person}{Roman
  Leshchinskiy}, \bibinfo{person}{Gabriele Keller}, {and}
  \bibinfo{person}{Manuel M~T Chakravarty}.} \bibinfo{year}{2008}\natexlab{}.
\newblock \showarticletitle{{Harnessing the Multicores: Nested Data Parallelism
  in Haskell}}. In \bibinfo{booktitle}{\emph{IARCS Annual Conference on
  Foundations of Software Technology and Theoretical Computer Science}}
  \emph{(\bibinfo{series}{Leibniz International Proceedings in Informatics
  (LIPIcs)}, Vol.~\bibinfo{volume}{2})},
  \bibfield{editor}{\bibinfo{person}{Ramesh Hariharan},
  \bibinfo{person}{Madhavan Mukund}, {and} \bibinfo{person}{V~Vinay}} (Eds.).
  \bibinfo{publisher}{Schloss Dagstuhl -- Leibniz-Zentrum f{\"u}r Informatik},
  \bibinfo{address}{Dagstuhl, Germany}, \bibinfo{pages}{383--414}.
\newblock
\showISBNx{978-3-939897-08-8}
\showISSN{1868-8969}
\href{https://doi.org/10.4230/LIPIcs.FSTTCS.2008.1769}{doi:\nolinkurl{10.4230/LIPIcs.FSTTCS.2008.1769}}


\bibitem[Prokopec et~al\mbox{.}(2013)]%
        {flowpools}
\bibfield{author}{\bibinfo{person}{Aleksandar Prokopec},
  \bibinfo{person}{Heather Miller}, \bibinfo{person}{Tobias Schlatter},
  \bibinfo{person}{Philipp Haller}, {and} \bibinfo{person}{Martin Odersky}.}
  \bibinfo{year}{2013}\natexlab{}.
\newblock \showarticletitle{FlowPools: A Lock-Free Deterministic Concurrent
  Dataflow Abstraction}. In \bibinfo{booktitle}{\emph{Languages and Compilers
  for Parallel Computing}}, \bibfield{editor}{\bibinfo{person}{Hironori
  Kasahara} {and} \bibinfo{person}{Keiji Kimura}} (Eds.).
  \bibinfo{publisher}{Springer Berlin Heidelberg}, \bibinfo{address}{Berlin,
  Heidelberg}, \bibinfo{pages}{158--173}.
\newblock
\showISBNx{978-3-642-37658-0}


\bibitem[Raffaillac and Huot(2019)]%
        {raffaillac-polyphony}
\bibfield{author}{\bibinfo{person}{Thibault Raffaillac} {and}
  \bibinfo{person}{St\'{e}phane Huot}.} \bibinfo{year}{2019}\natexlab{}.
\newblock \showarticletitle{Polyphony: Programming Interfaces and Interactions
  with the Entity-Component-System Model}.
\newblock \bibinfo{journal}{\emph{Proc. ACM Hum.-Comput. Interact.}}
  \bibinfo{volume}{3}, \bibinfo{number}{EICS}, Article \bibinfo{articleno}{8}
  (\bibinfo{date}{jun} \bibinfo{year}{2019}), \bibinfo{numpages}{22}~pages.
\newblock
\href{https://doi.org/10.1145/3331150}{doi:\nolinkurl{10.1145/3331150}}


\bibitem[Redmond et~al\mbox{.}(2025)]%
        {coreecsartifact}
\bibfield{author}{\bibinfo{person}{Patrick Redmond}, \bibinfo{person}{Jonathan
  Castello}, \bibinfo{person}{José~Manuel Calderón~Trilla}, {and}
  \bibinfo{person}{Lindsey Kuper}.} \bibinfo{year}{2025}\natexlab{}.
\newblock \bibinfo{booktitle}{\emph{Exploring the Theory and Practice of
  Concurrency in the Entity-Component-System Pattern (Artifact)}}.
\newblock
\href{https://doi.org/10.5281/zenodo.16890907}{doi:\nolinkurl{10.5281/zenodo.16890907}}


\bibitem[Sardois(2021)]%
        {slucas2021generational}
\bibfield{author}{\bibinfo{person}{Lucas Sardois}.}
  \bibinfo{year}{2021}\natexlab{}.
\newblock \bibinfo{title}{Generational indices guide}.
\newblock
\newblock
\shownote{Blog
  \url{https://lucassardois.medium.com/generational-indices-guide-8e3c5f7fd594}.
  Accessed 2024-10-09}.


\bibitem[Schaller(2022)]%
        {schaller2022lazyupdate}
\bibfield{author}{\bibinfo{person}{Thomas Schaller}.}
  \bibinfo{year}{2022}\natexlab{}.
\newblock \bibinfo{title}{LazyUpdate in specs::world}.
\newblock
\newblock
\shownote{Documentation
  \url{https://docs.rs/specs/0.20.0/specs/world/struct.LazyUpdate.html}.
  Accessed 2024-10-16}.


\bibitem[Schaller(2023)]%
        {schaller2023specs}
\bibfield{author}{\bibinfo{person}{Thomas Schaller}.}
  \bibinfo{year}{2023}\natexlab{}.
\newblock \bibinfo{title}{The Specs Book}.
\newblock
\newblock
\shownote{\url{https://amethyst.github.io/specs/docs/tutorials/}. Living
  \url{https://github.com/amethyst/specs/tree/master/docs/tutorials}. Accessed
  2024-04-04}.


\bibitem[Schmierer(2017)]%
        {lschmierer2017benchmark}
\bibfield{author}{\bibinfo{person}{Lukas Schmierer}.}
  \bibinfo{year}{2017}\natexlab{}.
\newblock \bibinfo{title}{Benchmarks of various Rust Entity Component Systems}.
\newblock
\newblock
\shownote{\url{https://github.com/lschmierer/ecs_bench}. Accessed 2024-10-9}.


\bibitem[Tesler and Enea(1968)]%
        {Tesler-1968}
\bibfield{author}{\bibinfo{person}{L.~G. Tesler} {and} \bibinfo{person}{H.~J.
  Enea}.} \bibinfo{year}{1968}\natexlab{}.
\newblock \showarticletitle{A language design for concurrent processes}. In
  \bibinfo{booktitle}{\emph{Proceedings of the April 30--May 2, 1968, Spring
  Joint Computer Conference}} (Atlantic City, New Jersey)
  \emph{(\bibinfo{series}{AFIPS '68 (Spring)})}.
  \bibinfo{publisher}{Association for Computing Machinery},
  \bibinfo{address}{New York, NY, USA}, \bibinfo{pages}{403–408}.
\newblock
\showISBNx{9781450378970}
\href{https://doi.org/10.1145/1468075.1468134}{doi:\nolinkurl{10.1145/1468075.1468134}}


\bibitem[{Unity Technologies}(2024)]%
        {unity2024ecs}
\bibfield{author}{\bibinfo{person}{{Unity Technologies}}.}
  \bibinfo{year}{2024}\natexlab{}.
\newblock \bibinfo{title}{ECS for Unity}.
\newblock
\newblock
\shownote{\url{https://unity.com/ecs}. Version 0.17
  \url{https://docs.unity3d.com/Packages/com.unity.entities@0.17/manual/index.html}.
  Version 1.2
  \url{https://docs.unity3d.com/Packages/com.unity.entities@1.2/manual/index.html}.
  Accessed 2024-04-04}.


\bibitem[West(2018)]%
        {west2018using}
\bibfield{author}{\bibinfo{person}{Catherine West}.}
  \bibinfo{year}{2018}\natexlab{}.
\newblock \bibinfo{title}{Using Rust for Game Development (and What You Can
  Learn From It)}.
\newblock
\newblock
\shownote{Presentation
  \url{https://kyren.github.io/rustconf_2018_slides/index.html}. Video
  \url{https://www.youtube.com/watch?v=aKLntZcp27M}. Blog
  \url{https://kyren.github.io/2018/09/14/rustconf-talk.html}. Accessed
  2024-04-04}.


\bibitem[West(2007)]%
        {west2007evolve}
\bibfield{author}{\bibinfo{person}{Mick West}.}
  \bibinfo{year}{2007}\natexlab{}.
\newblock \bibinfo{title}{Evolve Your Hierarchy}.
\newblock
\newblock
\shownote{Blog
  \url{https://cowboyprogramming.com/2007/01/05/evolve-your-heirachy/}.
  Accessed 2024-04-05}.


\bibitem[Zhang et~al\mbox{.}(2022)]%
        {zhang2022relational}
\bibfield{author}{\bibinfo{person}{Yihong Zhang}, \bibinfo{person}{Yisu~Remy
  Wang}, \bibinfo{person}{Max Willsey}, {and} \bibinfo{person}{Zachary
  Tatlock}.} \bibinfo{year}{2022}\natexlab{}.
\newblock \showarticletitle{Relational e-matching}.
\newblock \bibinfo{journal}{\emph{Proc. ACM Program. Lang.}}
  \bibinfo{volume}{6}, \bibinfo{number}{POPL}, Article \bibinfo{articleno}{35}
  (\bibinfo{date}{Jan.} \bibinfo{year}{2022}), \bibinfo{numpages}{22}~pages.
\newblock
\href{https://doi.org/10.1145/3498696}{doi:\nolinkurl{10.1145/3498696}}


\end{thebibliography}

\ifdefined\ARXIVVERSION 

\appendix
\crefalias{section}{appendix}

\section{Grammar of Underlying Language Used in Examples}
\label{apx:example-grammar}

The underlying language used in our running examples is a lambda calculus extended with conveniences.
Its full grammar is as follows.

\begin{mdframed}
\begin{center}
\begin{tabular}{ c r l l }
    \(n : \mathbb{Z}\)
        & \(\bnfeq\)
        & \(\dots
            \;\bnfalt\; -2
            \;\bnfalt\; -1
            \;\bnfalt\; 0
            \;\bnfalt\; 1
            \;\bnfalt\; 2
            \;\bnfalt\; \dots
            \;\bnfalt\; - n
            \;\bnfalt\; n + n'
            \;\bnfalt\; \lceil n / n' \rceil
            \)
        & Integer \\
    \(e : E\)
        & \(\bnfeq\)
        & \(e_0
            \;\bnfalt\; e_1
            \;\bnfalt\; e_2
            \;\bnfalt\;
            \dots\)
        & Entity \\
    \(b\;\quad\)
        & \(\bnfeq\)
        & \(
                        n\overset{?}{    =} n'
            \;\bnfalt\; n\overset{?}{\not=} n'
            \;\bnfalt\; b\land b'
            \)
        & Bool\\
    \(d\;\quad\)
        & \(\bnfeq\)
        & \(\dots\text{not otherwise mentioned}\dots\)
        & Binding \\
    \(p\;\quad\)
        & \(\bnfeq\)
        & \(d
        \;\bnfalt\; (d,\ (d',\ \dots))
            \;\bnfalt\; (\langle d,\ \dots\rangle,
                        \langle (d',\ \dots),\ \dots \rangle)
            \)
        & Pattern\\
    \(t : T\)
        & \(\bnfeq\)
        & \( d
            \;\bnfalt\; t\ t'
            \;\bnfalt\; \lambda p.\ t
            \;\bnfalt\; m
            \;\bnfalt\; n
            \;\bnfalt\; e
            \;\bnfalt\; \mathbf{if}\ b\ \mathbf{then}\ t\ \mathbf{else}\ t'
            \)
        & Term
\end{tabular}
\end{center}
\end{mdframed}

\section{Comparison of Expressiveness in Core ECS and Bevy ECS}
\label{apx:executable-examples}

In this section we present examples of ECS programs written against our
executable model of Core ECS, and compare them to similar ECS programs
written against Bevy ECS. These comparisons illustrate that the
interface presented by Core ECS (and its executable model) allow a
programmer to freely express concurrency in an ECS program. By
comparison, the interfaces of existing ECS frameworks (with Bevy ECS
serving as their representative) constrain what concurrency can be
expressed and executed with parallelism.

The interface for our executable model of Core ECS is quite similar to
the presentation of Core ECS in \Cref{sec:core-ecs}. There are slight
differences, however, because our executable model is embedded in
Haskell. In Lst.~\ref{lst:em-api} we show the symbols and types exported
by our executable model, and describe them with comments. For brevity we
have left out type-class constraints, but they are not burdensome.

\begin{codelisting}

\caption{Interface of the Core ECS executable model in the Haskell
module \VERB|\DataTypeTok{PartialMapECS}|.}\label{lst:em-api}

\begin{Shaded}
\begin{Highlighting}[numbers=left,,]
\KeywordTok{newtype} \DataTypeTok{Entity} \OtherTok{=} \DataTypeTok{Entity} \DataTypeTok{Int} \CommentTok{{-}{-} Entities are fixed to integers in all ECS programs.}

\KeywordTok{data} \DataTypeTok{State}\NormalTok{ types }\KeywordTok{where} \CommentTok{{-}{-} State type takes the list of component types.}
\OtherTok{  (:+) ::} \DataTypeTok{IntMap}\NormalTok{ t }\OtherTok{{-}\textgreater{}} \DataTypeTok{State}\NormalTok{ ts }\OtherTok{{-}\textgreater{}} \DataTypeTok{State}\NormalTok{ (t }\OperatorTok{:}\NormalTok{ ts) }\CommentTok{{-}{-} One IntMap per component type.}
  \DataTypeTok{Metadata}\OtherTok{ ::} \DataTypeTok{Entity} \OtherTok{{-}\textgreater{}} \DataTypeTok{State}\NormalTok{ \textquotesingle{}[] }\CommentTok{{-}{-} State stores the next fresh{-}entity at the end.}

\KeywordTok{type} \KeywordTok{data} \DataTypeTok{Q} \OtherTok{=} \DataTypeTok{Q} \OperatorTok{:}\NormalTok{∧}\OperatorTok{:} \DataTypeTok{Q} \OperatorTok{|} \DataTypeTok{Incl}\NormalTok{ (}\OperatorTok{*}\NormalTok{) }\OperatorTok{|} \DataTypeTok{Excl}\NormalTok{ (}\OperatorTok{*}\NormalTok{) }\OperatorTok{|} \DataTypeTok{Anyway}\NormalTok{ (}\OperatorTok{*}\NormalTok{) }\CommentTok{{-}{-} Queries are type{-}level.}

\KeywordTok{data} \DataTypeTok{M}\NormalTok{ types }\KeywordTok{where} \CommentTok{{-}{-} Mutation type takes the list of component types.}
  \DataTypeTok{Attach}\OtherTok{ ::} \DataTypeTok{Entity} \OtherTok{{-}\textgreater{}}\NormalTok{ i }\OtherTok{{-}\textgreater{}} \DataTypeTok{M}\NormalTok{ ts}
  \DataTypeTok{Detach}\OtherTok{ ::} \DataTypeTok{Entity} \OtherTok{{-}\textgreater{}} \DataTypeTok{Proxy}\NormalTok{ i }\OtherTok{{-}\textgreater{}} \DataTypeTok{M}\NormalTok{ ts}
\NormalTok{  (}\OperatorTok{:}\NormalTok{•)}\OtherTok{ ::} \DataTypeTok{M}\NormalTok{ ts }\OtherTok{{-}\textgreater{}} \DataTypeTok{M}\NormalTok{ ts }\OtherTok{{-}\textgreater{}} \DataTypeTok{M}\NormalTok{ ts}
  \DataTypeTok{Nu}\OtherTok{ ::}\NormalTok{ (}\DataTypeTok{Entity} \OtherTok{{-}\textgreater{}} \DataTypeTok{M}\NormalTok{ ts) }\OtherTok{{-}\textgreater{}} \DataTypeTok{M}\NormalTok{ ts}
  \DataTypeTok{Nil}\OtherTok{ ::} \DataTypeTok{M}\NormalTok{ ts}

\KeywordTok{type} \DataTypeTok{System}\OtherTok{ ::}\NormalTok{ [}\OperatorTok{*}\NormalTok{] }\OtherTok{{-}\textgreater{}}\NormalTok{ [}\DataTypeTok{Q}\NormalTok{] }\OtherTok{{-}\textgreater{}} \OperatorTok{*} \CommentTok{{-}{-} A system{-}function type depends on a query.}
\KeywordTok{newtype} \DataTypeTok{System}\NormalTok{ ts q }\OtherTok{=} \DataTypeTok{System}\NormalTok{ ((}\DataTypeTok{QueryE}\NormalTok{ q, }\DataTypeTok{QueryT}\NormalTok{ q) }\OtherTok{{-}\textgreater{}} \DataTypeTok{M}\NormalTok{ ts)}

\KeywordTok{data} \DataTypeTok{Schedule}\NormalTok{ types }\KeywordTok{where} \CommentTok{{-}{-} Schedule type takes the list of component types.}
  \DataTypeTok{Conc}\OtherTok{ ::} \DataTypeTok{System}\NormalTok{ ts q }\OtherTok{{-}\textgreater{}} \DataTypeTok{Schedule}\NormalTok{ ts}
  \DataTypeTok{Seq}\OtherTok{  ::} \DataTypeTok{System}\NormalTok{ ts q }\OtherTok{{-}\textgreater{}} \DataTypeTok{Schedule}\NormalTok{ ts}
\NormalTok{  (}\OperatorTok{:}\NormalTok{∥)}\OtherTok{ ::} \DataTypeTok{Schedule}\NormalTok{ ts }\OtherTok{{-}\textgreater{}} \DataTypeTok{Schedule}\NormalTok{ ts }\OtherTok{{-}\textgreater{}} \DataTypeTok{Schedule}\NormalTok{ ts}
\NormalTok{  (}\OperatorTok{:}\NormalTok{≫)}\OtherTok{ ::} \DataTypeTok{Schedule}\NormalTok{ ts }\OtherTok{{-}\textgreater{}} \DataTypeTok{Schedule}\NormalTok{ ts }\OtherTok{{-}\textgreater{}} \DataTypeTok{Schedule}\NormalTok{ ts}

\NormalTok{(↓)}\OtherTok{ ::} \DataTypeTok{State}\NormalTok{ ts }\OtherTok{{-}\textgreater{}} \DataTypeTok{M}\NormalTok{ ts }\OtherTok{{-}\textgreater{}} \DataTypeTok{State}\NormalTok{ ts }\CommentTok{{-}{-} Apply a mutation to compute updated state.}
\OtherTok{new ::} \DataTypeTok{M}\NormalTok{ ts }\OtherTok{{-}\textgreater{}} \DataTypeTok{State}\NormalTok{ ts }\CommentTok{{-}{-} Compute an initial (or empty) state.}

\NormalTok{(⇓)}\OtherTok{ ::} \DataTypeTok{State}\NormalTok{ ts }\OtherTok{{-}\textgreater{}} \DataTypeTok{Schedule}\NormalTok{ ts }\OtherTok{{-}\textgreater{}} \DataTypeTok{M}\NormalTok{ ts }\CommentTok{{-}{-} Apply a schedule to compute a mutation.}

\NormalTok{(↓⇓)}\OtherTok{ ::} \DataTypeTok{State}\NormalTok{ ts }\OtherTok{{-}\textgreater{}} \DataTypeTok{Schedule}\NormalTok{ ts }\OtherTok{{-}\textgreater{}} \DataTypeTok{State}\NormalTok{ ts }\CommentTok{{-}{-} Alias for \textasciigrave{}(c ↓ (c ⇓ z))\textasciigrave{}.}
\OtherTok{run ::} \DataTypeTok{State}\NormalTok{ ts }\OtherTok{{-}\textgreater{}} \DataTypeTok{Schedule}\NormalTok{ ts }\OtherTok{{-}\textgreater{}} \DataTypeTok{IO}\NormalTok{ (}\DataTypeTok{State}\NormalTok{ ts) }\CommentTok{{-}{-} Like ↓⇓ but uses parallelism.}

\KeywordTok{type}\NormalTok{ (×)}\OtherTok{ ::} \OperatorTok{*} \OtherTok{{-}\textgreater{}} \OperatorTok{*} \OtherTok{{-}\textgreater{}} \OperatorTok{*} \CommentTok{{-}{-} Infix 2{-}tuple type for convenience.}
\KeywordTok{pattern}\OtherTok{ (:::) ::}\NormalTok{ a }\OtherTok{{-}\textgreater{}}\NormalTok{ b }\OtherTok{{-}\textgreater{}}\NormalTok{ a × b }\CommentTok{{-}{-} Infix 2{-}tuple constructor for convenience.}
\end{Highlighting}
\end{Shaded}

\end{codelisting}

There are a couple of symbols in Lst.~\ref{lst:em-api} that differ from
\Cref{sec:core-ecs}. The fresh mutation constructor on line 13 is
written \VERB|\NormalTok{(}\DataTypeTok{Nu}\NormalTok{ f)}| instead of
\(\mutfresh{f}\). Sequential schedule composition on line 23 is written
\VERB|\NormalTok{(z }\OperatorTok{:}\NormalTok{≫ z\textquotesingle{})}|
instead of \((z \spSeqComp z')\). Finally, we use the convenience
aliases on lines 33-34 for fixed-size vectors, writing
\VERB|\NormalTok{(a }\OperatorTok{:::}\NormalTok{ b }\OperatorTok{:::}\NormalTok{ c)}|,
instead of \(\langle a, b, c \rangle\).

There are two new symbols in Lst.~\ref{lst:em-api}, foreshadowed in
\Cref{sec:executable-model-new}, that implement schedule application. We
provide an alias \VERB|\NormalTok{(c ↓⇓ z)}| on line 30 for the
composition of \VERB|\NormalTok{(c ↓)}| after
\VERB|\NormalTok{(c ⇓ z)}|, which computes the mutation for a state and
then applies the mutation to obtain the updated state. This alias
contrasts with \VERB|\NormalTok{run}| on the next line, which
accomplishes the same task quite differently: \VERB|\NormalTok{run}|
places the state into a shared mutable reference and uses threads
(parallelism) wherever concurrency is present in the schedule. Every
mutation produced by an invocation of a system function by
\VERB|\NormalTok{run}| is immediately applied to the state under mutual
exclusion.

\subsection{Toy Physics Example}

Our first example, the toy physics simulation used throughout this
paper, is implemented in Haskell against our executable model in
Lst.~\ref{lst:em-toyphys}. The main purpose of this example is to
familiarize the reader with the syntax of our executable model of Core
ECS. After enabling some language features and importing the executable
model on lines 1-2 (both of which apply through the end of this
document), we define the component data types for one-dimensional
position and velocity on lines 4-5. The inertia and collision systems
(lines 10-11 and 13-20) are implemented largely the same way they were
formalized in \Cref{fig:example-system-functions}, except that the query
vector is a type-level list (the second parameter to the
\VERB|\DataTypeTok{System}| types on lines 10 and 14). The schedule
(lines 22-23) is the same as that of \Cref{fig:example-exec-schedule}.
Lines 25-29 define the initial state by creating and attaching
components to three entities. The example main loop (lines 31-33) prints
the state and applies the schedule a fixed number of times; it is called
on line 8.

\begin{codelisting}

\caption{The toy physics example implemented in Haskell with our
executable model of Core ECS.}\label{lst:em-toyphys}

\begin{Shaded}
\begin{Highlighting}[numbers=left,,]
\OtherTok{\{{-}\# LANGUAGE GHC2024, NegativeLiterals \#{-}\}}
\KeywordTok{import} \DataTypeTok{PartialMapECS}

\KeywordTok{newtype} \DataTypeTok{Pos} \OtherTok{=} \DataTypeTok{Pos} \DataTypeTok{Int} \KeywordTok{deriving}\NormalTok{ (}\DataTypeTok{Show}\NormalTok{, }\DataTypeTok{Eq}\NormalTok{)}
\KeywordTok{newtype} \DataTypeTok{Vel} \OtherTok{=} \DataTypeTok{Vel} \DataTypeTok{Int} \KeywordTok{deriving}\NormalTok{ (}\DataTypeTok{Show}\NormalTok{)}

\OtherTok{toyPhys ::} \DataTypeTok{IO}\NormalTok{ ()}
\NormalTok{toyPhys }\OtherTok{=}\NormalTok{ loop }\DecValTok{2}\NormalTok{ schedule start}
  \KeywordTok{where}
\OtherTok{    inertia ::} \DataTypeTok{Has}\NormalTok{ ts \textquotesingle{}[}\DataTypeTok{Pos}\NormalTok{] }\OtherTok{=\textgreater{}} \DataTypeTok{System}\NormalTok{ ts \textquotesingle{}[}\DataTypeTok{Incl} \DataTypeTok{Pos} \OperatorTok{:}\NormalTok{∧}\OperatorTok{:} \DataTypeTok{Incl} \DataTypeTok{Vel}\NormalTok{]}
\NormalTok{    inertia }\OtherTok{=} \DataTypeTok{System} \OperatorTok{$}\NormalTok{ \textbackslash{}(ej, }\DataTypeTok{Pos}\NormalTok{ pj }\OperatorTok{:::} \DataTypeTok{Vel}\NormalTok{ vj) }\OtherTok{{-}\textgreater{}} \DataTypeTok{Attach}\NormalTok{ ej (}\DataTypeTok{Pos} \OperatorTok{$}\NormalTok{ pj }\OperatorTok{+}\NormalTok{ vj)}

\OtherTok{    collide ::} \DataTypeTok{Has}\NormalTok{ ts \textquotesingle{}[}\DataTypeTok{Vel}\NormalTok{, }\DataTypeTok{Pos}\NormalTok{] }\OtherTok{=\textgreater{}}
        \DataTypeTok{System}\NormalTok{ ts \textquotesingle{}[}\DataTypeTok{Incl} \DataTypeTok{Pos} \OperatorTok{:}\NormalTok{∧}\OperatorTok{:} \DataTypeTok{Incl} \DataTypeTok{Vel}\NormalTok{, }\DataTypeTok{Incl} \DataTypeTok{Pos} \OperatorTok{:}\NormalTok{∧}\OperatorTok{:} \DataTypeTok{Excl} \DataTypeTok{Vel}\NormalTok{]}
\NormalTok{    collide }\OtherTok{=} \DataTypeTok{System} \OperatorTok{$}\NormalTok{ \textbackslash{}(ej }\OperatorTok{:::}\NormalTok{ eh, (pj}\OperatorTok{:::}\DataTypeTok{Vel}\NormalTok{ vj) }\OperatorTok{:::}\NormalTok{ (ph}\OperatorTok{:::}\NormalTok{())) }\OtherTok{{-}\textgreater{}}
        \KeywordTok{if}\NormalTok{ pj }\OperatorTok{==}\NormalTok{ ph }\KeywordTok{then}
            \DataTypeTok{Detach}\NormalTok{ ej (}\DataTypeTok{Proxy} \OperatorTok{@}\DataTypeTok{Pos}\NormalTok{) }\OperatorTok{:}\NormalTok{• }\DataTypeTok{Detach}\NormalTok{ ej (}\DataTypeTok{Proxy} \OperatorTok{@}\DataTypeTok{Vel}\NormalTok{) }\OperatorTok{:}\NormalTok{•}
            \DataTypeTok{Attach}\NormalTok{ eh (}\DataTypeTok{Vel} \OperatorTok{$} \FunctionTok{quot}\NormalTok{ vj }\DecValTok{2}\NormalTok{) }\OperatorTok{:}\NormalTok{•}
            \DataTypeTok{Nu}\NormalTok{ (\textbackslash{}el }\OtherTok{{-}\textgreater{}} \DataTypeTok{Attach}\NormalTok{ el pj }\OperatorTok{:}\NormalTok{• }\DataTypeTok{Attach}\NormalTok{ el (}\DataTypeTok{Vel} \OperatorTok{$} \FunctionTok{quot}\NormalTok{ vj }\OperatorTok{{-}}\DecValTok{2}\NormalTok{))}
        \KeywordTok{else} \DataTypeTok{Nil}

\OtherTok{    schedule ::} \DataTypeTok{Schedule}\NormalTok{ \textquotesingle{}[}\DataTypeTok{Pos}\NormalTok{, }\DataTypeTok{Vel}\NormalTok{]}
\NormalTok{    schedule }\OtherTok{=} \DataTypeTok{Conc}\NormalTok{ inertia }\OperatorTok{:}\NormalTok{≫ }\DataTypeTok{Seq}\NormalTok{ collide}

\OtherTok{    start ::} \DataTypeTok{State}\NormalTok{ \textquotesingle{}[}\DataTypeTok{Pos}\NormalTok{, }\DataTypeTok{Vel}\NormalTok{]}
\NormalTok{    start }\OtherTok{=}\NormalTok{ new}
        \OperatorTok{$}  \DataTypeTok{Nu}\NormalTok{ (\textbackslash{}e }\OtherTok{{-}\textgreater{}} \DataTypeTok{Attach}\NormalTok{ e (}\DataTypeTok{Pos} \DecValTok{1}\NormalTok{) }\OperatorTok{:}\NormalTok{• }\DataTypeTok{Attach}\NormalTok{ e (}\DataTypeTok{Vel} \DecValTok{6}\NormalTok{))}
        \OperatorTok{:}\NormalTok{• }\DataTypeTok{Nu}\NormalTok{ (\textbackslash{}e }\OtherTok{{-}\textgreater{}} \DataTypeTok{Attach}\NormalTok{ e (}\DataTypeTok{Pos} \DecValTok{7}\NormalTok{))}
        \OperatorTok{:}\NormalTok{• }\DataTypeTok{Nu}\NormalTok{ (\textbackslash{}e }\OtherTok{{-}\textgreater{}} \DataTypeTok{Attach}\NormalTok{ e (}\DataTypeTok{Pos} \DecValTok{9}\NormalTok{) }\OperatorTok{:}\NormalTok{• }\DataTypeTok{Attach}\NormalTok{ e (}\DataTypeTok{Vel} \OperatorTok{{-}}\DecValTok{2}\NormalTok{))}

\OtherTok{loop ::} \DataTypeTok{Show}\NormalTok{ (}\DataTypeTok{State}\NormalTok{ ts) }\OtherTok{=\textgreater{}} \DataTypeTok{Integer} \OtherTok{{-}\textgreater{}} \DataTypeTok{Schedule}\NormalTok{ ts }\OtherTok{{-}\textgreater{}} \DataTypeTok{State}\NormalTok{ ts }\OtherTok{{-}\textgreater{}} \DataTypeTok{IO}\NormalTok{ ()}
\NormalTok{loop n z c }\OperatorTok{|} \DecValTok{0} \OperatorTok{\textless{}}\NormalTok{ n }\OtherTok{=} \FunctionTok{print}\NormalTok{ c }\OperatorTok{\textgreater{}\textgreater{}}\NormalTok{ loop (n }\OperatorTok{{-}} \DecValTok{1}\NormalTok{) z (c ↓ (c ⇓ z))}
           \OperatorTok{|} \FunctionTok{otherwise} \OtherTok{=} \FunctionTok{putStrLn}\NormalTok{ (}\FunctionTok{show}\NormalTok{ c }\OperatorTok{++} \StringTok{" END"}\NormalTok{)}
\end{Highlighting}
\end{Shaded}

\end{codelisting}

The function \VERB|\NormalTok{toyPhys}| in Lst.~\ref{lst:em-toyphys}
produces the three lines of output in Lst.~\ref{lst:em-toyphys-out}.
This output depicts ECS program state the same way as in
\Cref{fig:example-state}, and the first line in
Lst.~\ref{lst:em-toyphys-out} matches \(c\) in \Cref{fig:example-state}.
The three lines of output together describe the three frames of
\Cref{fig:badphysics-conc-a}: \texttt{e0} collides with \texttt{e1},
destroying \texttt{e0} and sending \texttt{e1} flying to the right while
creating \texttt{e3} in the process. Meanwhile, \texttt{e2} passes by
unimpeded. One difference to note, between ECS program state in
Lst.~\ref{lst:em-toyphys-out} and that of Core ECS, is that because the
executable model implements a mechanism to generate fresh entities,
there is a piece of metadata at the end of each line that describes the
next fresh entity to use.

\begin{codelisting}

\caption{Output produced by \VERB|\NormalTok{toyPhys}| in
Lst.~\ref{lst:em-toyphys}.}\label{lst:em-toyphys-out}

\begin{Shaded}
\begin{Highlighting}[numbers=left,,]
\NormalTok{Pos↦\{e0 ↦ Pos 1, e1 ↦ Pos 7, e2 ↦ Pos 9\} :+ Vel↦\{e0 ↦ Vel 6, e2 ↦ Vel ({-}2)\} :+ Metadata \{nextFresh = e3\}}
\NormalTok{Pos↦\{e1 ↦ Pos 7, e2 ↦ Pos 7, e3 ↦ Pos 7\} :+ Vel↦\{e1 ↦ Vel 3, e2 ↦ Vel ({-}2), e3 ↦ Vel ({-}3)\} :+ Metadata \{nextFresh = e4\}}
\NormalTok{Pos↦\{e1 ↦ Pos 10, e2 ↦ Pos 5, e3 ↦ Pos 4\} :+ Vel↦\{e1 ↦ Vel 3, e2 ↦ Vel ({-}2), e3 ↦ Vel ({-}3)\} :+ Metadata \{nextFresh = e4\} END}
\end{Highlighting}
\end{Shaded}

\end{codelisting}

Now, we develop the Bevy ECS implementation of the toy physics
simulation in Lst.~\ref{lst:rs-toyphys}. There are several differences
from \VERB|\NormalTok{toyPhys}| in Lst.~\ref{lst:em-toyphys}:

\begin{codelisting}

\caption{The toy physics example implemented in Rust with Bevy
ECS.}\label{lst:rs-toyphys}

\begin{Shaded}
\begin{Highlighting}[numbers=left,,]
\KeywordTok{use} \PreprocessorTok{bevy\_ecs::prelude::}\OperatorTok{*;}

\KeywordTok{mod}\NormalTok{ toy\_phys }\OperatorTok{\{}
  \KeywordTok{use} \KeywordTok{super}\PreprocessorTok{::}\OperatorTok{*;} \KeywordTok{use} \PreprocessorTok{std::collections::}\NormalTok{BTreeSet}\OperatorTok{;}

  \AttributeTok{\#[}\NormalTok{derive}\AttributeTok{(}\BuiltInTok{Debug}\OperatorTok{,} \BuiltInTok{Clone}\OperatorTok{,} \BuiltInTok{PartialEq}\OperatorTok{,}\NormalTok{ Component}\AttributeTok{)]}
  \KeywordTok{struct}\NormalTok{ Pos(}\DataTypeTok{i32}\NormalTok{)}\OperatorTok{;}
  \AttributeTok{\#[}\NormalTok{derive}\AttributeTok{(}\BuiltInTok{Debug}\OperatorTok{,}\NormalTok{ Component}\AttributeTok{)]}
  \KeywordTok{struct}\NormalTok{ Vel(}\DataTypeTok{i32}\NormalTok{)}\OperatorTok{;}

  \KeywordTok{fn}\NormalTok{ inertia(}\KeywordTok{mut}\NormalTok{ q}\OperatorTok{:}\NormalTok{ Query}\OperatorTok{\textless{}}\NormalTok{(}\OperatorTok{\&}\KeywordTok{mut}\NormalTok{ Pos}\OperatorTok{,} \OperatorTok{\&}\NormalTok{Vel)}\OperatorTok{\textgreater{}}\NormalTok{) }\OperatorTok{\{}
\NormalTok{    q}\OperatorTok{.}\NormalTok{par\_iter\_mut()}\OperatorTok{.}\NormalTok{for\_each(}\OperatorTok{|}\NormalTok{(}\KeywordTok{mut}\NormalTok{ p}\OperatorTok{,}\NormalTok{ Vel(v))}\OperatorTok{|} \OperatorTok{\{}\NormalTok{ p}\OperatorTok{.}\DecValTok{0} \OperatorTok{+=}\NormalTok{ v}\OperatorTok{;} \OperatorTok{\}}\NormalTok{)}
  \OperatorTok{\}}

  \KeywordTok{fn}\NormalTok{ collide(}\KeywordTok{mut}\NormalTok{ command}\OperatorTok{:}\NormalTok{ Commands}\OperatorTok{,}\NormalTok{ moving}\OperatorTok{:}\NormalTok{ Query}\OperatorTok{\textless{}}\NormalTok{(Entity}\OperatorTok{,} \OperatorTok{\&}\NormalTok{Pos}\OperatorTok{,} \OperatorTok{\&}\NormalTok{Vel)}\OperatorTok{\textgreater{},}
\NormalTok{                                stationary}\OperatorTok{:}\NormalTok{ Query}\OperatorTok{\textless{}}\NormalTok{(Entity}\OperatorTok{,} \OperatorTok{\&}\NormalTok{Pos)}\OperatorTok{,}\NormalTok{ Without}\OperatorTok{\textless{}}\NormalTok{Vel}\OperatorTok{\textgreater{}\textgreater{}}\NormalTok{) }\OperatorTok{\{}
    \KeywordTok{let} \KeywordTok{mut}\NormalTok{ used }\OperatorTok{=} \PreprocessorTok{BTreeSet::}\NormalTok{new()}\OperatorTok{;}
    \ControlFlowTok{for}\NormalTok{ (ej}\OperatorTok{,}\NormalTok{ pj}\OperatorTok{,}\NormalTok{ vj) }\KeywordTok{in} \OperatorTok{\&}\NormalTok{moving }\OperatorTok{\{}
      \ControlFlowTok{for}\NormalTok{ (eh}\OperatorTok{,}\NormalTok{ ph) }\KeywordTok{in} \OperatorTok{\&}\NormalTok{stationary }\OperatorTok{\{}
        \ControlFlowTok{if}\NormalTok{ used}\OperatorTok{.}\NormalTok{contains(}\OperatorTok{\&}\NormalTok{ej) }\OperatorTok{||}\NormalTok{ used}\OperatorTok{.}\NormalTok{contains(}\OperatorTok{\&}\NormalTok{eh) }\OperatorTok{\{} \ControlFlowTok{continue}\OperatorTok{;} \OperatorTok{\}}
        \ControlFlowTok{if}\NormalTok{ pj }\OperatorTok{==}\NormalTok{ ph }\OperatorTok{\{}
\NormalTok{          command}\OperatorTok{.}\NormalTok{entity(ej)}\OperatorTok{.}\NormalTok{despawn()}\OperatorTok{;}
\NormalTok{          used}\OperatorTok{.}\NormalTok{insert(ej)}\OperatorTok{;} \CommentTok{// was deleted}
\NormalTok{          command}\OperatorTok{.}\NormalTok{entity(eh)}\OperatorTok{.}\NormalTok{insert(Vel(vj}\OperatorTok{.}\DecValTok{0} \OperatorTok{/} \DecValTok{2}\NormalTok{))}\OperatorTok{;}
\NormalTok{          used}\OperatorTok{.}\NormalTok{insert(eh)}\OperatorTok{;} \CommentTok{// was given velocity}
\NormalTok{          command}\OperatorTok{.}\NormalTok{spawn((pj}\OperatorTok{.}\NormalTok{clone()}\OperatorTok{,}\NormalTok{ Vel(}\OperatorTok{{-}}\NormalTok{vj}\OperatorTok{.}\DecValTok{0} \OperatorTok{/} \DecValTok{2}\NormalTok{)))}\OperatorTok{;}
  \OperatorTok{\}} \OperatorTok{\}} \OperatorTok{\}} \OperatorTok{\}}

  \KeywordTok{fn}\NormalTok{ print\_all(ps}\OperatorTok{:}\NormalTok{ Query}\OperatorTok{\textless{}}\NormalTok{(Entity}\OperatorTok{,} \OperatorTok{\&}\NormalTok{Pos)}\OperatorTok{\textgreater{},}\NormalTok{ vs}\OperatorTok{:}\NormalTok{ Query}\OperatorTok{\textless{}}\NormalTok{(Entity}\OperatorTok{,} \OperatorTok{\&}\NormalTok{Vel)}\OperatorTok{\textgreater{}}\NormalTok{) }\OperatorTok{\{}
    \PreprocessorTok{print!}\NormalTok{(   }\StringTok{"Pos↦["}\NormalTok{)}\OperatorTok{;} \ControlFlowTok{for}\NormalTok{ (e}\OperatorTok{,}\NormalTok{ p) }\KeywordTok{in} \OperatorTok{\&}\NormalTok{ps }\OperatorTok{\{} \PreprocessorTok{print!}\NormalTok{(}\StringTok{"e\{\}↦\{:?\}, "}\OperatorTok{,}\NormalTok{ e}\OperatorTok{.}\NormalTok{index()}\OperatorTok{,}\NormalTok{ p) }\OperatorTok{\}}
    \PreprocessorTok{print!}\NormalTok{(}\StringTok{"], Vel↦["}\NormalTok{)}\OperatorTok{;} \ControlFlowTok{for}\NormalTok{ (e}\OperatorTok{,}\NormalTok{ v) }\KeywordTok{in} \OperatorTok{\&}\NormalTok{vs }\OperatorTok{\{} \PreprocessorTok{print!}\NormalTok{(}\StringTok{"e\{\}↦\{:?\}, "}\OperatorTok{,}\NormalTok{ e}\OperatorTok{.}\NormalTok{index()}\OperatorTok{,}\NormalTok{ v) }\OperatorTok{\}}
    \PreprocessorTok{print!}\NormalTok{(}\StringTok{"]}\SpecialCharTok{\textbackslash{}n}\StringTok{"}\NormalTok{)}\OperatorTok{;}
  \OperatorTok{\}}

  \KeywordTok{pub} \KeywordTok{fn}\NormalTok{ demo() }\OperatorTok{\{}
    \KeywordTok{let} \KeywordTok{mut}\NormalTok{ world }\OperatorTok{=} \PreprocessorTok{World::}\KeywordTok{default}\NormalTok{()}\OperatorTok{;}
\NormalTok{    world}\OperatorTok{.}\NormalTok{spawn((Pos(}\DecValTok{1}\NormalTok{)}\OperatorTok{,}\NormalTok{ Vel(}\DecValTok{6}\NormalTok{)))}\OperatorTok{;}
\NormalTok{    world}\OperatorTok{.}\NormalTok{spawn(Pos(}\DecValTok{7}\NormalTok{))}\OperatorTok{;}
\NormalTok{    world}\OperatorTok{.}\NormalTok{spawn((Pos(}\DecValTok{9}\NormalTok{)}\OperatorTok{,}\NormalTok{ Vel(}\OperatorTok{{-}}\DecValTok{2}\NormalTok{)))}\OperatorTok{;}
    \KeywordTok{let} \KeywordTok{mut}\NormalTok{ schedule }\OperatorTok{=} \PreprocessorTok{Schedule::}\KeywordTok{default}\NormalTok{()}\OperatorTok{;}
\NormalTok{    schedule}\OperatorTok{.}\NormalTok{add\_systems((print\_all}\OperatorTok{,}\NormalTok{ inertia}\OperatorTok{,}\NormalTok{ collide)}\OperatorTok{.}\NormalTok{chain())}\OperatorTok{;}
    \ControlFlowTok{for}\NormalTok{ \_ }\KeywordTok{in} \DecValTok{0}\OperatorTok{..}\DecValTok{3} \OperatorTok{\{}\NormalTok{ schedule}\OperatorTok{.}\NormalTok{run(}\OperatorTok{\&}\KeywordTok{mut}\NormalTok{ world)}\OperatorTok{;} \OperatorTok{\}}
  \OperatorTok{\}}
\OperatorTok{\}} \CommentTok{// end mod toy\_phys}
\end{Highlighting}
\end{Shaded}

\end{codelisting}

In Bevy ECS, each system must iterate manually over its query results.
This has a few consequences: First, in Bevy ECS a system (not the
schedule) specifies whether it uses intra-system concurrency by choosing
to use a parallel iteration interface. This can be seen on line 11 in
Lst.~\ref{lst:rs-toyphys} where the \VERB|\NormalTok{inertia}| system
iterates over its query result using \VERB|\NormalTok{par\_iter\_mut}|.
Second, in Bevy ECS a system with multiple queries may be written to be
more efficient than a \(O(n^x)\) nesting of loops; it may iterate first
over one query, and accumulate a local variable, to filter when
iterating over subsequent queries. We do not leverage this in the
\VERB|\NormalTok{collide}| system; instead we use two nested loops on
lines 18-19, \(O(n^2)\), to match the behavior of the Core ECS
executable model.

In Bevy ECS, a schedule can express inter-system concurrency, but not
intra-system concurrency. In the schedule on line 42 of
Lst.~\ref{lst:rs-toyphys}, the systems are run one-after-the-other, as
though with sequential composition of schedules in Core ECS. If the bare
tuple of systems were passed in without the call to
\VERB|\NormalTok{chain}|, then Bevy ECS would assess and potentially run
the systems concurrently \citep{bevy2024systemcompat}. However, leaving
systems ``ambiguously ordered'' in this way is not a guarantee that Bevy
ECS will consider them eligible to run concurrently.

In Bevy ECS, \emph{commands} are required to perform changes such as
creating (or deleting) entities and attaching (or detaching) components
within a system function. However these commands do not take effect
until after the system stops running (they are buffered). Whereas in the
\VERB|\NormalTok{inertia}| system in Lst.~\ref{lst:rs-toyphys} on line
11 we only perform an in-place mutation of entities' positions (an Owned
update according to \Cref{fig:concurrency-slice-and-dice}), the
\VERB|\NormalTok{collide}| system uses commands to delete entity
\texttt{ej} (line 22; a Deferred delete) and to attach a velocity
component to entity \texttt{eh} (line 24; a Deferred insert). Since
these commands do not take effect while the \VERB|\NormalTok{collide}|
system is running, it is necessary to create and update a local set of
entities that are invalid for further collisions (lines 17, 23, and 25),
and to explicitly filter such entities during manual iteration (line
20). By contrast, in Core ECS a system run with \(\spSeqOne{-}\) does
not iterate manually and so the writes performed by one call are visible
in subsequent calls, making the use of a tombstone-set unnecessary.

Running \VERB|\PreprocessorTok{toy\_phys::}\NormalTok{demo}| in
Lst.~\ref{lst:rs-toyphys} produces the three lines of output in
Lst.~\ref{lst:rs-toyphys-out}. These differ from the output produced by
our executable model of Core ECS in minor ways. Bevy ECS does not
iterate over query results sorted by entity index. This means that we
might have seen the output differ from our executable model in a
significant way (entity \texttt{e2} could have collided with entity
\texttt{e1} instead of \texttt{e0}) and it is only by chance that the
same collision took place in both demos.

\begin{codelisting}

\caption{Output produced by
\VERB|\PreprocessorTok{toy\_phys::}\NormalTok{demo}| in
Lst.~\ref{lst:rs-toyphys}.}\label{lst:rs-toyphys-out}

\begin{Shaded}
\begin{Highlighting}[numbers=left,,]
\NormalTok{Pos↦[e0↦Pos(1), e2↦Pos(9), e1↦Pos(7), ], Vel↦[e0↦Vel(6), e2↦Vel({-}2), ]}
\NormalTok{Pos↦[e2↦Pos(7), e1↦Pos(7), e3↦Pos(7), ], Vel↦[e2↦Vel({-}2), e1↦Vel(3), e3↦Vel({-}3), ]}
\NormalTok{Pos↦[e2↦Pos(5), e1↦Pos(10), e3↦Pos(4), ], Vel↦[e2↦Vel({-}2), e1↦Vel(3), e3↦Vel({-}3), ]}
\end{Highlighting}
\end{Shaded}

\end{codelisting}

While the toy physics simulation is expressible in both Core ECS and
Bevy ECS, it is already somewhat clear that Core ECS has an expressivity
benefit over traditional ECS frameworks, with the latter represented
here by Bevy ECS. The \emph{essence} of an ECS program (querying the
entity-component association and attaching or detaching components to
entities) is more clear and concise with the Core ECS implementation.
Let us now look directly at whether Core ECS can express programs that
Bevy ECS cannot.

\subsection{Disjoint Entities Example}

Our next example explores the situation where two systems that run
concurrently update the same component on disjoint subsets of the live
entities. While this is safe under the second case of
\Cref{thm:infl-safe-schedule} (and so deterministic despite the presence
of scheduler non-determinism), its safety is difficult to check
statically. It can be expressed easily with the executable model of Core
ECS, but it is less straightforward to do so with Bevy ECS.

In this example each entity has a numeric component attached, and while
one system increases those numbers that are below a threshold,
simultaneously another system decreases the numbers that are not below
the threshold. Since both systems use the same threshold, the component
values that they write to are disjoint, despite being the same
component. For example, if we start with two entities having \(3\) and
\(8\) (and a threshold of \(4\)), then after one application of the
schedule we will have entities with \(4\) and \(7\) (since \(3 < 4\) the
\(3\) was incremented, and since \(8 \not< 4\) the \(8\) was
decremented). On the next iteration both will be decremented, so we have
entities with \(3\) and \(6\).

First we show the executable Core ECS model of the disjoint entities
example in Lst.~\ref{lst:em-disjointentities}. It uses the same system
twice (lines 4-7), passing in the threshold value and whether to be the
incrementing or the decrementing system on lines 9-10; these
configuration arguments are held constant while the ECS program is
running. The systems are run concurrently with themselves, and with each
other, by the schedule on lines 12-13. The output produced by this demo,
shown in Lst.~\ref{lst:em-disjointentities-out}, is exactly as it was
predicted to be in the paragraph above.

\begin{codelisting}

\caption{The disjoint entities example implemented against our
executable model of Core ECS.}\label{lst:em-disjointentities}

\begin{Shaded}
\begin{Highlighting}[numbers=left,,]
\OtherTok{disjointEntities ::} \DataTypeTok{IO}\NormalTok{ ()}
\NormalTok{disjointEntities }\OtherTok{=}\NormalTok{ loop }\DecValTok{2}\NormalTok{ schedule start}
  \KeywordTok{where}
\OtherTok{    adjust ::} \DataTypeTok{Has}\NormalTok{ ts \textquotesingle{}[}\DataTypeTok{Int}\NormalTok{] }\OtherTok{=\textgreater{}} \DataTypeTok{Int} \OtherTok{{-}\textgreater{}} \DataTypeTok{Bool} \OtherTok{{-}\textgreater{}} \DataTypeTok{System}\NormalTok{ ts \textquotesingle{}[}\DataTypeTok{Incl} \DataTypeTok{Int}\NormalTok{]}
\NormalTok{    adjust bound bncrement}
      \OperatorTok{|}\NormalTok{ bncrement }\OtherTok{=} \DataTypeTok{System} \OperatorTok{$}\NormalTok{ \textbackslash{}(e, n) }\OtherTok{{-}\textgreater{}} \KeywordTok{if}\NormalTok{ n }\OperatorTok{\textless{}}\NormalTok{ bound }\KeywordTok{then} \DataTypeTok{Attach}\NormalTok{ e (n }\OperatorTok{+} \DecValTok{1}\NormalTok{) }\KeywordTok{else} \DataTypeTok{Nil}
      \OperatorTok{|} \FunctionTok{otherwise} \OtherTok{=} \DataTypeTok{System} \OperatorTok{$}\NormalTok{ \textbackslash{}(e, n) }\OtherTok{{-}\textgreater{}} \KeywordTok{if}\NormalTok{ n }\OperatorTok{\textless{}}\NormalTok{ bound }\KeywordTok{then} \DataTypeTok{Nil} \KeywordTok{else} \DataTypeTok{Attach}\NormalTok{ e (n }\OperatorTok{{-}} \DecValTok{1}\NormalTok{)}

\NormalTok{    increment }\OtherTok{=}\NormalTok{ adjust }\DecValTok{4} \DataTypeTok{True}
\NormalTok{    decrement }\OtherTok{=}\NormalTok{ adjust }\DecValTok{4} \DataTypeTok{False}

\OtherTok{    schedule ::} \DataTypeTok{Schedule}\NormalTok{ \textquotesingle{}[}\DataTypeTok{Int}\NormalTok{]}
\NormalTok{    schedule }\OtherTok{=} \DataTypeTok{Conc}\NormalTok{ increment }\OperatorTok{:}\NormalTok{∥ }\DataTypeTok{Conc}\NormalTok{ decrement}

\OtherTok{    start ::} \DataTypeTok{State}\NormalTok{ \textquotesingle{}[}\DataTypeTok{Int}\NormalTok{]}
\NormalTok{    start }\OtherTok{=}\NormalTok{ new}
        \OperatorTok{$}  \DataTypeTok{Nu}\NormalTok{ (\textbackslash{}e }\OtherTok{{-}\textgreater{}} \DataTypeTok{Attach}\NormalTok{ e (}\DecValTok{3}\OtherTok{ ::} \DataTypeTok{Int}\NormalTok{))}
        \OperatorTok{:}\NormalTok{• }\DataTypeTok{Nu}\NormalTok{ (\textbackslash{}e }\OtherTok{{-}\textgreater{}} \DataTypeTok{Attach}\NormalTok{ e (}\DecValTok{8}\OtherTok{ ::} \DataTypeTok{Int}\NormalTok{))}
\end{Highlighting}
\end{Shaded}

\end{codelisting}

\begin{codelisting}

\caption{Output produced by \VERB|\NormalTok{disjointEntities}| in
Lst.~\ref{lst:em-disjointentities}.}\label{lst:em-disjointentities-out}

\begin{Shaded}
\begin{Highlighting}[numbers=left,,]
\NormalTok{Int↦\{e0 ↦ 3, e1 ↦ 8\} :+ Metadata \{nextFresh = e2\}}
\NormalTok{Int↦\{e0 ↦ 4, e1 ↦ 7\} :+ Metadata \{nextFresh = e2\}}
\NormalTok{Int↦\{e0 ↦ 3, e1 ↦ 6\} :+ Metadata \{nextFresh = e2\} END}
\end{Highlighting}
\end{Shaded}

\end{codelisting}

Next we develop the Bevy ECS implementation of the disjoint entities
example in Lst.~\ref{lst:rs-disjointentities}. We use separate systems
to implement the increment and decrement behavior (lines 9-11 and 12-14;
both Owned updates per \Cref{fig:concurrency-slice-and-dice}) because it
is difficult to reuse one system with different configurations in Bevy
ECS.\footnote{We did not pursue running the same system concurrently with different configurations because that is not our focus here.}
The schedule we specified on line 26 first prints the world state and
then runs the increment and decrement systems without an ordering
requirement --- implying that we desire them to be run concurrently.
However, Bevy ECS will run the increment and decrement systems
sequentially anyway: Both mutably borrow the \VERB|\NormalTok{Num}|
component, so they conflict.

\begin{codelisting}

\caption{A first implementation of the disjoint entities example with
Bevy ECS.}\label{lst:rs-disjointentities}

\begin{Shaded}
\begin{Highlighting}[numbers=left,,]
\AttributeTok{\#[}\NormalTok{derive}\AttributeTok{(}\BuiltInTok{Debug}\OperatorTok{,}\NormalTok{ Component}\AttributeTok{)]}
\KeywordTok{struct}\NormalTok{ Num(}\DataTypeTok{i32}\NormalTok{)}\OperatorTok{;}

\KeywordTok{const}\NormalTok{ THRESHOLD}\OperatorTok{:} \DataTypeTok{i32} \OperatorTok{=} \DecValTok{4}\OperatorTok{;}

\KeywordTok{mod}\NormalTok{ disjoint\_entities }\OperatorTok{\{}
  \KeywordTok{use} \KeywordTok{super}\PreprocessorTok{::}\OperatorTok{*;} \KeywordTok{use} \PreprocessorTok{bevy\_ecs::schedule::}\OperatorTok{*;}

  \KeywordTok{fn}\NormalTok{ increment(}\KeywordTok{mut}\NormalTok{ q}\OperatorTok{:}\NormalTok{ Query}\OperatorTok{\textless{}\&}\KeywordTok{mut}\NormalTok{ Num}\OperatorTok{\textgreater{}}\NormalTok{) }\OperatorTok{\{}
\NormalTok{    q}\OperatorTok{.}\NormalTok{par\_iter\_mut()}\OperatorTok{.}\NormalTok{for\_each(}\OperatorTok{|}\KeywordTok{mut}\NormalTok{ n}\OperatorTok{|} \ControlFlowTok{if}\NormalTok{ n}\OperatorTok{.}\DecValTok{0} \OperatorTok{\textless{}}\NormalTok{ THRESHOLD }\OperatorTok{\{}\NormalTok{ n}\OperatorTok{.}\DecValTok{0} \OperatorTok{+=} \DecValTok{1} \OperatorTok{\}} \ControlFlowTok{else} \OperatorTok{\{\}}\NormalTok{)}\OperatorTok{;}
  \OperatorTok{\}}
  \KeywordTok{fn}\NormalTok{ decrement(}\KeywordTok{mut}\NormalTok{ q}\OperatorTok{:}\NormalTok{ Query}\OperatorTok{\textless{}\&}\KeywordTok{mut}\NormalTok{ Num}\OperatorTok{\textgreater{}}\NormalTok{) }\OperatorTok{\{}
\NormalTok{    q}\OperatorTok{.}\NormalTok{par\_iter\_mut()}\OperatorTok{.}\NormalTok{for\_each(}\OperatorTok{|}\KeywordTok{mut}\NormalTok{ n}\OperatorTok{|} \ControlFlowTok{if}\NormalTok{ n}\OperatorTok{.}\DecValTok{0} \OperatorTok{\textless{}}\NormalTok{ THRESHOLD }\OperatorTok{\{\}} \ControlFlowTok{else} \OperatorTok{\{}\NormalTok{ n}\OperatorTok{.}\DecValTok{0} \OperatorTok{{-}=} \DecValTok{1} \OperatorTok{\}}\NormalTok{)}\OperatorTok{;}
  \OperatorTok{\}}

  \KeywordTok{fn}\NormalTok{ print\_all(q}\OperatorTok{:}\NormalTok{ Query}\OperatorTok{\textless{}}\NormalTok{(Entity}\OperatorTok{,} \OperatorTok{\&}\NormalTok{Num)}\OperatorTok{\textgreater{}}\NormalTok{) }\OperatorTok{\{}
    \PreprocessorTok{print!}\NormalTok{(}\StringTok{"Num↦["}\NormalTok{)}\OperatorTok{;} \ControlFlowTok{for}\NormalTok{ (e}\OperatorTok{,}\NormalTok{ n) }\KeywordTok{in} \OperatorTok{\&}\NormalTok{q }\OperatorTok{\{} \PreprocessorTok{print!}\NormalTok{(}\StringTok{"e\{\}↦\{:?\}, "}\OperatorTok{,}\NormalTok{ e}\OperatorTok{.}\NormalTok{index()}\OperatorTok{,}\NormalTok{ n) }\OperatorTok{\}}
    \PreprocessorTok{print!}\NormalTok{(}\StringTok{"]}\SpecialCharTok{\textbackslash{}n}\StringTok{"}\NormalTok{)}\OperatorTok{;}
  \OperatorTok{\}}

  \KeywordTok{pub} \KeywordTok{fn}\NormalTok{ demo() }\OperatorTok{\{}
    \KeywordTok{let} \KeywordTok{mut}\NormalTok{ world }\OperatorTok{=} \PreprocessorTok{World::}\KeywordTok{default}\NormalTok{()}\OperatorTok{;}
\NormalTok{    world}\OperatorTok{.}\NormalTok{spawn(Num(}\DecValTok{3}\NormalTok{))}\OperatorTok{;}
\NormalTok{    world}\OperatorTok{.}\NormalTok{spawn(Num(}\DecValTok{8}\NormalTok{))}\OperatorTok{;}
    \KeywordTok{let} \KeywordTok{mut}\NormalTok{ schedule }\OperatorTok{=} \PreprocessorTok{Schedule::}\KeywordTok{default}\NormalTok{()}\OperatorTok{;}
\NormalTok{    schedule}\OperatorTok{.}\NormalTok{add\_systems((print\_all}\OperatorTok{,}\NormalTok{ (increment}\OperatorTok{,}\NormalTok{ decrement))}\OperatorTok{.}\NormalTok{chain())}\OperatorTok{;}
    \ControlFlowTok{for}\NormalTok{ \_ }\KeywordTok{in} \DecValTok{0}\OperatorTok{..}\DecValTok{3} \OperatorTok{\{}\NormalTok{ schedule}\OperatorTok{.}\NormalTok{run(}\OperatorTok{\&}\KeywordTok{mut}\NormalTok{ world)}\OperatorTok{;} \OperatorTok{\}}
  \OperatorTok{\}}
\OperatorTok{\}} \CommentTok{// end mod disjoint\_entities}
\end{Highlighting}
\end{Shaded}

\end{codelisting}

When \VERB|\PreprocessorTok{disjoint\_entities::}\NormalTok{demo}| in
Lst.~\ref{lst:rs-disjointentities} runs, it arbitrarily chooses an order
for the conflicting systems that are ambiguously ordered in the
schedule, and runs them sequentially. In this case,
\VERB|\NormalTok{decrement}| runs entirely before
\VERB|\NormalTok{increment}|. This ordering manifests in transition from
the middle to the final state in the output shown by
Lst.~\ref{lst:rs-disjointentities-out}.\footnote{
    To make it clearer, add start and end print-lines to both systems and observe that they do not interlace.
} In the final state, \texttt{e0} has \(4\) instead of \(3\) because,
during schedule application, \(4\) is decremented to \(3\) \emph{and
then incremented again}, before the final state is printed. We cannot
force Bevy ECS to run the conflicting systems concurrently.
Additionally, to our knowledge we cannot force Bevy ECS to tell us ahead
of time that the ambiguously-ordered conflicting systems \emph{cannot}
be run concurrently.

\begin{codelisting}

\caption{Output produced by
\VERB|\PreprocessorTok{disjoint\_entities::}\NormalTok{demo}| in
Lst.~\ref{lst:rs-disjointentities}.}\label{lst:rs-disjointentities-out}

\begin{Shaded}
\begin{Highlighting}[numbers=left,,]
\NormalTok{Num↦[e0↦Num(3), e1↦Num(8), ]}
\NormalTok{Num↦[e0↦Num(4), e1↦Num(7), ]}
\NormalTok{Num↦[e0↦Num(4), e1↦Num(6), ]}
\end{Highlighting}
\end{Shaded}

\end{codelisting}

In Bevy ECS there are at least three ways to access component data. One
might assume that we could implement this example as specified using
some alternate means.

\begin{enumerate}
\def\labelenumi{\arabic{enumi}.}
\tightlist
\item
  \VERB|\PreprocessorTok{bevy\_ecs::system::}\NormalTok{SystemParam}| is
  the trait for arguments to system functions, subject to the write
  conflict rules we have run up against, above.
\item
  \VERB|\PreprocessorTok{bevy\_ecs::world::World::}\NormalTok{run\_system}|
  is a function for running a system with exclusive world access, but no
  other parallel threads are allowed access while it runs.
\item
  \VERB|\PreprocessorTok{bevy\_ecs::system::}\NormalTok{SystemState}| is
  a mechanism for non-systems to use \VERB|\NormalTok{SystemParam}| data
  to access component data as a system does, but its use assumes
  exclusive world access.
\end{enumerate}

We have confirmed with the authors of Bevy ECS that, unless the
programmer can ``prove to Bevy'' that the entities requested by their
queries are disjoint, there is no way for two systems to run
concurrently if both mutate components of the same type. For
completeness, we show an implementation in
Lst.~\ref{lst:rs-disjointentities2} that does this, however, we
emphasize that this is not the same ECS program. It first classifies
each entity according to which side of the threshold it is on (the
\VERB|\NormalTok{classify}| system on lines 7-15; via Deferred inserts
and Deferred deletes by \Cref{fig:concurrency-slice-and-dice}). Then
simplified \VERB|\NormalTok{increment}| and \VERB|\NormalTok{decrement}|
systems (lines 17-19 and 20-22; still via Owned updates) use this
classification in their queries, allowing Bevy ECS to run them
concurrently. The use of
\VERB|\NormalTok{With}\OperatorTok{\textless{}}\NormalTok{Below}\OperatorTok{\textgreater{}}|
and
\VERB|\NormalTok{Without}\OperatorTok{\textless{}}\NormalTok{Below}\OperatorTok{\textgreater{}}|
on lines 17 and 20 ensure that the entity matches of the two systems are
disjoint. Unfortunately, we were unable to locate a public interface in
Bevy ECS to ask ahead of time whether systems will run
concurrently.\footnote{
    ScheduleGraph and Access structures exist in Bevy and correctly report that these systems access conflicting components, however Bevy uses a FilteredAccess structure to make the final determination of whether systems can run concurrently.
    We were unable to locate an interface by which a populated FilteredAccess structure could be queried.
    In any case, as before, it can be made clear by adding start and end print-lines to the systems and observing that they interlace.
}

\begin{codelisting}

\caption{The fixed implementation of the disjoint entities example with
Bevy ECS.}\label{lst:rs-disjointentities2}

\begin{Shaded}
\begin{Highlighting}[numbers=left,,]
\AttributeTok{\#[}\NormalTok{derive}\AttributeTok{(}\BuiltInTok{Debug}\OperatorTok{,}\NormalTok{ Component}\AttributeTok{)]}
\KeywordTok{struct}\NormalTok{ Below}\OperatorTok{;}

\KeywordTok{mod}\NormalTok{ disjoint\_entities\_fixed }\OperatorTok{\{}
  \KeywordTok{use} \KeywordTok{super}\PreprocessorTok{::}\OperatorTok{*;}

  \KeywordTok{fn}\NormalTok{ classify(pc}\OperatorTok{:}\NormalTok{ ParallelCommands}\OperatorTok{,}\NormalTok{ q}\OperatorTok{:}\NormalTok{ Query}\OperatorTok{\textless{}}\NormalTok{(Entity}\OperatorTok{,} \OperatorTok{\&}\NormalTok{Num}\OperatorTok{,} \DataTypeTok{Option}\OperatorTok{\textless{}\&}\NormalTok{Below}\OperatorTok{\textgreater{}}\NormalTok{)}\OperatorTok{\textgreater{}}\NormalTok{) }\OperatorTok{\{}
\NormalTok{    q}\OperatorTok{.}\NormalTok{par\_iter()}\OperatorTok{.}\NormalTok{for\_each(}\OperatorTok{|}\NormalTok{(e}\OperatorTok{,}\NormalTok{ n}\OperatorTok{,}\NormalTok{ b)}\OperatorTok{|} \ControlFlowTok{match}\NormalTok{ b }\OperatorTok{\{}
      \ConstantTok{None} \ControlFlowTok{if}\NormalTok{ n}\OperatorTok{.}\DecValTok{0} \OperatorTok{\textless{}}\NormalTok{ THRESHOLD }\OperatorTok{=\textgreater{}} \CommentTok{// not flagged "Below" and currently below}
\NormalTok{        pc}\OperatorTok{.}\NormalTok{command\_scope(}\OperatorTok{|}\KeywordTok{mut}\NormalTok{ c}\OperatorTok{|} \OperatorTok{\{}\NormalTok{ c}\OperatorTok{.}\NormalTok{entity(e)}\OperatorTok{.}\NormalTok{insert(Below)}\OperatorTok{;} \OperatorTok{\}}\NormalTok{)}\OperatorTok{,} \CommentTok{// flag}
      \ConstantTok{Some}\NormalTok{(\_) }\ControlFlowTok{if} \OperatorTok{!}\NormalTok{(n}\OperatorTok{.}\DecValTok{0} \OperatorTok{\textless{}}\NormalTok{ THRESHOLD) }\OperatorTok{=\textgreater{}} \CommentTok{// flagged "Below" and currently not below}
\NormalTok{        pc}\OperatorTok{.}\NormalTok{command\_scope(}\OperatorTok{|}\KeywordTok{mut}\NormalTok{ c}\OperatorTok{|} \OperatorTok{\{}\NormalTok{ c}\OperatorTok{.}\NormalTok{entity(e)}\OperatorTok{.}\PreprocessorTok{remove::}\OperatorTok{\textless{}}\NormalTok{Below}\OperatorTok{\textgreater{}}\NormalTok{()}\OperatorTok{;} \OperatorTok{\}}\NormalTok{)}\OperatorTok{,} \CommentTok{// unflag}
\NormalTok{      \_ }\OperatorTok{=\textgreater{}} \OperatorTok{\{\},} \CommentTok{// leave already correct flag}
    \OperatorTok{\}}\NormalTok{)}
  \OperatorTok{\}}

  \KeywordTok{fn}\NormalTok{ increment(}\KeywordTok{mut}\NormalTok{ q}\OperatorTok{:}\NormalTok{ Query}\OperatorTok{\textless{}\&}\KeywordTok{mut}\NormalTok{ Num}\OperatorTok{,}\NormalTok{ With}\OperatorTok{\textless{}}\NormalTok{Below}\OperatorTok{\textgreater{}\textgreater{}}\NormalTok{) }\OperatorTok{\{}
\NormalTok{    q}\OperatorTok{.}\NormalTok{par\_iter\_mut()}\OperatorTok{.}\NormalTok{for\_each(}\OperatorTok{|}\KeywordTok{mut}\NormalTok{ n}\OperatorTok{|}\NormalTok{ n}\OperatorTok{.}\DecValTok{0} \OperatorTok{+=} \DecValTok{1}\NormalTok{ )}
  \OperatorTok{\}}
  \KeywordTok{fn}\NormalTok{ decrement(}\KeywordTok{mut}\NormalTok{ q}\OperatorTok{:}\NormalTok{ Query}\OperatorTok{\textless{}\&}\KeywordTok{mut}\NormalTok{ Num}\OperatorTok{,}\NormalTok{ Without}\OperatorTok{\textless{}}\NormalTok{Below}\OperatorTok{\textgreater{}\textgreater{}}\NormalTok{) }\OperatorTok{\{}
\NormalTok{    q}\OperatorTok{.}\NormalTok{par\_iter\_mut()}\OperatorTok{.}\NormalTok{for\_each(}\OperatorTok{|}\KeywordTok{mut}\NormalTok{ n}\OperatorTok{|}\NormalTok{ n}\OperatorTok{.}\DecValTok{0} \OperatorTok{{-}=} \DecValTok{1}\NormalTok{ )}
  \OperatorTok{\}}

  \KeywordTok{fn}\NormalTok{ print\_all(ns}\OperatorTok{:}\NormalTok{ Query}\OperatorTok{\textless{}}\NormalTok{(Entity}\OperatorTok{,} \OperatorTok{\&}\NormalTok{Num)}\OperatorTok{\textgreater{},}\NormalTok{ bs}\OperatorTok{:}\NormalTok{ Query}\OperatorTok{\textless{}}\NormalTok{(Entity}\OperatorTok{,} \OperatorTok{\&}\NormalTok{Below)}\OperatorTok{\textgreater{}}\NormalTok{) }\OperatorTok{\{}
    \PreprocessorTok{print!}\NormalTok{(     }\StringTok{"Num↦["}\NormalTok{)}\OperatorTok{;} \ControlFlowTok{for}\NormalTok{ (e}\OperatorTok{,}\NormalTok{ n) }\KeywordTok{in} \OperatorTok{\&}\NormalTok{ns }\OperatorTok{\{} \PreprocessorTok{print!}\NormalTok{(}\StringTok{"e\{\}↦\{:?\}, "}\OperatorTok{,}\NormalTok{ e}\OperatorTok{.}\NormalTok{index()}\OperatorTok{,}\NormalTok{ n) }\OperatorTok{\}}
    \PreprocessorTok{print!}\NormalTok{(}\StringTok{"], Below↦["}\NormalTok{)}\OperatorTok{;} \ControlFlowTok{for}\NormalTok{ (e}\OperatorTok{,}\NormalTok{ b) }\KeywordTok{in} \OperatorTok{\&}\NormalTok{bs }\OperatorTok{\{} \PreprocessorTok{print!}\NormalTok{(}\StringTok{"e\{\}↦\{:?\}, "}\OperatorTok{,}\NormalTok{ e}\OperatorTok{.}\NormalTok{index()}\OperatorTok{,}\NormalTok{ b) }\OperatorTok{\}}
    \PreprocessorTok{print!}\NormalTok{(}\StringTok{"]}\SpecialCharTok{\textbackslash{}n}\StringTok{"}\NormalTok{)}\OperatorTok{;}
  \OperatorTok{\}}

  \KeywordTok{pub} \KeywordTok{fn}\NormalTok{ demo() }\OperatorTok{\{}
    \KeywordTok{let} \KeywordTok{mut}\NormalTok{ world }\OperatorTok{=} \PreprocessorTok{World::}\KeywordTok{default}\NormalTok{()}\OperatorTok{;}
\NormalTok{    world}\OperatorTok{.}\NormalTok{spawn(Num(}\DecValTok{3}\NormalTok{))}\OperatorTok{;}
\NormalTok{    world}\OperatorTok{.}\NormalTok{spawn(Num(}\DecValTok{8}\NormalTok{))}\OperatorTok{;}
    \KeywordTok{let} \KeywordTok{mut}\NormalTok{ schedule }\OperatorTok{=} \PreprocessorTok{Schedule::}\KeywordTok{default}\NormalTok{()}\OperatorTok{;}
\NormalTok{    schedule}\OperatorTok{.}\NormalTok{add\_systems((classify}\OperatorTok{,}\NormalTok{ print\_all}\OperatorTok{,}\NormalTok{ (increment}\OperatorTok{,}\NormalTok{ decrement))}\OperatorTok{.}\NormalTok{chain())}\OperatorTok{;}
    \ControlFlowTok{for}\NormalTok{ \_ }\KeywordTok{in} \DecValTok{0}\OperatorTok{..}\DecValTok{3} \OperatorTok{\{}\NormalTok{ schedule}\OperatorTok{.}\NormalTok{run(}\OperatorTok{\&}\KeywordTok{mut}\NormalTok{ world)}\OperatorTok{;} \OperatorTok{\}}
  \OperatorTok{\}}
\OperatorTok{\}} \CommentTok{// end mod disjoint\_entities\_fixed}
\end{Highlighting}
\end{Shaded}

\end{codelisting}

\begin{codelisting}

\caption{Output produced by
\VERB|\PreprocessorTok{disjoint\_entities\_fixed::}\NormalTok{demo}| in
Lst.~\ref{lst:rs-disjointentities2}.}\label{lst:rs-disjointentities2-out}

\begin{Shaded}
\begin{Highlighting}[numbers=left,,]
\NormalTok{Num↦[e1↦Num(8), e0↦Num(3), ], Below↦[e0↦Below, ]}
\NormalTok{Num↦[e1↦Num(7), e0↦Num(4), ], Below↦[]}
\NormalTok{Num↦[e1↦Num(6), e0↦Num(3), ], Below↦[e0↦Below, ]}
\end{Highlighting}
\end{Shaded}

\end{codelisting}

The output in Lst.~\ref{lst:rs-disjointentities2-out} from running
\VERB|\PreprocessorTok{disjoint\_entities\_fixed::}\NormalTok{demo}|
restores the correct behavior specified in the executable model we first
introduced for the disjoint entities example.

From this example it is clear that Bevy ECS allows a narrow form of case
four from \Cref{thm:infl-safe-schedule} to execute concurrently. Our
claims at the end of \Cref{sec:ecs-conc,sec:practical-ecs} still hold,
however, because only \emph{in-place} updates of existing components
execute concurrently --- creating or deleting entities and attaching or
detaching components are still changes that are deferred in a buffer to
be executed sequentially after the system completes. The
\VERB|\NormalTok{Commands}| and (poorly named)
\VERB|\NormalTok{ParallelCommands}| structures in Bevy ECS both
exemplify this. As we have seen, deferring changes may require the
programmer to use temporary local data structures that duplicate the
meaning of those changes (as with the \VERB|\NormalTok{used}| set in the
\VERB|\NormalTok{collide}| system of Lst.~\ref{lst:rs-toyphys}).
Furthermore, in Bevy ECS it is difficult to know whether systems that
are intended to execute concurrently with each other actually
\emph{will} execute that way. Core ECS and its executable model have an
expressivity advantage in this regard, which suggests that practical ECS
frameworks could adopt this simpler interface style as a step toward
taking advantage of the additional concurrency that we have identified
is safe.

\fi 

\end{document}